\documentclass[11pt]{amsart}

\usepackage[T1]{fontenc}
\usepackage[a4paper,top=4cm,bottom=4cm,inner=3cm,outer=3cm]{geometry}
\usepackage[pdftex]{graphicx,graphics} 
\usepackage{amsmath,amssymb,amsfonts,color,hyperref} 
 \usepackage{epsfig,array}

\def\red{\textcolor{red}} 

\usepackage{graphicx,psfrag,mathrsfs}
\usepackage{tikz}  
\usepackage{slashbox}

\hypersetup{
    linktoc=page,
   linkcolor=red,          
  citecolor=blue,        
    filecolor=blue,      
   urlcolor=cyan,
    colorlinks=true           
}

\usepackage[refpage]{nomencl}

\newtheorem{theorem}{Theorem}
\newtheorem{proposition}{Proposition}

\newtheorem{remark}{Remark}
\newtheorem{lemma}{Lemma}
\newtheorem{definition}{Definition}
\newtheorem{corollary}{Corollary}

\newcommand\ep{{\varepsilon}}

\newcommand\zu{[0,1]}

\newcommand\mk{\medskip}
\newcommand\sk{\smallskip}

\newcommand{\R}{\mathbb{R}}
\newcommand{\N}{\mathbb{N}}

\newcommand{{\taut}}{\widetilde{\tau}}

\newcommand \cjeta{\mathcal{S}_j(\eta)}

\newcommand{\wep}{\widetilde \ep}
\newcommand\weta{{\widetilde \eta}}
\newcommand\wtau{\widetilde \tau}
\newcommand\zud{[0,1]^d}
\newcommand{\etal}{\eta_\ell}
\newcommand{\etar}{\eta_r}
\newcommand\lf{\lfloor}
\newcommand\rf{\rfloor}
\newcommand{\Ht}{{\widetilde H}}
\newcommand{\locloc}{{}}

\newcommand\Rmu{\mathcal{R}_\mu}
\newcommand\Tl{{\mathcal{T}_{\mu,\ell}}}
\newcommand\Tr{{\mathcal{T}_{\mu,r}}}
\newcommand\Tt{{\mathcal{T}_\mu}}

\newcommand{\M}{\mathsf{M}}
\begin{document}
\title[Random sparse sampling]{Random sparse sampling \\ in a Gibbs weighted tree}

\author{Julien Barral \and St\'ephane Seuret }         

\address{Julien BARRAL, LAGA, CNRS UMR 7539, Institut Galil\'ee, Universit\'e Paris 13, Sorbonne Paris Cit\'e, 99 avenue Jean-Baptiste Cl\'ement , 93430  Villetaneuse, 
France}
\address{St\'ephane SEURET, Universit\'e Paris-Est, LAMA (UMR 8050), UPEMLV, UPEC, CNRS, F-94010, Cr\'eteil, France}
\date{Received: date / Revised version: date}
\begin{abstract}
Let $\mu$ be the geometric realization on $[0,1]$ of a Gibbs measure on $\Sigma=\{0,1\}^\N$ associated with a H\"older  potential.  The thermodynamic and multifractal properties of $\mu$ are  well known to be linked via  the multifractal formalism.  In this article, the impact of a random  sampling procedure on this structure is studied. 

More precisely, let $\{I_w\}_{w\in \Sigma^*}$ stand for the collection of dyadic subintervals of $[0,1]$ naturally indexed by the set of finite dyadic words $\Sigma^*$. Fix $\eta\in(0,1)$,  and a sequence $(p_w)_{w\in \Sigma^*}$ of independent Bernoulli variables of   parameters $2^{-|w|(1-\eta)}$ ($|w|$ is the length of $w$). 
We consider the (very sparse) remaining values $\widetilde\mu=\{\mu(I_w): w\in \Sigma^*, p_w=1\}$.

We   prove that when $\eta<1/2$, it is possible to entirely reconstruct $\mu$ from the sole knowledge of $\widetilde\mu$, while it is not possible when $\eta>1/2$, hence a first phase transition phenomenon.

We show that, for all $\eta \in (0,1)$, it is possible to reconstruct a large part of the initial multifractal structure of $\mu$, via the fine study of $\widetilde\mu$. After reorganization,  these coefficients give rise to a random capacity with new remarkable scaling and multifractal properties: its $L^q$-spectrum exhibits two phase transitions, and has a   rich   thermodynamic and geometric structure. 
\end{abstract}

\maketitle


\section{Introduction}

Statistical mechanics and multifractals are well known to be closely related. Typical situations are provided by the energy model associated with a Gibbs measure on  the boundary $\Sigma$ of the dyadic tree $\Sigma^*$ in the context of the thermodynamic formalism \cite{Ruelle,Collet,Rand}, or the random energy model associated with a branching random walk on $\Sigma^*$, namely directed polymers on disordered trees \cite{DerSpo,CK,HW,Mol,AB,B2014}. 
The purpose of this paper is to investigate the thermodynamic and geometric impact of a random sparse sampling on such structures. 

Let us start by describing the interplay between thermodynamics and multifractals. 

\subsection{Free energy and singularity spectrum as a Legendre pair}

For the sake of generality, we   work on the $d$-dimensional dyadic tree  and on $[0,1]^d$, $d\ge 1$. Let $\Sigma_j$ be the set of words of length $j\ge 1$ over the alphabet $\{0,1\}^d$, i.e. 
$$\Sigma_j= \Big\{(w_1w_2\cdots w_j): \forall \, k\in \{1,...,j\}, \ w_k=(w_k^{(1)}, w_k^{(2)}, ,...,w_k^{(d)} ) \in \{0,1\}^d \Big \}.$$

If $w\in \Sigma_j$, we denote by $|w|=j$ its length (or its generation). Then, $\Sigma^* = \bigcup_{j\geq 1} \Sigma_j$ and $\Sigma=(\{0,1\}^d)^{\N_+}$ denote the set of finite words and infinite words over $\{0,1\}^d$ respectively. The set $\Sigma $ is endowed with the standard ultra-metric distance, and $\Sigma^*\cup\Sigma$ is endowed with the shift operation denoted $\sigma$.

If $w\in \Sigma^*\cup\Sigma$ and $1\le j\le |w|$ is finite, $w_{|j}$ stands for the prefix of length $j$ of $w$. If  $W\in \Sigma^*$,  $[W]$ is the {\em cylinder} of those words  $w\in\Sigma$ such that $w_{||W|}=W$.  

With each $w=w_1...w_j\in \Sigma_j$ is naturally associated the dyadic point 
\begin{equation}
\label{defxw}
x_w= \left(\sum_{k=1}^j w^{(i)}_k 2^{-k} \right)_{1\le i\le d},
\end{equation}
of $[0,1]^d$, and the dyadic subcube $I_w=\prod_{i=1}^d \left[x^{(i)}_w,x^{(i)}_w+2^{-j} \right]$ of $[0,1]^d$. 
 
 If $x=(x^{(1)}, x^{(2)}, ,...,x^{(d)})\in [0,1]^d$ has no dyadic component, then $x$ is encoded by a unique $w=w^{(1)}w^{(2)}...w^{(d)} \in \Sigma$, and $I_j(x)$ stands for $I_{w_{|j}}$. When $x^{(i)}$ is dyadic, we choose  $w^{(i)}$ as the largest element of $\{0,1\}^{\N_+}$ in lexicographical order which encodes $x^{(i)}$. In both cases, $w_{|j}$  is also denoted $x_{|j}$.
 
\begin{definition}
We call {\em capacity} a non-negative and non-decreasing function $\mu$ of the dyadic subcubes of $[0,1]^d$, i.e. for every $W,w \in \Sigma^*$ such that $I_w\subset I_W$, $0\le \mu(I_w)\leq \mu(I_W)$.

The set of capacities  is denoted by $ \mathrm{Cap}([0,1]^d)$. 

 The support of $\mu \in \mathrm{Cap}([0,1]^d)$  is the set $\displaystyle \mathrm{supp}(\mu)=\bigcap_{j\ge 1}\bigcup_{w\in\Sigma_j: \mu(I_w)>0}I_w$. 
 \end{definition}

We focus on two quantities especially relevant  in the thermodynamic  and  geometric measure theoretic  contexts. 

\mk

$\bullet$ The {\em free energy} of a capacity $\mu\in  \mathrm{Cap}([0,1]^d)$ with a non empty support  is defined as the thermodynamic (lower) limit given   for $q\in \R$ by 
\begin{equation}
\label{deftau}
\tau_\mu(q)=\liminf_{j\to\infty} \tau_{\mu,j}(q), \ \ \ \mbox{ where } \  \tau_{\mu,j}(q) :=\frac{-1}{j}\log_2\sum_{w\in\Sigma_{j}: \mu(I_w)>0}\mu(I_w)^q,
\end{equation}
and $q$ is interpreted as the inverse of a temperature when it is positive (the precise connection with statistical mechanics terminology is that in finite volume $j$, $\tau_{\mu,j}(q)$ is the free energy associated with the  potential $V(w)=-\log(\mu(I_w))$, $w\in\Sigma_j$).   

When the free energy  $\tau_\mu(q)$ is a limit (not only a liminf) and is differentiable, the value $\tau_\mu(q)$ allows one to describe the asymptotical distribution properties of $\mu$ over $\Sigma_j$ thanks to large deviations theory, which roughly gives the approximation:
$$
\forall \, H\in \R, \ \#\left\{w\in\Sigma^*: |w|=j,\, \mu(I_w)\approx 2^{-jH}\right \} \approx 2^{j\tau_\mu^*(H)}\quad\text{as $j\to + \infty$},
$$ 
where $\tau_\mu^*$ is the Legendre transform of $\tau_\mu$, i.e.  
\begin{equation}
\label{legendre}
\tau_\mu^*(H):=\inf_{q\in\R} \big(Hq-\tau_\mu(q)\big).
\end{equation}

$\bullet$   The {\em singularity}, or {\em multifractal}, {\em spectrum} of~$\mu$ is defined as
$$
D_\mu:H\mapsto \dim \underline E_\mu(H), \quad H\in\R,
$$
where 
$$
\underline E_\mu(H)=\left\{x\in \mathrm{supp}(\mu):\liminf_{j\to\infty} \frac{\log_2 \big(\mu(I_{x_{|j}} ) \big)}{-j }=H\right\}.
$$
The Hausdorff dimension in  $\R^d$ is denoted by $\dim$, and by convention, $\dim \emptyset = -\infty$.   The singularity spectrum  provides a fine geometric description of the energy distribution at small scales by giving the Hausdorff dimension of the iso-H\"older sets $\underline E_\mu(H)$ of $\mu$.


\begin{figure}  \begin{tikzpicture}[xscale=0.9,yscale=0.9]
{\tiny
\draw [->] (0,-2.8) -- (0,1.5) [radius=0.006] node [above] {$\tau_\mu(q)$};
\draw [->] (-1.,0) -- (3.7,0) node [right] {$q$};
 \draw [thick, domain=1.7:3.5, color=black]  plot ({\x},  {-ln(exp(\x*ln(1/5)) +exp(\x*ln(0.8)))/(ln(2)});
 \draw [thick, domain=0.6:1.7, color=black]  plot ({\x},  {-ln(exp(\x*ln(1/5)) +exp(\x*ln(0.8)))/(ln(2)});
\draw [thick, domain=-0.8:0.6, color=black]  plot ({\x},  {-ln(exp(\x*ln(1/5)) +exp(\x*ln(0.8)))/(ln(2)});
  \draw[dashed] (0,0) -- (2.8,1);
\draw [fill] (-0.1,-0.20)   node [left] {$0$}; 
\draw [fill] (-0,-1) circle [radius=0.03]  node [left] {$-d$ \ }; 

\draw [fill] (1,-0) circle [radius=0.03] [fill] (1,-0.25) node [ ] {$1 $}; }
\end{tikzpicture} \hskip .5cm
 \begin{tikzpicture}[xscale=1.9,yscale=2.6]
    {\tiny
\draw [->] (0,-0.2) -- (0,1.25) [radius=0.006] node [above] {$D_\mu(H)=\tau_\mu^*(H)$};
\draw [->] (-0.2,0) -- (2.8,0) node [right] {$H$};
\draw [thick, domain=0:5]  plot ({-(exp(\x*ln(1/5))*ln(0.2)+exp(\x*ln(0.8))*ln(0.8))/(ln(2)*(exp(\x*ln(1/5))+exp(\x*ln(0.8)) ) )} , {-\x*( exp(\x*ln(1/5))*ln(0.2)+exp(\x*ln(0.8))*ln(0.8))/(ln(2)*(exp(\x*ln(1/5))+exp(\x*ln(0.8))))+ ln((exp(\x*ln(1/5))+exp(\x*ln(0.8))))/ln(2)});
\draw [thick, domain=0:5]  plot ({-( ln(0.2)+ ln(0.8))/(ln(2)) +(exp(\x*ln(1/5))*ln(0.2)+exp(\x*ln(0.8))*ln(0.8))/(ln(2)*(exp(\x*ln(1/5))+exp(\x*ln(0.8)) ) )} , {-\x*( exp(\x*ln(1/5))*ln(0.2)+exp(\x*ln(0.8))*ln(0.8))/(ln(2)*(exp(\x*ln(1/5))+exp(\x*ln(0.8))))+ ln((exp(\x*ln(1/5))+exp(\x*ln(0.8))))/ln(2)});
  \draw[dashed] (0,1) -- (2.6,1);
\draw [fill] (-0.1,-0.10)   node [left] {$0$}; 
\draw [fill] (0,1) circle [radius=0.03] node [left] {$d \ $}; 
\draw  [fill] (0.32,0) circle [radius=0.03]  (0.32,-0.05)  node [below]{$H_{\min}$};
\draw  [fill] (2.32,0) circle [radius=0.03]   (2.32,-0.05) node [below] {$H_{\max}$};
\draw  [fill] (1.32,1) circle [radius=0.03]  [dashed]   (1.32,1) -- (1.32,0)  [fill] (1.32,0) circle [radius=0.03]  (1.32,-0.05) node [below] {$H_{s}$};
}
\end{tikzpicture}
\caption{{\bf Left:} Free energy function of a Gibbs measure $\mu$ on $[0,1]^d$. {\bf Right:} The singularity spectrum of $\mu$.}
\label{figureLegpair}
\end{figure}
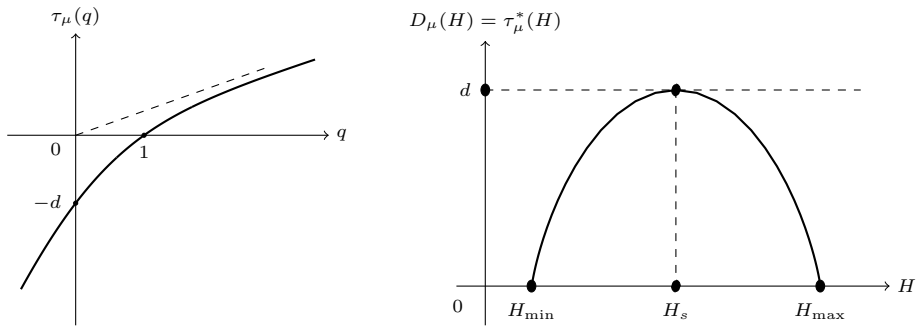


\mk

It turns out that  when $\mu$ possesses nice scaling properties, one has
$$
  \forall\,H\in\R, \quad D_\mu(H)=\tau_\mu^*(H).
$$ 
\begin{definition}
When the above formula is satisfied, $\tau_\mu$ and $D_\mu$ are said to form a Legendre pair (see Figure~\ref{figureLegpair}). In this situation, one says that $\mu$ obeys the multifractal formalism at any $H\in\R$.
\end{definition}

Forming a Legendre pair implies that the geometric description of $\mu$ provided by its singularity spectrum $D_\mu$ matches with the asymptotic statistical description of the energy distribution~$\mu$ provided by the free energy $\tau_\mu$ and its Legendre transform.
Our goal is to  investigate the impact on such well-organized structures of a natural sampling procedure.
   
\subsection{Random sparse sampling operation on capacities}

We perform on any capacity$\mu$ the  random sampling process consisting in acting independently on the vertices of $\Sigma^*$ by letting a vertex of generation $j$ survive with probability $2^{-jd(1-\eta)}$, {where $\eta\in (0,1)$} (it is also a special case of decimation rule used in percolation theory on $\Sigma^*$). More formally:

 \begin{definition}\label{def11}
Fix  a real parameter $0<\eta <1$,  called the {\em sampling index}. Let  $(\Omega, \mathcal{F}, \mathbb{P})$ be   a probability space, 
 and  $(p_w)_{w\in \Sigma^*}$ a sequence of independent Bernoulli random variables so that $p_w\sim B(2^{-d(1-\eta)|w|})$, i.e.
 \begin{equation}
 \label{defeta}
 \mathbb{P}(p_w=1)=1-\mathbb{P}(p_w=0) = 2^{-d(1-\eta)|w|}.
 \end{equation}

When $p_w=1$, $w$ is said to be a surviving vertex (or a survivor).

 For every $j\geq 1$, denote by $\mathcal{S}_j(\eta)$   the (random) set of  surviving vertices   in $ \Sigma_j$:
    $$ \mathcal{S}_j(\eta):= \big\{ w\in \Sigma_j:  p_w =1\}.$$
 
\smallskip

 Let $\mu \in \mathrm{Cap}([0,1]^d)$.  We denote by $\widetilde \mu: \Sigma^* \to \R^+$ the function defined by
 $$\forall \, w\in \Sigma^*,  \ \ \ \widetilde\mu(I_w) = \mu(I_w) \cdot p_w.$$
  \end{definition} 
  
 See Figure  \ref{figure1}   for an illustration.
The problems we address are the following. 
\begin{itemize}

\smallskip
\item
{\bf Recovering   from sparse information:} The set of surviving vertices  $\mathcal{S}_j(\eta)$ has a cardinality of expectation $2^{dj \eta}$ (which is exponentially less than the $2^{dj}$ initial coefficients), and is very sparse.  The first question concerns the   information remaining after the sampling operation.  Can one recover the initial Gibbs measure $\mu$ (i.e. all the values $(\mu(I_w))_{w\in \Sigma^*}$) from the sole knowledge of    $\widetilde \mu$?  If not, what about recovering the free energy and multifractal spectrum of $\mu$? 

\smallskip
\item
{\bf Structure of $\widetilde \mu$:}  The new object $\widetilde \mu$ is not a capacity any more. Does it have a well-organized structure though?  
 \end{itemize}

\begin{figure}
\tikzstyle{lien}=[->,>=stealth,rounded corners=5pt,thick]
\tikzset{individu/.style={draw,thick,fill=#1!25},
individu/.default={green}}

\begin{tikzpicture}[ xscale=0.58,yscale=0.63]
{\tiny
\node   {$\mu(I_\emptyset)$} [sibling distance=6cm]
child { node {$\mu(I_0)$}    [sibling distance=3cm]
child { node {$\mu(I_{00})$} [sibling distance=1.5cm]
child { node  {$\mu(I_{000})$}  [sibling distance=0.8cm]
child { node {:} } child {node {:} }}
child { node  {$\mu(I_{001})$}  [sibling distance=0.8cm]
child { node {:} } child {node {:} }}  } 
child { node {$\mu(I_{01})$} [sibling distance=1.5cm]
child { node  {$\mu(I_{010})$}  [sibling distance=0.8cm]
child { node {:} } child {node {:} }}
child { node  {$\mu(I_{011})$}  [sibling distance=0.8cm]
child { node {:} } child {node {:} }}  
}}
child { node {$\mu(I_1)$}    [sibling distance=3cm]
child { node {$\mu(I_{10})$} [sibling distance=1.5cm]
child { node  {$\mu(I_{100})$}  [sibling distance=0.8cm]
child { node {:} } child {node {:} }}
child { node  {$\mu(I_{101})$}  [sibling distance=0.8cm]
child { node {:} } child {node {:} }}  } 
child { node {$\mu(I_{11})$} [sibling distance=1.5cm]
child { node  {$\mu(I_{110})$}  [sibling distance=0.8cm]
child { node {:} } child {node {:} }}
child { node  {$\mu(I_{111})$}  [sibling distance=0.8cm]
child { node {:} } child {node {:} }}  
}};
}
\end{tikzpicture}
  \begin{tikzpicture}[ xscale=0.54,yscale=0.63]
{\tiny
\node   {$\mu(I_\emptyset)$} [sibling distance=6cm]
child { node {$\textcolor{red}{0}$}    [sibling distance=3cm]
child { node {$\textcolor{red}{0}$} [sibling distance=1.5cm]
child { node  {$\textcolor{red}{0}$}  [sibling distance=0.8cm]
child { node {:} } child {node {:} }}
child { node  {$\mu(I_{001})$}  [sibling distance=0.8cm]
child { node {:} } child {node {:} }}  } 
child { node {$\mu(I_{01})$} [sibling distance=1.5cm]
child { node  {$\textcolor{red}{0}$}  [sibling distance=0.8cm]
child { node {:} } child {node {:} }}
child { node  {$\textcolor{red}{0}$}  [sibling distance=0.8cm]
child { node {:} } child {node {:} }}  
}}
child { node {$\mu(I_1)$}    [sibling distance=3cm]
child { node {$\mu(I_{10})$} [sibling distance=1.5cm]
child { node  {$\textcolor{red}{0}$}  [sibling distance=0.8cm]
child { node {:} } child {node {:} }}
child { node  {$\textcolor{red}{0}$}  [sibling distance=0.8cm]
child { node {:} } child {node {:} }}  } 
child { node {$\textcolor{red}{0}$} [sibling distance=1.5cm]
child { node  {$\textcolor{red}{0}$}  [sibling distance=0.8cm]
child { node {:} } child {node {:} }}
child { node  {$\mu(I_{111})$}  [sibling distance=0.8cm]
child { node {:} } child {node {:} }}  
}};
}
\end{tikzpicture} 
\caption{{\bf Left:} Capacity $\mu$ on dyadic cubess. {\bf Right:} Function $\widetilde \mu$ and surviving vertices after sampling.}\label{figure1}
   \end{figure}
%

 The last two above questions are of course related to each other.   
 
 \sk
 
 Concerning the reconstruction problematics, recovering the scaling behavior from sparse information is a very natural issue in signal processing (this is one issue in compressive sensing). This allows one to evaluate the ``incompressible'' information represented by the initial capacity. We   bring an answer when  $\mu$ is the geometric realization on $[0,1]$ of a Gibbs measure associated with a H\"older continuous potential on $\Sigma$, and more generally a non trivial Gibbs capacity. Specifically, the capacity $\mu$ satisfies   that there exists a Gibbs measure $\nu$, $K>0$ and~$(\alpha,\beta)\in \R_+\times \R_+$ such that  
$$
\mu(I_w)=K\nu(I_w)^\alpha 2^{-\beta |w|},\quad \forall \, w\in\Sigma^*,
$$
and $\mu$ is not constant, so $(\alpha,\beta)\neq (0,0)$ (see Section~\ref{GibbsConstruction} for a precise definition of Gibbs measures and capacities).  

 \begin{theorem} 
 \label{th_intro_reco}  Suppose that $\mu$ is a Gibbs capacity. With probability one, when $\eta<1/2$, one can reconstruct, up to some multiplicative constant depending only on $\mu$, all the values $\{\mu(I_w):w\in \Sigma^*\}$, while when $\eta>1/2$, it is impossible  provided $\mu$ is not built from a potential on $\Sigma$ depending on only finitely many letters.
 \end{theorem}
 
See Section~\ref{recovering}, Theorem \ref{th_reconstruction} for a more precise statement.  This constitutes a first phase transition phenomenon at $\eta=1/2$.

\sk

 Regarding recovering of statistical and geometrical properties of $\mu$,  we first  reorganize the surviving information in a suitable and exploitable way, as follows. If $w,v\in \Sigma^*$, $wv$ stands for the   concatenation of $w$ and $v$.
 
 \begin{definition} 
  Let $\mu \in \mathrm{Cap}([0,1]^d)$. We consider the  random capacity $\M_\mu\in \mathrm{Cap}([0,1]^d)$ associated with $\mu$ and the sequence  $(p_w)_{w\in \Sigma^*}$ defined by
 \begin{equation}
 \label{defwmu}
 \M_\mu(I_w) = \max\Big\{\mu(I_{wv}): v\in \Sigma^* \ \mbox{ and } p_{wv} =1 \Big\} . 
 \end{equation}
 \end{definition}

See Figure  \ref{figure2} for the construction of $\M_\mu$.
By construction, any capacity $\mu  \in \mathrm{Cap}([0,1]^d)$ satisfies $\mu(I_w) = \max\Big\{\mu(I_{wv}): v\in \Sigma^*\Big\} , $ hence  \eqref{defwmu} is the most natural  formula to be used to build a capacity from $\widetilde \mu$.

It is not difficult to see that with probability 1, for every $w\in \Sigma^*$, the set $\big \{v\in \Sigma^* \ \mbox{ and } p_{wv} =1\big\}$ is non-empty, so that $\M_\mu$ is well defined. Observe that by our choice \eqref{defeta}, most of the coefficients $\widetilde\mu(I_w)$  equal 0, hence typically one has $\M_\mu(I_w) <\!\!\!<  \mu(I_w) $ when $\lim_{j\to+\infty}\max\{\mu(I_w):w\in \Sigma_j\}=0$.

The definition of $\M_\mu$ can be rephrased as
$$ \M_\mu(I_w) = \max \left \{\mu(I_v): v\text{ survives, $[v]\subset[w]$}\right\} =  \max\left\{\widetilde\mu(I_{wv}): v\in \Sigma^*   \right\}.$$

 We notice that $\M_\mu$ and $\widetilde \mu$ are equivalent objects in the following sense. If $\mu$ is strictly positive,  $\widetilde\mu$ can be recovered from $\M_\mu$  since $\widetilde \mu(I_w) \neq 0$ if and only if $\M_\mu(I_w)>\M_\mu(I_{wv})$ for all $v\in\Sigma^*$ such that $|v|\ge 1$.  From now on, we  work with the capacitiy $\M_\mu$ only.

\smallskip

Starting from a positive capacity $\mu$ whose free energy $\tau_\mu$ and singularity spectrum $D_\mu$ form a Legendre pair, we consider the following questions in order to estimate the structural perturbations induced by the sampling process: 

\begin{itemize}
\item
Do the free energies in finite volume $\tau_{\M_\mu,j}$ converge to a thermodynamic limit $\tau_{\M_\mu}$ as $j\to+\infty$? 
\item
Is it possible to conduct a fine   analysis of the local behavior of $\M_\mu$ so that $D_{\M_\mu}$ is computable? If so, do $\tau_{\M_\mu}$ and $D_{\M_\mu}$ form a Legendre pair?
\item
Are there explicit relations between the new pair $(\tau_{\M_\mu},D_{\M_\mu})$ and the original one~$(\tau_\mu,D_\mu)$, so that one can recover the  initial information (before sampling)?
\end{itemize}


\begin{figure}
  \begin{tikzpicture}[ xscale=0.54,yscale=0.63]
{\tiny
\node   {$\mu(I_\emptyset)$} [sibling distance=6cm]
child { node {$\textcolor{red}{0}$}    [sibling distance=3cm]
child { node {$\textcolor{red}{0}$} [sibling distance=1.5cm]
child { node  {$\textcolor{red}{0}$}  [sibling distance=0.8cm]
child { node {:} } child {node {:} }}
child { node  {$\mu(I_{001})$}  [sibling distance=0.8cm]
child { node {:} } child {node {:} }}  } 
child { node {$\mu(I_{01})$} [sibling distance=1.5cm]
child { node  {$\textcolor{red}{0}$}  [sibling distance=0.8cm]
child { node {:} } child {node {:} }}
child { node  {$\textcolor{red}{0}$}  [sibling distance=0.8cm]
child { node {:} } child {node {:} }}  
}}
child { node {$\mu(I_1)$}    [sibling distance=3cm]
child { node {$\mu(I_{10})$} [sibling distance=1.5cm]
child { node  {$\textcolor{red}{0}$}  [sibling distance=0.8cm]
child { node {:} } child {node {:} }}
child { node  {$\textcolor{red}{0}$}  [sibling distance=0.8cm]
child { node {:} } child {node {:} }}  } 
child { node {$\textcolor{red}{0}$} [sibling distance=1.5cm]
child { node  {$\textcolor{red}{0}$}  [sibling distance=0.8cm]
child { node {:} } child {node {:} }}
child { node  {$\mu(I_{111})$}  [sibling distance=0.8cm]
child { node {:} } child {node {:} }}  
}};
\draw [thick] (3.0,-3.0) -- (4.2,-3.0)  [thick] (4.8,-3.0) -- (6.1,-3.0);
\draw [thick] (3.0,-3.0) -- (3.0,-6.5)  [thick] (6.1,-3.0) -- (6.1,-6.5);
\draw [thick] (-3.4,-1.5) -- (-6.0,-1.5)  [thick] (-2.7,-1.5) -- (-0.1,-1.5);
\draw [thick] (-6.0,-1.5) -- (-6.0,-6.5)  [thick] (-0.1,-1.5) -- (-0.1,-6.5);
}
\end{tikzpicture}
  \begin{tikzpicture}[ xscale=0.54,yscale=0.63]
{\tiny
\node   {$\mu(I_\emptyset)$} [sibling distance=6cm]
child { node {$\textcolor{red}{\mu(I_{001})}$}   [sibling distance=3cm]
child { node {$\textcolor{red}{\mu(I_{001})}$} [sibling distance=1.5cm]
child { node  {$\textcolor{red}{\mu(I_{...})}$}  [sibling distance=0.8cm]
child { node {:} } child {node {:} }}
child { node  {$\mu(I_{001})$}  [sibling distance=0.8cm]
child { node {:} } child {node {:} }}  } 
child { node {$\mu(I_{01})$} [sibling distance=1.5cm]
child { node  {$\textcolor{red}{\mu(I_{...})}$}  [sibling distance=0.8cm]
child { node {:} } child {node {:} }}
child { node  {$\textcolor{red}{\mu(I_{...})}$}  [sibling distance=0.8cm]
child { node {:} } child {node {:} }}  }}
child { node {$\mu(I_1)$}    [sibling distance=3cm]
child { node {$\mu(I_{10})$} [sibling distance=1.5cm]
child { node  {$\textcolor{red}{\mu(I_{...})}$}  [sibling distance=0.8cm]
child { node {:} } child {node {:} }}
child { node  {$\textcolor{red}{\mu(I_{...})}$}  [sibling distance=0.8cm]
child { node {:} } child {node {:} }}  } 
child { node {$\textcolor{red}{\mu(I_{111})}$} [sibling distance=1.5cm]
child { node  {$\textcolor{red}{\mu(I_{...})}$}  [sibling distance=0.8cm]
child { node {:} } child {node {:} }}
child { node  {$\mu(I_{111})$}  [sibling distance=0.8cm]
child { node {:} } child {node {:} }}  
}};
}
\end{tikzpicture}
\caption{{\bf Left:} surviving vertices after  sampling, and the coefficients used to compute  $\M_\mu(I_{0})$ and $ \M_\mu(I_{11})$.    {\bf Right:} The  capacity $\M_\mu$.}\label{figure2}
   \end{figure}
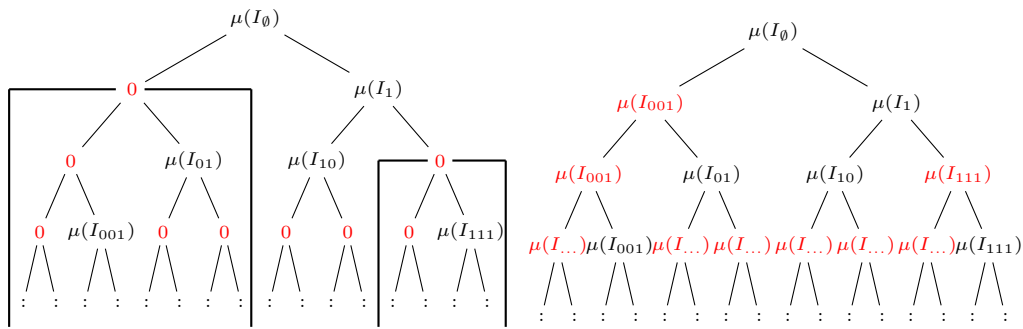

When $\mu$ is a Gibbs Capacity, we are going to  prove that the free energy $\tau_{\M_\mu}$ exists as a limit, and that it forms a Legendre pair with the singularity spectrum of $\M_\mu$. Nevertheless, we will see that the sampling   deeply modifies and complexifies  the  initial structure,  creating several phenomenological differences between $\mu$ and $\M_\mu$, both from thermodynamic and geometric viewpoints.    
  

\subsection{Statement of the main result for the random capacity $\M_{\mu}$ when $\mu$ is Gibbsian}

We  only consider capacities with full support, i.e.   $\mu(I_w)>0$ for all $w\in\Sigma^* $. 

\begin{definition}
Let  $\mu  \in \mathcal{C}([0,1]^d) $ with full support.  
For $x\in \zud$, the lower and upper  local dimensions of $\mu$ at $x$ are respectively defined as 
$$
{\underline \dim_\locloc}(\mu,x)=\liminf_{r \to 0^+}\frac{\log_2 \mu(B(x,r))}{ \log r}  \ \text{ and } \  {\overline \dim_\locloc}(\mu,x)=\limsup_{r\to 0^+}\frac{\log_2 \mu(B(x,r))}{ \log r }.
$$

When ${\underline \dim_\locloc}(\mu,x) = {\overline \dim_\locloc}(\mu,x)$,  their common value is denoted by $ {\dim_\locloc}(\mu,x)$.

For $H\in\R$, set
\begin{equation*}
\begin{split}
\underline E_\mu(H)&=\left \{x\in[0,1]^d: {\underline \dim_\locloc}(\mu,x)=H\right \},\\
 \overline E_\mu(H)&=\left \{x\in[0,1]^d: {\overline \dim_\locloc}(\mu,x)=H\right \},\\
  E_\mu(h)&= \underline E_\mu(H)\cap   \overline E_\mu(H).
  \end{split}
 \end{equation*}

\end{definition}

Recall that  the singularity   spectrum of $\mu$ is the mapping 
$$
D_\mu : \ H\in\R\longmapsto \dim \underline E_\mu(H).
$$

 The lower local dimension is distinguished with respect to $ {\overline \dim_\locloc}(\mu,x) $ or ${\dim_\locloc}(\mu,x)$, because it provides at any $x$ the best local    control   of the capacity  $\mu$.  Since $\mu$ is bounded, one has ${\dim_\locloc}(\mu,x)\ge 0$ at any $x$, hence $\underline E_\mu(H)=\emptyset=\overline E_\mu(H)$ for all $H<0$. 

\medskip

The multifractal formalism states that for every  capacity $\mu \in \mathcal{C}(\zu^d)$  with full support,
\begin{equation}
\label{formalism}
\dim \underline E_\mu(H)\le \tau_\mu^*(H):=\inf _{ q\in \R}\big (Hq-\tau_\mu(q)  \big ), \quad \forall\, H\in\R,
\end{equation}
see for instance \cite{BRMICHPEY,OLSEN}, which deal with measures, but easily extend to capacities. Recall that the multifractal formalism holds for $\mu$ at $H\in\R$ when  there is equality in \eqref{formalism}.   
\medskip

\begin{figure}
 
     \begin{tikzpicture}[xscale=1.8,yscale=2.4]{\tiny
\draw [->] (0,-0.2) -- (0,1.2) [radius=0.006] node [above] {$D_\mu(H)$};
\draw [->] (-0.2,0) -- (2.8,0) node [right] {$H$};
\draw [thick,  domain=1.9:5]  plot ({-(exp(\x*ln(1/5))*ln(0.2)+exp(\x*ln(0.8))*ln(0.8))/(ln(2)*(exp(\x*ln(1/5))+exp(\x*ln(0.8)) ) )} , {-\x*( exp(\x*ln(1/5))*ln(0.2)+exp(\x*ln(0.8))*ln(0.8))/(ln(2)*(exp(\x*ln(1/5))+exp(\x*ln(0.8))))+ ln((exp(\x*ln(1/5))+exp(\x*ln(0.8))))/ln(2)});
\draw [thick, domain=0:1.9]  plot ({-(exp(\x*ln(1/5))*ln(0.2)+exp(\x*ln(0.8))*ln(0.8))/(ln(2)*(exp(\x*ln(1/5))+exp(\x*ln(0.8)) ) )} , {-\x*( exp(\x*ln(1/5))*ln(0.2)+exp(\x*ln(0.8))*ln(0.8))/(ln(2)*(exp(\x*ln(1/5))+exp(\x*ln(0.8))))+ ln((exp(\x*ln(1/5))+exp(\x*ln(0.8))))/ln(2)});
\draw [thick, domain=1.9:5]  plot ({-( ln(0.2)+ ln(0.8))/(ln(2)) +(exp(\x*ln(1/5))*ln(0.2)+exp(\x*ln(0.8))*ln(0.8))/(ln(2)*(exp(\x*ln(1/5))+exp(\x*ln(0.8)) ) )} , {-\x*( exp(\x*ln(1/5))*ln(0.2)+exp(\x*ln(0.8))*ln(0.8))/(ln(2)*(exp(\x*ln(1/5))+exp(\x*ln(0.8))))+ ln((exp(\x*ln(1/5))+exp(\x*ln(0.8))))/ln(2)});
\draw [thick, domain=0:1.9]  plot ({-( ln(0.2)+ ln(0.8))/(ln(2)) +(exp(\x*ln(1/5))*ln(0.2)+exp(\x*ln(0.8))*ln(0.8))/(ln(2)*(exp(\x*ln(1/5))+exp(\x*ln(0.8)) ) )} , {-\x*( exp(\x*ln(1/5))*ln(0.2)+exp(\x*ln(0.8))*ln(0.8))/(ln(2)*(exp(\x*ln(1/5))+exp(\x*ln(0.8))))+ ln((exp(\x*ln(1/5))+exp(\x*ln(0.8))))/ln(2)});
  \draw[dashed] (0,1) -- (2.6,1);
 \draw [fill]  (-0,0.33)  circle  [radius=0.03] node [left] {$d(1-\eta)\ $}[dashed] (0,0.33) -- (2.6,0.33);
  \draw [fill]  (0,0.33)  circle  [radius=0.002]  ;
 \draw [fill]  (0.43,0.33)  circle  [radius=0.03]  ;
\draw[dashed] (0.43,0.33) -- (0.43,-0.3)  node [below]  {$H_\ell(\etal)$} (0.43,-0) [fill] circle [radius=0.03]  ;

\draw[fill] (0.93,0.89) circle [radius=0.03] [dashed] (0.93,0.89) -- (0.93,-0.0) [fill] circle [radius=0.03]  [dashed] (0.93,0.0) -- (0.93,-0.1)  node [below]  {$H_\ell(\weta)$};

 \draw [fill] (-0.1,-0.10)   node [left] {$0$}; 
 \draw [fill] (0,1) circle [radius=0.03] node [left] {$d \ $}; 
 \draw [thick] (0,0.33) -- (0.93,0.89);
 }
\end{tikzpicture}
 \    \begin{tikzpicture}[xscale=1.8,yscale=2.4]{\tiny
\draw [->] (0,-0.2) -- (0,1.2) [radius=0.006] node [above] {$D_\mu(H)$};
\draw [->] (-0.2,0) -- (2.8,0) node [right] {$H$};
\draw [thick, domain=-0.5:2]  plot ({-(2*exp(\x*ln(4/9))*ln(4/9)+exp(\x*ln(1/9))*ln(1/9))/(ln(3)*(2*exp(\x*ln(4/9))+exp(\x*ln(1/9)) ) )-0.06} , {\x*(-2* exp(\x*ln(4/9))*ln(4/9)-exp(\x*ln(1/9))*ln(1/9))/(ln(3)*(2*exp(\x*ln(4/9))+exp(\x*ln(1/9))))+ ln(2*(exp(\x*ln(4/9))+exp(\x*ln(1/9))))/ln(3)-0.31});

\draw [thick, domain=0:5]  plot ({-( ln(0.2)+ ln(0.8))/(ln(2)) +(exp(\x*ln(1/5))*ln(0.2)+exp(\x*ln(0.8))*ln(0.8))/(ln(2)*(exp(\x*ln(1/5))+exp(\x*ln(0.8)) ) )} , {-\x*( exp(\x*ln(1/5))*ln(0.2)+exp(\x*ln(0.8))*ln(0.8))/(ln(2)*(exp(\x*ln(1/5))+exp(\x*ln(0.8))))+ ln((exp(\x*ln(1/5))+exp(\x*ln(0.8))))/ln(2)});
  \draw[dashed] (0,1) -- (2.6,1);
  \draw[dashed] (0,1) -- (2.6,1);
 \draw [fill]  (-0,0.33)  circle  [radius=0.03] node [left] {$d(1-\eta)\ $}[dashed] (0,0.33) -- (2.6,0.33);
  \draw [fill]  (0,0.33)  circle  [radius=0.002]  ;
 
\draw  [fill] (0.72,0.45) circle [radius=0.03]  [dashed]   (0.72,0.45) -- (0.72,-0.4) (0.72,0)  [fill] circle [radius=0.03]  (0.72,-0.3) node [below] {$H_\ell(\etal)$};
\draw[fill] (1.05,0.92) circle [radius=0.03] [dashed] (1.05,0.92) -- (1.05,-0.0) [fill] circle [radius=0.03]  [dashed] (1.05,0.0) -- (1.05,-0.1)  node [below]  {$H_\ell(\weta)$};
 \draw [fill] (0,1) circle [radius=0.03] node [left] {$d \ $}; 
 \draw [thick] (0,0.33) -- (1.05,0.92);
 }
\end{tikzpicture}
 
\caption{Values of $H_\ell(\etal)$  and $H_\ell(\weta)$ depending on $D_\mu$ and $\eta$: {\bf Left:} when $D_\mu(H_{\min}) \leq d(1-\eta)$.  {\bf Right:} when $D_\mu(H_{\min}) >d(1-\eta)$.}
\label{fighmin}
\end{figure}
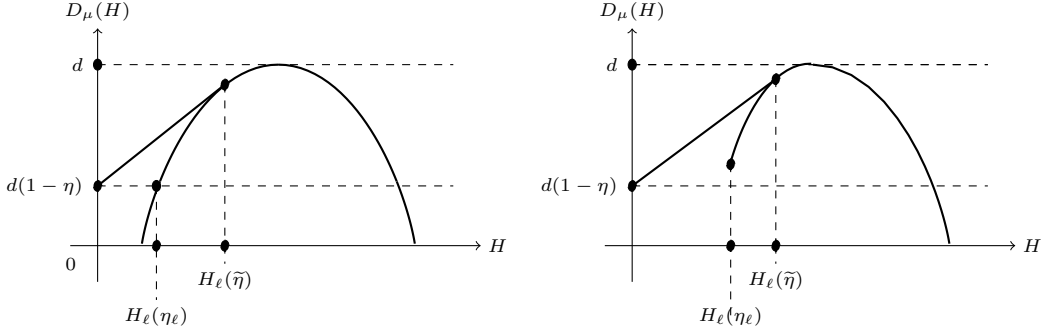


We consider a non-homogeneous Gibbs capacity  $\mu$, i.e. associated with a H\"older potential non cohomologous to a constant   (see  {Definition~\ref{Def 8} in} Section~\ref{GibbsConstruction} for a precise description). For such an object, the following statement gathers information deduced from the study of Gibbs measures and almost-additive potentials \cite{Ruelle,Collet,Rand,BRMICHPEY,H1,Feng-Lau,Fengmatrix}:
Let ${H_{\min}}=\tau_\mu'(+\infty) \, \leq \, H_s := \tau_\mu'(0) \, \leq  \, {H_{\max}}= \tau_\mu'(-\infty)$.
\begin{enumerate}
\item The   free energy  function $\tau_\mu$ is the limit of $(\tau_{\mu,j})_{j\ge 1}$ as $j\to+\infty$. The function $\tau_\mu$  is analytic, increasing,  and  strictly concave on $\R$.

\smallskip

\item The stricly concave function $\tau_\mu^*$ is non-negative on its domain of definition, namely $[H_{\min},H_{\max}]\subset \R^*_+$, and analytic on  $(H_{\min},H_{\max})$. It reaches its maximum at $H_s $, and $\tau_\mu^*(H_s)=d$.

\smallskip

\item For all $H\ge 0$, we have  $D_\mu(H)=\dim  E_\mu(H)=\dim \overline E_\mu(H)=\tau_\mu^*(H).$
\noindent
The multifractal formalism holds for $\mu$, and $(\tau_\mu,D_\mu)$ forms a Legendre pair.
\end{enumerate}


Let us describe our result  on the  random capacity ${\M_\mu}$ obtained after the sampling of $\mu$.  

For this, let us introduce some notations.
 \begin{definition} Let $\mu$  be a non-{homogeneous} Gibbs capacity.   {Given $\eta\in (0,1)$}, one introduces three exponents $H_\ell(\etal) $,  $ H_\ell(\weta) $ and  $\widetilde H_\ell(\weta) $, which depend on $\mu$ and $\eta$ only, by the following formulas:
\begin{itemize}
\item
$H_\ell(\etal) $ is defined as
\begin{equation*}
\label{defHmin}
H_\ell(\etal) =\min\{H\geq 0: D_\mu(H) \geq d(1-\eta)\}.  \text{ Then we set $q_{\eta_\ell}=D_\mu'(H_\ell(\etal))$}.
\end{equation*}

 \item
  $H_\ell(\weta)$  is  the (unique) real number such that the tangent to the graph of  $D_\mu$ at $(H_\ell(\weta),D_\mu(H_\ell(\weta))$ passes through $(0,d(1-\eta))$. Also let $q_\weta= D_\mu'(H_\ell(\weta))$.

\sk\item
Finally,   $\widetilde H_\ell(\weta)    =- \frac{\tau_\mu(q_\weta)}{q_\weta} $.
\end{itemize}
\end{definition}

See Figure \ref{fighmin} for an illustration.  {The origin and roles of the three exponents $H_\ell(\etal) $, $H_\ell(\weta)$ and $\widetilde H_\ell(\weta) $, as well as the notations themselves, will be explained  in Sections 4 and next.}
Observe that {these exponents} depend continuously on $D_\mu$ and $\eta$. 
\begin{figure}
    \begin{tikzpicture}[xscale=1.8,yscale=2.4]{\tiny
\draw [->] (0,-0.2) -- (0,1.2) [radius=0.006] node [above] {$D_\mu(H)$ };
\draw [->] (-0.2,0) -- (2.7,0) node [right] {$H$};
\draw [thick,  domain=1.9:5, color=purple]  plot ({-(exp(\x*ln(1/5))*ln(0.2)+exp(\x*ln(0.8))*ln(0.8))/(ln(2)*(exp(\x*ln(1/5))+exp(\x*ln(0.8)) ) )} , {-\x*( exp(\x*ln(1/5))*ln(0.2)+exp(\x*ln(0.8))*ln(0.8))/(ln(2)*(exp(\x*ln(1/5))+exp(\x*ln(0.8))))+ ln((exp(\x*ln(1/5))+exp(\x*ln(0.8))))/ln(2)});
\draw [thick, domain=0:0.6,color=red]  plot ({-(exp(\x*ln(1/5))*ln(0.2)+exp(\x*ln(0.8))*ln(0.8))/(ln(2)*(exp(\x*ln(1/5))+exp(\x*ln(0.8)) ) )} , {-\x*( exp(\x*ln(1/5))*ln(0.2)+exp(\x*ln(0.8))*ln(0.8))/(ln(2)*(exp(\x*ln(1/5))+exp(\x*ln(0.8))))+ ln((exp(\x*ln(1/5))+exp(\x*ln(0.8))))/ln(2)});
\draw [thick, domain=0.6:1.9,color=black]  plot ({-(exp(\x*ln(1/5))*ln(0.2)+exp(\x*ln(0.8))*ln(0.8))/(ln(2)*(exp(\x*ln(1/5))+exp(\x*ln(0.8)) ) )} , {-\x*( exp(\x*ln(1/5))*ln(0.2)+exp(\x*ln(0.8))*ln(0.8))/(ln(2)*(exp(\x*ln(1/5))+exp(\x*ln(0.8))))+ ln((exp(\x*ln(1/5))+exp(\x*ln(0.8))))/ln(2)});
 \draw [thick, domain=0:6,color=red]  plot ({-( ln(0.2)+ ln(0.8))/(ln(2)) +(exp(\x*ln(1/5))*ln(0.2)+exp(\x*ln(0.8))*ln(0.8))/(ln(2)*(exp(\x*ln(1/5))+exp(\x*ln(0.8)) ) )} , {-\x*( exp(\x*ln(1/5))*ln(0.2)+exp(\x*ln(0.8))*ln(0.8))/(ln(2)*(exp(\x*ln(1/5))+exp(\x*ln(0.8))))+ ln((exp(\x*ln(1/5))+exp(\x*ln(0.8))))/ln(2)});
  \draw[dashed] (0,1) -- (2.6,1);
  \draw[dashed] (0,0.33) -- (0.45,0.33) ; 
 \draw [fill]  (-0,0.33)  circle  [radius=0.03] node [left] {$d(1-\eta)\ $};
 \draw [fill]  (0,0.33)  circle  [radius=0.002]  ;
 \draw [fill]  (0.45,0.33)  circle  [radius=0.03]  ;
\draw[dashed] (0.45,0.33) -- (0.45,-0.3 )  node [below] {${H_\ell(\etal)}$}  (0.45,0) [fill] circle [radius=0.03] ; 
 \draw[fill] (0.93,0.89) circle [radius=0.03] [dashed] (0.93,0.89) -- (0.93,-0.0) [fill] circle [radius=0.03]  [dashed] (0.93,0.0) -- (0.93,-0.1)  node [below]  {$H_\ell(\weta)$};
 \draw [fill] (2.32,0.0)  circle [radius=0.03]  [dashed] (2.32,0.0) -- (2.32,-0.3)  node [below]  {$H_{\max}$};
  \draw [fill] (-0.1,-0.10)   node [left] {$0$}; 
 \draw [fill] (0,1) circle [radius=0.03] node [left] {$d \ $}; 
  \draw[fill,color=white] (1.41,0.89) circle [radius=0.03] [dashed] (1.41,0.89) -- (1.41,-0.0) [fill] circle [radius=0.03]  [dashed] (1.41,0.0) -- (1.41,-0.3)  node [below]  {$H_\ell(\weta)+{\Ht}_\ell(\weta)$};
}
\end{tikzpicture}    \       \begin{tikzpicture}[xscale=1.8,yscale=2.4]{\tiny
   \draw [->] (0,-0.2) -- (0,1.2) [radius=0.006] node [above] {$D_{\M_\mu}(H)$ };
\draw [->] (-0.2,0) -- (3.2,0) node [right] {$H$};
\draw [thick, domain=0.6:1.9,color=black]  plot ({-(exp(\x*ln(1/5))*ln(0.2)+exp(\x*ln(0.8))*ln(0.8))/(ln(2)*(exp(\x*ln(1/5))+exp(\x*ln(0.8)) ) )} , {(-0.33)-\x*( exp(\x*ln(1/5))*ln(0.2)+exp(\x*ln(0.8))*ln(0.8))/(ln(2)*(exp(\x*ln(1/5))+exp(\x*ln(0.8))))+ ln((exp(\x*ln(1/5))+exp(\x*ln(0.8))))/ln(2)});
\draw [thick, domain=0:0.6, color=red]  plot ({(0.47)-(exp(\x*ln(1/5))*ln(0.2)+exp(\x*ln(0.8))*ln(0.8))/(ln(2)*(exp(\x*ln(1/5))+exp(\x*ln(0.8)) ) )} , {-\x*( exp(\x*ln(1/5))*ln(0.2)+exp(\x*ln(0.8))*ln(0.8))/(ln(2)*(exp(\x*ln(1/5))+exp(\x*ln(0.8))))+ ln((exp(\x*ln(1/5))+exp(\x*ln(0.8))))/ln(2)});
 \draw [thick, domain=0:6,color=red]  plot ({(0.47)-( ln(0.2)+ ln(0.8))/(ln(2)) +(exp(\x*ln(1/5))*ln(0.2)+exp(\x*ln(0.8))*ln(0.8))/(ln(2)*(exp(\x*ln(1/5))+exp(\x*ln(0.8)) ) )} , {-\x*( exp(\x*ln(1/5))*ln(0.2)+exp(\x*ln(0.8))*ln(0.8))/(ln(2)*(exp(\x*ln(1/5))+exp(\x*ln(0.8))))+ ln((exp(\x*ln(1/5))+exp(\x*ln(0.8))))/ln(2)});
  \draw[dashed] (0,1) -- (2.6,1);
\draw  [fill, dashed]  (0.45,-0.) --(0.45,-0.3 )  node [below] {${H_\ell(\etal)}$} (0.45,-0) circle [radius=0.03];
 \draw[fill] (0.93,0.56) circle [radius=0.03] [dashed] (0.93,0.56) -- (0.93,-0.0) [fill] circle [radius=0.03]  [dashed] (0.93,0.0) -- (0.93,-0.1)  node [below]  {$H_\ell(\weta)$};
 \draw[fill] (1.41,0.89) circle [radius=0.03] [dashed] (1.41,0.89) -- (1.41,-0.0) [fill] circle [radius=0.03]  [dashed] (1.41,0.0) -- (1.41,-0.3)  node [below]  {$H_\ell(\weta)+{\Ht}_\ell(\weta)$};
  \draw [fill] (2.79,0.0)  circle [radius=0.03]  [dashed] (2.79,0.0) -- (2.79,-0.3)  node [below]  {$H_{\max}+{\Ht}_\ell(\weta)$};
 \draw [thick,color=blue]    (0.93,0.56) -- (1.41,0.89); 
\draw [dotted,color=blue]   (0,0) --  (0.93,0.57) ;
 \draw [fill] (-0.1,-0.10)   node [left] {$0$}; 
 \draw [fill] (0,1) circle [radius=0.03] node [left] {$d \ $}; 
   } 
\end{tikzpicture} 
\caption{{\bf Case $D_\mu(H_{\min}) \leq d(1-\eta)$}:   {\bf Left:} singularity spectrum of $\mu$. {\bf Right:} Almost sure singularity spectrum of  $\M_\mu$. The parts drawn with same color are translated copies of each other. One sees that the left part $H\leq H_\ell(\weta)$ of the spectrum of $\mu$ (drawn in purple) does not appear in the singularity spectrum of $\M_\mu$, and   a linear part appears in $D_{\M_\mu}$ which was not present in~$D_\mu$. Observe that the slope of $D_{\M_\mu}$ at $H_\ell(\etal)$ is  finite.} 
\label{figure6}
\end{figure}
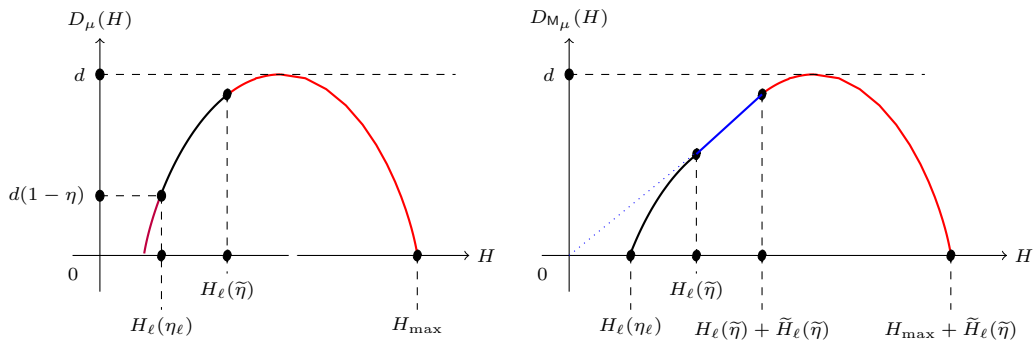

\begin{theorem}\label{thm-0} Let $\mu$ be a {non-homogeneous} Gibbs capacity on $\zu^d$.    
Let  $0<\eta<1$ be a sampling parameter.
With probability 1:
\begin{enumerate}
\sk\item
The singularity spectrum of $\M_\mu$ reads:
$$
D_{\M_\mu}(H)=
\begin{cases}\sk\sk
\ {D_\mu}(H)-d(1-\eta) &\text{when }  \ \ \ \ \ \  \ \    \ \ \ \,   \ H_\ell(\etal)  \le H \le H_\ell(\weta) ,\\ \sk\sk
 \ q_{\weta}\cdot  H&\text{when }  \ \ \ \ \ \ \ \  \ \ \ \ \, H_\ell(\weta)\le  H \leq  H_\ell(\weta)+\Ht_\ell (\weta) ,\\\sk\sk
 \ D_\mu \big (H-\Ht_\ell(\weta) \big)&\text{when } \ \ \  H_\ell(\weta)+\Ht_\ell(\weta) \le  H \leq   H_{\max}+\Ht_\ell(\weta),\\ \sk\sk
 \ -\infty&\text{otherwise}.
\end{cases}
$$

   \sk\item
  The free energy function of $\M_\mu$ is  is the limit of $(\tau_{M_\mu,j})_{j\ge 1}$ as $j\to\infty$, and $ (\M_\mu,\tau_{\M_\mu})$ forms a Legendre pair. One has
$$
 \tau_{\M_\mu}(q)=
\begin{cases}\sk
\ \tau_\mu(q)+  \Ht_\ell(\weta)\cdot  q&\text{when } \ \  q\le q_{\weta},\\\sk
\ \tau_\mu(q)+d(1-\eta)&\text{when } \ \  q_\weta  <q< q_{\eta_\ell},  \\ 
 \ H_\ell(\etal) \cdot  q&\text{when  \ \  $q_{\eta_\ell}<+\infty$ and  $q\ge q_{\eta_\ell}$}.
\end{cases}
$$

\sk\item {For all $H\ge H_\ell(\weta)+\widetilde H_\ell(\weta)$, 
$$
\dim  E_{\M_\mu}(H)=\dim\overline E_{\M_\mu}(H)= \begin{cases}
D_\mu \big(H-\Ht_\ell (\weta) \big)\!\!\!\!&\text{if }   H_\ell(\weta)+\Ht_\ell(\weta) \le  H \leq   H_{\max}+\Ht_\ell(\weta),\\
-\infty\!\!\!\!&\text{if } H>H_{\max}+\Ht_\ell(\weta).
\end{cases}
$$}
\end{enumerate}
\end{theorem}

\subsection{Comments}\label{comments}

\mk
$\bullet$
It is quite easy to see that the lower local dimension  of $\M_\mu$  at any $x$ must be greater than $H_\ell(\etal)$ (see Lemma \ref{lempointgauche}). 
It is much more involved to define and to understand the role  of the other parameters.

 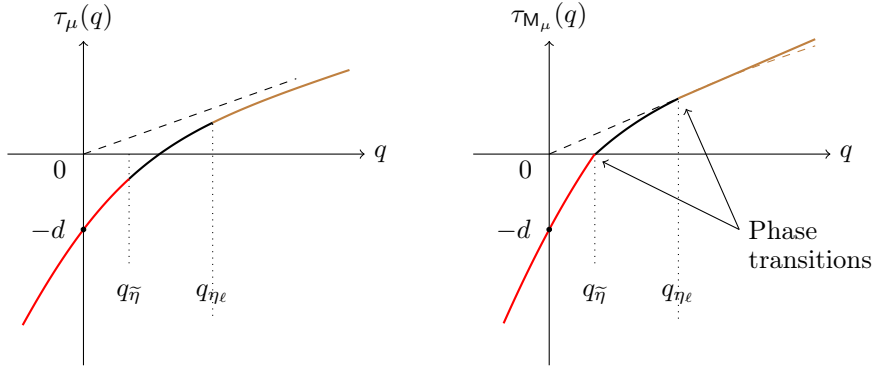
\begin{figure}\begin{tikzpicture}[xscale=1,yscale=1]
{\small
\draw [->] (0,-2.8) -- (0,1.5) [radius=0.006] node [above] {$\tau_\mu(q)$};
\draw [->] (-1.,0) -- (3.7,0) node [right] {$q$};
 \draw [thick, domain=1.7:3.5, color=brown]  plot ({\x},  {-ln(exp(\x*ln(1/5)) +exp(\x*ln(0.8)))/(ln(2)});
 \draw [thick, domain=0.6:1.7, color=black]  plot ({\x},  {-ln(exp(\x*ln(1/5)) +exp(\x*ln(0.8)))/(ln(2)});
\draw [thick, domain=-0.8:0.6, color=red ]  plot ({\x},  {-ln(exp(\x*ln(1/5)) +exp(\x*ln(0.8)))/(ln(2)});
 
 \draw [dotted] (0.6,0)-- (0.6,-1.5) [fill] (0.6,-1.6) node [below] {${q_\weta}$}; 
 \draw [dotted] (1.7,0.4)-- (1.7,-2.2) [fill] (1.7,-1.6) node [below] {$q_{\eta_\ell}$}; 
 
  \draw[dashed] (0,0) -- (2.8,1);
 
\draw [fill] (-0.1,-0.20)   node [left] {$0$}; 
\draw [fill] (-0,-1) circle [radius=0.03]  node [left] {$-d$ \ }; 
}
\end{tikzpicture}    \ \hspace{5mm} \ 
\begin{tikzpicture}[xscale=1,yscale=1]{\small
\draw [->] (0,-2.8) -- (0,1.5) [radius=0.006] node [above]  {$ \tau_{\M_\mu}(q)$};
\draw [->] (-1.,0) -- (3.7,0) node [right] {$q$};
  
 \draw [dotted] (0.6,0)-- (0.6,-1.5) [fill] (0.6,-1.6) node [below] {${q_\weta}$}; 
 \draw [dotted] (1.7,0.7)-- (1.7,-2.2) [fill] (1.7,-1.6) node [below] {$q_{\eta_\ell}$}; 
 
  \draw [dashed, domain=1.7:3.5, color=brown]  plot ({\x},  {(0.32)-ln(exp(\x*ln(1/5)) +exp(\x*ln(0.8)))/(ln(2)});
 \draw [thick, domain=0.6:1.7, color=black]  plot ({\x},  {(0.32)-ln(exp(\x*ln(1/5)) +exp(\x*ln(0.8)))/(ln(2) });
\draw [thick, domain=-0.6:0.6, color=red ]  plot ({\x},  {((0.55)*\x)-ln(exp(\x*ln(1/5)) +exp(\x*ln(0.8)))/(ln(2)});
 \draw [thick, domain=1.7:3.5, color=brown]  plot ({\x},  {(0.74/1.7)*\x });
  \draw[dashed] (0,0) -- (1.7,0.74);
 
 \draw [<-] (0.7,-0.1) -- (2.5,-1) node [right] {Phase};
 \draw [<-] (1.8,0.6) -- (2.5,-1);
 \draw[fill] (2.5,-1.4)   node   [right] {transitions};

\draw [fill] (-0.1,-0.20)   node [left] {$0$}; 
\draw [fill] (-0,-1) circle [radius=0.03]  node [left] {$-d$ \ }; 
}
\end{tikzpicture} 
\caption{{\bf Case $D_\mu(H_{\min}) \leq d(1-\eta)$}:   {\bf Left:} Free energy  $\tau_\mu$ of  $\mu$; {\bf Right:} Free energy  $ \tau_{\M_\mu}$ of $\M_\mu$.}
\label{figtau}
\end{figure}

\mk
$\bullet$
From the  free energy function $\tau_{\M_\mu}$, one   recovers the initial free energy $\tau_\mu$, except for $q\geq  q_{\eta_\ell}$. Similarly,  one   recovers $D_\mu$ from $D_{\M_\mu}$ for $H\geq \widetilde H_\ell(\etal)$. In this sense, the sampling procedure implies a loss of information on the   local dimensions, since the values of the singularity spectrum $D_\mu(H)$ are ``lost'' when $H< H_\ell(\etal)$.

\mk
 {$\bullet$ The  singularity spectra associated with the level sets $E_{\M_{\mu}}(H)$ or $\overline E_{\M_{\mu}}(H)$  are just   translated from  $D_\mu$ over $[H_\ell(\weta)+\widetilde H_\ell(\weta),H_{\max}+\widetilde H_\ell(\weta)]$.  In fact, $D_\mu(\cdot-\widetilde H_\ell(\weta))$ is still a lower bound for these spectra over $[H_{\min}+\widetilde H_\ell(\weta),H_\ell(\weta)+\widetilde H_\ell(\weta))$, but, whether $D_\mu(\cdot-\widetilde H_\ell(\weta))$ is a sharp  upper bound in this case or not, remains an open question  (see Remark~\ref{remmin}).  }

\mk
$\bullet$
The thermodynamic and geometric phase transitions mentioned earlier can now be made more precise. The free energy $\tau_{\M_{\mu}}$ is not differentiable at $ q_\weta$, and it differentiable but not twice differentiable at  $ q_{\eta_\ell}$  when $d(1-\eta) >D_\mu(H_{\min})$. Moreover, $\tau_{\M_\mu}$ is analytic outside these singularities. In the thermodynamics language,  $\tau_{\M_\mu}$ presents a first order phase transition at the inverse temperature $ q_\weta $, and a second order phase transition at the inverse temperature $q_{\eta_\ell}$ whenever  $d(1-\eta) > {D_\mu(H _{\min})}$. 

\sk

Let us mention that the study of phase transitions for  weak Gibbs measures associated with continuous potentials, started with \cite{Ruelle,Hof}, is still an active domain of research \cite{Sar,IT,BL1,BL2,FFW01,FO}.

\mk
$\bullet$
In most of the usual situations, upper bounds for dimensions of ``fractal sets'' are easily deduced from covering arguments, and lower bounds are more difficult to derive. The   structure  of $\M_\mu$, combining random and dynamical phenomena,  makes  both the derivation of the sharp upper bound {\em and}  lower bound  for $D_{\M_\mu}$ delicate.

It is too soon in the paper to give an intuition of   the proofs. Let us only say that they follow from a careful analysis of the distribution and the scaling behavior (with respect to $\mu$) of the surviving vertices. Also,  results on large deviations for Gibbs measures, heterogeneous mass transference principles (which combines ergodic and approximation theories) and percolation theory, are involved.

\mk
$\bullet$
One may also want to describe  the asymptotical statistical distribution of $\M_\mu  $ through the notion of large deviations, as is often the case in statistical physics.
\begin{definition} 
\label{defsmu}
Let $\mu\in  \mathcal{C}([0,1]^d) $ with full support.
For every set $I\subset \R^+$, and every integer $j\geq 1$, set
$$\mathcal{E}_\mu(j,I) =
  \left\{ w\in \Sigma_j: \frac{\log_2 \mu(I_w)}{ {-j}} \in
I \right\}.$$
If $H\geq 0$ and $\ep>0$, we introduce the notation 
$$\mathcal{E}_\mu(j,H\pm \ep) =
  \left\{ w\in \Sigma_j: \frac{\log_2 \mu(I_w)}{ {-j}} \in
[H-\ep,H+\ep] \right\}.$$
Then, the lower and upper large deviations spectra of $\mu$ are respectively
\begin{eqnarray*} 
\underline {f_\mu}(H)&= &  \lim_{\ep\to 0}\liminf_{j\to+\infty} \frac{\log_2 \#\mathcal{E}_\mu(j,H\pm\ep) }{j}  \\
\mbox{ and } \ \  \   \overline {f_\mu}(H)  & =  & \lim_{\ep\to 0}\limsup_{j\to + \infty} \frac{\log_2 \#\mathcal{E}_\mu(j,H\pm\ep)}{j}.
\end{eqnarray*}
 \end{definition}

Heuristically, one should have in mind that the number of words of length $j$ satisfying $\mu(I_w) \sim 2^{-jH}$ is between $2^{j \underline {f_\mu}(H)}$ and $2^{j \overline {f_\mu}(H)}$. 
Next theorem states that $\M_\mu$ behaves nicely with respect to the large deviations theory, as the Gibbs capacity $\mu$ does.
\begin{theorem}
\label{thm-2}
Under the same assumptions as in Theorem \ref{thm-0}, with probability 1,  we have 
$$ \mbox{ for all $H\ge 0$, } \ \ \  \ \underline f_{\M_\mu}(H)=\overline f_{\M_\mu}(H)=D_{\M_\mu}(H).$$ 
\end{theorem}

\subsection{Conclusion and further perspectives}

The hierarchical structure of  the initial capacity $\mu$ is so robust  that, although we greatly sample it, the remaining coefficients still possess a rich structure, especially in terms of  scaling properties and  multifractal formalism. For instance, one consequence of Theorem \ref{thm-0} is that no matter how close to 0 $\eta $ is (i.e. even if only a very small {logarithmic} proportion of vertices are kept), it is always possible to reconstruct from   the knowledge of $\tau_{\M_\mu}$  all the dimensions of the set of points with local dimension greater than $H_s$ (one can even show that when $\eta=0$,  one has $\underline E_{\M_\mu}(H)=  E_{\M_\mu}(H)=\overline E_{\M_\mu}(H)=\emptyset $ if $H<{H_s}$, $\dim\underline E_{\M_\mu}(H)=\dim  E_{\M_\mu}(H)=\dim\overline E_{\M_\mu}(H)=0$ if  $H\ge {H_s}$,  and $\dim E_{\M_\mu}(+\infty)=d$).

This phenomenon is remarkable, since at the same time,  most of  the information on the dimensions of the set of points with local dimension smaller than $H_s$ is lost. This asymmetry was, at least from our point of view, unexpected.

\mk 

Let us finish with some  perspectives:
\begin{itemize} 

\mk\item
A remaining question concerns the possible reconstruction of the Gibbs tree at the critical parameter 1/2 (see Section~\ref{recovering}).  

\mk\item It is natural to expect our result to extend to capacities obtained after sampling of  branching random walks. 

\mk\item
Instead of starting by assigning the value $\mu(I_w)$ at every node $w \in \Sigma_j$, one could give the value $\mu(I_{w_{|\lfloor j\rho\rfloor}})$ with $\rho <1$. This creates redundancy in the dyadic tree, which may balance the sparsity associated with the sampling process and provide different behaviors than those exhibited in Theorem~\ref{thm-0}.  

\mk\item
Other sampling procedures  can be investigated. In particular, one would like to allow correlations between the  $p_w$,  or make $\eta$ depend on the vertex $w$. One may also multiply $\mu(I_w)$ by some positive  random variable when $p_w=1$. Other interesting phase transitions phenomena will certainly occur.

\mk\item
{It is tempting to  iterate the sampling process by applying it to $\M_\mu$. Unfortunately our analysis does not apply to $\M_\mu$ any more, since $\M_\mu$ is not a Gibbs capacity in the sense considered in this paper. An interesting related question is whether the capacity $\M_\mu$ could be made equivalent,  after a natural renormalization procedure, to a measure, as it is the case for Gibbs capacities.}

\end{itemize}

\mk

\begin{figure}  
  \begin{tikzpicture}[xscale=1.8,yscale=2.4]{\small
   \draw [->] (0,-0.2) -- (0,1.2) [radius=0.006] node [above] {$D_{\M_\mu}(H)$ };
\draw [->] (-0.2,0) -- (3.2,0) node [right] {$H$};
\draw [thick, domain=0:2]  plot ({-(2*exp(\x*ln(4/9))*ln(4/9)+exp(\x*ln(1/9))*ln(1/9))/(ln(3)*(2*exp(\x*ln(4/9))+exp(\x*ln(1/9)) ) )-0.06} , {\x*(-2* exp(\x*ln(4/9))*ln(4/9)-exp(\x*ln(1/9))*ln(1/9))/(ln(3)*(2*exp(\x*ln(4/9))+exp(\x*ln(1/9))))+ ln(2*(exp(\x*ln(4/9))+exp(\x*ln(1/9))))/ln(3)-0.57});
\draw [thick, domain=-0.5:0]  plot ({0.47-(2*exp(\x*ln(4/9))*ln(4/9)+exp(\x*ln(1/9))*ln(1/9))/(ln(3)*(2*exp(\x*ln(4/9))+exp(\x*ln(1/9)) ) )-0.06} , {\x*(-2* exp(\x*ln(4/9))*ln(4/9)-exp(\x*ln(1/9))*ln(1/9))/(ln(3)*(2*exp(\x*ln(4/9))+exp(\x*ln(1/9))))+ ln(2*(exp(\x*ln(4/9))+exp(\x*ln(1/9))))/ln(3)-0.31});
  \draw [thick, domain=0:6,color=black]  plot ({(0.47)-( ln(0.2)+ ln(0.8))/(ln(2)) +(exp(\x*ln(1/5))*ln(0.2)+exp(\x*ln(0.8))*ln(0.8))/(ln(2)*(exp(\x*ln(1/5))+exp(\x*ln(0.8)) ) )} , {-\x*( exp(\x*ln(1/5))*ln(0.2)+exp(\x*ln(0.8))*ln(0.8))/(ln(2)*(exp(\x*ln(1/5))+exp(\x*ln(0.8))))+ ln((exp(\x*ln(1/5))+exp(\x*ln(0.8))))/ln(2)});
  \draw[dashed] (0,1) -- (2.6,1);
 \draw[fill] (0.72,0.21) circle [radius=0.03] [dashed] (0.71,0.21) -- (0.71,-0.0) [fill] circle [radius=0.03]  [dashed] (0.71,0.0) -- (0.71,-0.4)  node [below]    {$H_\ell(\etal)$};
  \draw[fill] (1.00,0.63) circle [radius=0.03] [dashed] (1.,0.63) -- (1.,-0.0) [fill] circle [radius=0.03]  [dashed] (1.,0.0) -- (1.,-0.1)  node [below]    {$H_\ell(\weta)$};
 \draw[fill] (1.57,0.94) circle [radius=0.03] [dashed] (1.57,0.91) -- (1.57,-0.0) [fill] circle [radius=0.03]  [dashed] (1.57,0.0) -- (1.57,-0.6)  node [below]  {$H_\ell(\weta) + \Ht_\ell(\weta)$};
 \draw [fill] (2.79,0.0)  circle [radius=0.03]  [dashed] (2.79,0.0) -- (2.79,-0.3)  node [below]  {$H_{\max}+  \Ht_\ell(\weta)$};
 \draw [thick,color=blue]    (1.00,0.63) -- (1.54,0.94); 
\draw [dotted,color=blue]   (0,0) --  (0.93,0.57) ;
 \draw [fill] (-0.1,-0.10)   node [left] {$0$}; 
 \draw [fill] (0,1) circle [radius=0.03] node [left] {$d \ $}; 
   } 
\end{tikzpicture}  \ 
\begin{tikzpicture}[xscale=1,yscale=1]{\small
\draw [->] (0,-2.8) -- (0,1.5) [radius=0.006] node [above]  {$\widetilde\tau(q)$};
\draw [->] (-1.,0) -- (3.7,0) node [right] {$q$};
  
 \draw [dotted] (0.6,0)-- (0.6,-1.5) [fill] (0.6,-1.6) node [below] {$q_\weta$};  
  
  \draw[color=white] (0.8,-2)-- (0.8,-4.0) ;
 \draw [thick, domain=0.6:3.5, color=black]  plot ({\x},  {(0.32)-ln(exp(\x*ln(1/5)) +exp(\x*ln(0.8)))/(ln(2) });
\draw [thick, domain=-0.6:0.6, color=red ]  plot ({\x},  {((0.55)*\x)-ln(exp(\x*ln(1/5)) +exp(\x*ln(0.8)))/(ln(2)});

 \draw [<-] (0.7,-0.1) -- (2.5,-1) node [right] {Phase};
 \draw[fill] (2.5,-1.4)   node   [right] {transition};

\draw [fill] (-0.1,-0.20)   node [left] {$0$}; 
\draw [fill] (-0,-1) circle [radius=0.03]  node [left] {$-d$ \ }; 

\draw [fill] (1,-0) circle [radius=0.03] [fill] (1,-0.2) node [ ] {$1 $}; }
\end{tikzpicture}

\caption{{\bf Case $D_\mu(H_{\min}) >d(1-\eta)$}:  Observe that $D_\mu$ and $D_{\M_\mu}$ have an infinite slope at ${H_\ell}(\etal)= H_{\min}$.}
\label{Figetalpos}
\end{figure}
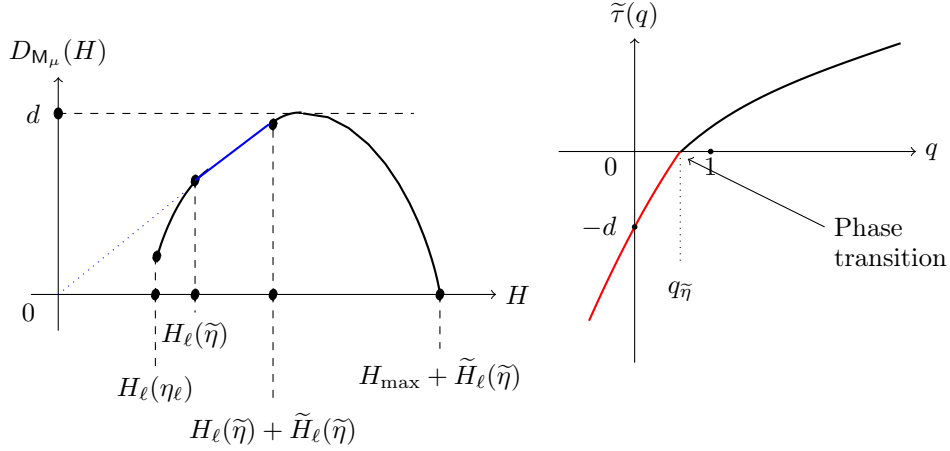
 
 The paper is organized as follows. 
 
 Section \ref{sec2} provides the reader with details on Gibbs measures and capacities, and  gathers some information about large deviations and multifractal analysis. 
 
 Section~\ref{recovering} focuses on the reconstruction of  the original capacity $\mu$ from its sample $\widetilde\mu$. 
 
 The rest of the paper is devoted to the investigation of the structure of $\M_\mu$.

 We first need to introduce new definitions to explain the origin of the parameters introduced  in Theorem \ref{thm-0}. This is achieved  in Section \ref{restate}.  {There, we first explain  that  we will work with a slight, and natural, modification of $\M_\mu$ possessing the same statistical and geometric properties as $\M_\mu$, but necessary to get an application of our result to Gibbs weighted wavelets series}.

  In Section~\ref{surv}, we investigate the scaling and distribution properties of the surviving vertices.   A key decomposition of  the value of $\mu(I_w)$ when $w$ survives, is proved (see Proposition~\ref{p2}).  

Sections~\ref{sec_upper} and~\ref{sec_lower} respectively establish the sharp upper bound and lower bound for the singularity spectrum $D_{\M_\mu}$, while Section~\ref{limlimd} is devoted to the dimensions of the sets $E_{\M_\mu}(H)$ and $\overline E_{\M_\mu}(H)$.
 The studies achieved in the Sections~\ref{sec_upper} and~\ref{sec_lower} are used in Section~\ref{free} to get  the free energy $\tau_{\M_{\mu}}$ as the limit of $(\tau_{\M_\mu,j})_{j\ge 1}$, as well as the large deviations spectra $\underline f_{\M_{\mu}}$ and $\overline f_{\M_{\mu}}$. The case of homogeneous Gibbs capacities (i.e. when the associated Gibbs measure is  the Lebesgue measure) is dealt with in Section~\ref{sechomo}. 

\medskip

{\bf Notational conventions:} 

 \sk

$\bullet $ We    always use:
\begin{itemize}
\item[-] {\bf capital letters} ($E_{\M_\mu}(H)$, $F_\mu$, ...) to characterize sets of points $x\in \zu^d$ enjoying some properties,
 
\item[-]
 {\bf curved letters} for sets of finite words having specific properties ($\mathcal{S}_{j}(\eta,W)$ for some surviving coefficients, $\Rmu(j,\eta',\alpha\pm\ep)$ or  $\Tt (j,\eta',\ep)$ for words with specific properties, see next Definition \ref{defRP}).

\item[-]
{\bf calligraphic letters} ($ \mathscr{A}$, $ \mathscr{B}$,...)   to denote probabilistic events.

\end{itemize}
 
 \sk

$\bullet $ For every finite word $W\in \Sigma_J$, $\mathcal{N}(W)$ stands for the set of $3^d -1$ words of length $J$ corresponding to the $3^d-1$ neighboring cubes at generation $J$ of $I_W$. Sometimes we will write $\mathcal{N}_J(W)$ when the length $J$  of $W$ is specified.

\section{Complements on Gibbs measures and capacities structure}
\label{sec2}

 
\subsection{Formal definition of Gibbs measures and capacities}\label{GibbsConstruction}

Let $\psi:\Sigma\to \R$ be a H\"older continuous mapping. 
Then, the function $\Psi$ defined as 
$$
\Psi([w])=\sup_{t\in [w]} \sum_{i=0}^{|w|-1}\psi(\sigma^{i} t),\quad \forall\, w\in\Sigma^*
$$
is almost additive: there exists $C_1\in\R$ such that  for all $u,v\in \Sigma^*$,
$$
|\Psi([u])+\Psi([v])- \Psi([{uv}])|\le C_1
$$
(see \cite{Ruelle}). This almost additivity property implies that the topological pressure 
$$
P(\sigma,\phi)=\lim_{j\to\infty}\frac{1}{j}\log\sum_{w\in\Sigma_j} \exp (\Psi([w]))
$$
exists in $\R$, and there exists a fully supported Gibbs measure $\nu$ on $\Sigma$ such that for another constant $C_2>0$ one has 
$$
C_2^{-1} \exp (\Psi([w])-nP(\sigma, \psi))
\le \nu([w])\le  C_2\exp (\Psi([w])-nP(\sigma,\psi)),\quad \forall\, w\in\Sigma^*.
$$
Also, there is a unique choice of such a $\nu$ so that $\nu$ is ergodic. Moreover,  the mapping $q\in\R\mapsto P(\sigma,q\psi)$ is convex, analytic, and it is strictly convex if and only if $\psi$ is not cohomologous to a constant, i.e. there is no continuous function $\varphi$ on $\Sigma$ and constant $c\in \R$ such that $\psi=c+\varphi-\varphi\circ\sigma$. These are important facts from thermodynamic formalism (see e.g. \cite{Ruelle}). 

\mk

\begin{definition}\label{Def 8} A capacity $\mu\in\mathrm{Cap}([0,1]^d)$ is a  Gibbs capacity if 
\begin{equation}\label{Kabnu}
\mu(I_w)= K \nu([w])^\alpha e^{-|w|\beta},\quad \forall\, w\in\Sigma^*,
\end{equation}
where $K>0$, $(\alpha,\beta)\in\R_+\times\R_+\setminus\{(0,0)\}$, and $\nu$ is a Gibbs measure associated with a H\"older continuous potential $\psi$ has above. 

Equivalently, one says that  $\mu$ is associated with the H\"older potential 
$$
\phi=\alpha\psi-\alpha P(\sigma,\psi) -\beta.
$$ 
The capacity $\mu$ is said to be  {\em homogeneous} when $\psi$  is cohomologous to a constant or $\alpha=0$, {i.e. when $\phi$ is cohomologous to a constant}, and {\em non-homogeneous} otherwise.
\end{definition}
Observe that if $(\alpha,\beta)=(1,0)$, $\mu$ reduces to  the Gibbs measure associated with  $\psi$, and that
$$
\tau_\mu(q)=\frac{1}{\log(2)}\Big ((\beta+\alpha P(\sigma,\psi)) q-P(\sigma,\alpha q\psi)\Big ),\quad  \forall\, q\in\R.
$$
The following fact is key: the capacity $\mu$ possesses self-similarity properties expressed through the following almost multiplicative property (easy to check): there exists a constant  $C>0$ such that 
\begin{equation}\label{quasib}
\mbox{for all words $v$ and $w$, } \ \ \ C^{-1}\mu(I_w)\mu(I_v)\le \mu(I_{wv})\le C \mu(I_w)\mu(I_v).
\end{equation}

\subsection{Large deviations and multifractal properties}

Let $\mu\in\mathcal C([0,1]^d)$ with non empty support. The concave function $\tau_\mu^*$ is called {Legendre spectrum} of $\mu$ (recall that the  Legendre transform $\tau_\mu^*$ is given by  \eqref{legendre}).  {For  a non-homogeneous Gibbs capacity $\mu$, one always has:}
{\begin{itemize}
\sk\item $\tau_\mu$ is strictly concave and analytic, and ${D_\mu}$ is strictly concave,  and real analytic   over $(H_{\min}, H_{\max})$. Also, ${D_\mu} = \tau_\mu^*$ and $(D_\mu^*)^*={D_\mu}$.
 \sk\item  If $H=\tau_\mu'(q)$, then $ \tau_\mu(q) = D_\mu ^* (q) =  q H  - {D_\mu}(H) = q \tau_\mu'(q) - {D_\mu}(\tau_\mu'(q)) $.
\sk\item
If  $q= {D_\mu}'(H)$, then  $ {D_\mu}(H) = \tau_\mu ^* (H) =    H q - \tau_\mu(q) = H D'_\mu(H) - \tau_\mu(D'_\mu(H)) $.
\end{itemize}}

These relationships will be used repeatedly in the following.

\begin{definition}\label{defnewsets}
For any fully supported capacity $\mu\in \mathcal{C}([0,1]^d)$, define  the level sets 
$$E^{\leq}_\mu(H) =\{x \in \zud: {\dim_\locloc}(\mu,x) \leq H\} \ \ \mbox{ and } \ \ E^{\geq}_\mu(H) =\{x \in \zud: {\dim_\locloc}(\mu,x) \geq H \}.$$
The sets $\underline{E}^{\leq}_\mu(H)$, $\underline{E}^{\geq}_\mu(H) $, $\overline{E}^{\leq}_\mu(H)$, $\overline{E}^{\geq}_\mu(H) $ are defined similarly using the lower and upper local dimensions, respectively.

\sk

If $j\ge 1$ and $w\in \Sigma_j$, denote by $\mathcal N_j(w)$ the set of at most $3^d$ elements $v\in\Sigma_{|w|}$ such that $I_v$ is a neighbor of $I_w$ in $\R^d$.  Also, for $x\in[0,1]^d$ and $j\ge 1$, set $\mathcal{N}_j(x)=\mathcal{N}(x_{|j})$. One defines  the set
$$ \widetilde{E}_\mu(H)\! = \!  \left \{ \!x\in \zud: \!  \lim_{j\to+\infty} \!\! \dfrac{\log _2 \max_{w\in \mathcal{N}_j(x)} \mu(I_w) }{j} = \! \lim_{j\to+\infty}  \!\! \dfrac{\log _2 \min_{w\in \mathcal{N}_j(x)} \mu(I_w) }{j} =     H \right \}.$$
\end{definition}

Obviously $\widetilde{E}_\mu(H)\subset E_\mu(H)$.  This refinement of $E_\mu(H)$ is needed when looking for the lower bound of the Hausdorff dimensions of some sets in Section \ref{sec_lower}.

\mk

A direct consequence of large deviations theory (see e.g. \cite{BRMICHPEY,OLSEN}) is a property valid for all capacities.

\begin{proposition}\label{chernov}
Let $\mu\in\mathcal C([0,1]^d)$ with full support. For all $H\le \tau_\mu'(0^+)$, one has 
$$
\limsup_{j\to\infty}\frac{1}{j}\log_2  \# \mathcal{E}_\mu(j,[0,H]) \le \tau_\mu^*(H).
$$
\end{proposition}

Next proposition  gathers information about  upper bounds for the singularity spectrum of 
$D_\mu$ in terms of   Legendre and large deviation spectra, when $\mu$ is a Gibbs capacity.  
\begin{proposition}
\label{fm}   
Let $ \mu$ be a non-homogeneous Gibbs capacity.  Recall that ${H_{\min}}=\tau_\mu'(+\infty) \, < \, H_s := \tau_\mu'(0) \, < \, {H_{\max}}= \tau_\mu'(-\infty)$.

\begin{enumerate}
\smallskip\item  For every $H\ge 0$, one has 
$$\dim \underline E_\mu(H)=\dim  E_\mu(H)=\dim \overline E_\mu(H)=\underline {D_\mu}(H)=\overline {D_\mu}(H)=\tau_\mu^*(H)={D_\mu}(H),$$ 
with  $\underline{E}_\mu(H) = \emptyset$ if and only if ${D_\mu}(H)=-\infty$. 

\smallskip \item  For every $H  \in [H_{\min},H_s]$ (i.e., in the increasing part of the singularity spectrum ${D_\mu}$),  one has  
$$\dim\,
E^{\leq}_\mu(H) =\dim\, \underline{E}^{\leq}_\mu(H) =\dim\,
\overline{E}^{\leq}_\mu(H) = {D_\mu}(H).$$

\smallskip \item  For every $H  \in [H_s, H_{\max}]$ (i.e. in the decreasing part of   ${D_\mu}$),  one has  
$$\dim\,
E^{\geq}_\mu(H)=\dim\, \underline{E}^{\geq}_\mu(H) =\dim\,
\overline{E}^{\geq}_\mu(H) =  {D_\mu}(H).$$
  
\mk \item
 For every possible local dimension $H\in (H_{\min},H_{\max})$, there exists a unique $q\in\R$ such that $H=\tau_\mu'(q)$. The Gibbs measure $\mu_H$ associated with the potential $q\phi$ is exact dimensional with dimension ${D_\mu}(H)$, and $\mu_H\big (E_\mu(H)\big ) = \mu\big(\widetilde{E}_\mu(H)\big) = 1$.
 
 \sk\item
For every $\ep>0$ and every interval {$I \subset \R_+$, there exists an integer $J_{I}$ such that for every $j\geq J_{I}$,  
$$\left| \frac{ \log_2  \mathcal{E}_\mu(j,I) }{ j } - 
 \sup_{h\in I} 
 {D_\mu}(h) \right| \leq  \ep  .  $$}
 
 \sk\item
There exists a constant $K>0$ such that  for every finite word  $w\in \Sigma^*$,
$$ \left|\frac{\log_2 \mu(I_w)}{-|w| }\right | \leq K .  $$
 

\end{enumerate}
\end{proposition}
This is deduced from \cite{BRMICHPEY,OLSEN,JLVVOJAK,BBP}.

Items (1) and (3) of the last proposition say in particular that  the Hausdorff dimension of the sets of points at which  ${\dim_\locloc}(\mu,x) = H$ is the same as the Hausdorff dimension of the set of points at which  ${\underline \dim_\locloc}(\mu,x) = H$. This will be of particular importance.

We   often use item (5) under the following form. Recall the formula for  $ \mathcal{E}_\mu(j,H\pm\ep) $  in Definition \ref{defsmu}: heuristically, $ \mathcal{E}_\mu(j,H\pm\ep) $ contains those words of length $j$ such that $\mu(I_w) \sim 2^{-j (H\pm\ep)}$.
 For every $H_{\min}\leq H \leq H_{\max}$ and $\ep, \wep>0$,  there exists a
generation  $J$ such that $j \geq J$ implies
\begin{equation}
\label{eq1sss}
\left| \frac{ \log_2 \#  \mathcal{E}_\mu(j,H\pm\ep)  }{j } - 
 \sup_{h\in [H-\ep, H+\ep]} 
 {D_\mu}(h) \right| \leq \wep.
\end{equation}
One needs to keep in mind that $ \#  \mathcal{E}_\mu(j,H\pm\ep)    \approx 2^{j {D_\mu}(H)}$.


\section{Reconstruction of the initial capacity $\mu$}
\label{recovering}

\newcommand \Mm{\M_\mu}

Fix a Gibbs capacity $\mu$. We investigate the possibility to reconstitute the whole Gibbs tree $(\mu(I_w))_{w\in \Sigma^*}$ from the sole knowledge of  $\tilde\mu$ (or equivalently, from $\Mm$). 

\mk

Assume first that  the capacity $\mu$ is associated with a  Bernoulli measure, i.e. there exists $q_{0},q_1>0$ such that   for any word $w\in\Sigma_*$  one has $\mu(I_{w1}) = q_1\mu(I_w)$ and $\mu(I_{w0}) = q_{0}\mu(I_w)$.  Hence, in order to reconstitute $\mu$, it is enough to find $q_{0}$ and $q_1$. Assume that two surviving vertices $w$ and $w'$ have different proportions of zeros and ones in their dyadic decomposition.  It is easy to check that this event has probability one.
Then the knowledge of $\mu(I_w)$ and $\mu(I_{w'})$ leads to two linearly independent equations with unknowns $q_{0}$ and $q_1$, hence to their values.  

  This idea generalizes  to the case where $\mu$ is constructed from a Markov measure, i.e. there exist an integer $k\ge 0$ and $((q_{v0},q_{v_1}))_{v\in\Sigma_k}\in (0,\infty)^{2^{k+1}}$ such that for all $w\in\Sigma_*$ and $v\in\Sigma_k$ one has $\mu(I_{wv0})=q_{v0}\mu(I_{wv})$ and $\mu(I_{wv1})=q_{v1}\mu(I_{wv})$.

\mk

When $\mu$ is associated with a general Gibbs measure the situation is not that simple. The answer we propose uses the basic tools we have at our disposal, namely concatenation of words and quasi-Bernoulli property \eqref{quasib}; it depends on the value of $\eta$, and there is a phase transition at $\eta=1/2$.

\begin{definition}
Let $k\in \N^*$. A word  $u\in \Sigma^*$ is {\em $k$-reconstructible} when there is a finite sequence of words $(w_1,u_1,w_2,u_2,... w_k,u_k)$ in $\Sigma^*$ such that 
\begin{itemize}
\mk\item
for every $i\in \{1,...,k\}$, $p_{w_i} = p_{w_i u_i} = 1 $,

\mk\item
$u = u_1 u_2\cdots u_k$.
\end{itemize}
One says that $S\subset \Sigma^*$ is $k$-reconstructible when every word $u\in S$ is $k$-reconstructible.
\end{definition}

This definition follows from the idea that when $u$ is   $k$-reconstructible, after sampling of the initial tree one has access  to the value of the weights $\mu(I_{w_i})$ and $\mu(I_{w_iu_i})$ for every $i$. Hence, by the quasi-Bernoulli property \eqref{quasib}, one estimates, up to the constant $C>1$, the value of  $\mu(I_{ u_i})$, and by concatenation of the words $u_1$, ..., $u_k $ and  \eqref{quasib} again, one reconstructs the value of $ \mu(I_{ u_i})$ up to the constant $C^{k+1}$. Next Theorem completes Theorem \ref{th_intro_reco} in the introduction.

\begin{theorem}
\label{th_reconstruction}
When $\eta<1/2$, $\Sigma^*$ is $1$-reconstructible.

When $\eta>1/2$, $\Sigma^*$ is not $k$-reconstructible, for any integer $k\geq 1$.
\end{theorem}

\begin{proof}
$\bullet$ Assume first that $\eta<1/2$.

Fix  a generation $\ell\geq 1$, and a word $u\in \Sigma_\ell$. By construction, for any word $w\in \Sigma_j$, 
\begin{equation}
\label{ptilde}
\mathbb{P}( p _{w} p_{wu} =1) = 2^{-j(1-\eta)}  2^{-(j+\ell)(1-\eta)}  =2^{-\ell(1-\eta)}  2^{-j 2(1-\eta)} .
\end{equation}

Consider the random variable $Z_j = \# \{w\in \Sigma_j:  \ p _{w} p_{wu} = 1\}$ and the event $\mathcal{Z}_j = \{ Z_j = 0\}$. By independence,   $\mathbb{P}(\mathcal{Z} _j) = \big(1- 2^{-\ell(1-\eta)}  2^{-j 2(1-\eta)}\big )^{2^j} = e^{ - 2^{-\ell (1-\eta) +j(  1-2(1-\eta)) + o(j)}}$, which tends exponentially fast to zero. By the Borel-Cantelli Lemma,  there exists  almost surely an (infinite number of) words $w\in \Sigma^*$  such that $p _{w} p_{wu} =1$, i.e. $u $ is 1-reconstructible.

\mk

$\bullet$ Assume now that $\eta>1/2$.

For every  $u\in \Sigma^*$, denote  by $r_u$ the random variable equal to 1 if $u$ is $1$-reconstructible, and 0 otherwise.  Hence,  $r_u$ is a Bernoulli variable, with parameter $\tilde p_{|u|}$, the probability that there exists $w\in \Sigma^*$ such that $p _{w} p_{wu} =1$ (which depends only on $|u|$).  By \eqref{ptilde},  
$$\forall \ j\geq 1, \ \ \  \tilde p_j \leq \sum_{w\in \Sigma^*} \mathbb{P}( p _{w} p_{wu} =1)  \leq \tilde C 2^{-j(1-\eta)},$$
for some constant $\tilde C$ independent on $w$.

Fix $\ep>0$ so small that $\eta+\ep<1$ and $(\ep_j)_{j\geq 1}$ a positive sequence converging to zero, such that $0<\ep_j\leq \ep$ and $\sum_{ j\geq 1} 2^{-j \ep_j} <+\infty$.

Let us introduce    $ \widetilde  Z_j^1  =\sum_{u\in \Sigma_j}  r_u $, the number of $1$-reconstructible words at generation $j$. The random variable $ \widetilde Z_j^1$   is a sum of non-independent random variables with common law  the Bernoulli law with parameter $\tilde p_j$.   Markov's inequality yields $\mathbb{P}(  \widetilde Z^1_j \geq  2^{j\ep_j} 2^j \tilde p_j) \leq 2^{-j\ep_j}$, and 
Borel-Cantelli's lemma  implies that almost surely, for $j$ large enough, we have 
\begin{equation}
\label{majZ1}
\widetilde Z_j^1 \leq \tilde C 2^{j\ep_j} 2^j \tilde p_j \leq   C_1 2^{j (\eta+\ep_j)},
\end{equation}
 for some other constant $  C_1$.
 This implies that $\Sigma^* $ is not $1$-reconstructible, since at most $  C_1 2^{j (\eta+\ep_j)} <\!\!\!< 2^j$ words can be reconstructed.

 Assume that for $k\geq 2$, the number $ \widetilde  Z_j^k$   of $k$-reconstructible words  at any generation $j$ is bounded from above by $  C_k j^{k}  2^{j (\eta+\ep )} $ for some constant $  C_k$. Let $J\geq k+1$. Any $(k+1)$-reconstructible word $u$  in $\Sigma_J$  is the concatenation of a $k$-reconstructible word and a $1$-reconstructible word. Hence, by \eqref{majZ1}, for the constant $  C_{ k +1}  =   C_{1}   C_{k} $, one has
 $$ \widetilde Z_J^{k+1} \leq \sum_{i=1}^{J-k}    \widetilde Z_i^1    \widetilde Z_{J-i}^k \leq    C_1   C_k \sum_{i=1}^{J-k}   2^{i(\eta+\ep_i)}  (J-i)^{k} 2^{ (J-i) (\eta+\ep)} \leq C_{k+1} J^k  2^{ J(\eta+\ep)} .$$
One concludes that  $\Sigma^* $ is not $k$-reconstructible, for any $k$, since   $  \widetilde Z_J^{k}<\!\!\!< 2^J$.
\end{proof}

The situation at the critical sampling index $\eta=1/2$ must still be investigated.

\section{New parameters, alternative definitions for the parameters $H_\ell(\etal)$, $H_\ell(\weta)$ and $\Ht(\weta)$}\label{restate}
  
From now on, we consider  a non-homogeneous Gibbs capacity $\mu$. The homogeneous case will be dealt with at the all end of the paper (Section \ref{sechomo}).

 We work with the $\|\ \|_\infty$ over $\R^d$, so that balls are Euclidean cubes.

\subsection{Modified version of the capacity $\M_\mu$}

 We will study  a   slight  modification of $\M_\mu$.

\begin{definition}
\label{defnewMmu}
 Let $\mu \in \mathrm{Cap}(\zu^d)$. We set
\begin{equation}
\label{defnewmmu}
 \widetilde\M_\mu (I_w)=\max_{u\in \mathcal N_j(w)}\M_\mu(I_u) = \max \left\{\mu(I_{uv}): u\in \mathcal N_j(w), v\in \Sigma^*, p_{uv}=1\right \}.
\end{equation} 
 \end{definition}
{Thus, the difference between the capacities $\M_\mu$ and $ \widetilde  \M_\mu$ is that $\widetilde  \M_\mu(I_j(x))$ carries information about the behavior of $\mu$ in the neighborhood of  $x$, not only in the dyadic cube $I_j(x)$ of generation $j$ containing $x$. We consider $\widetilde  \M_\mu$  for the following reasons. First it  is natural to extend $\M_\mu$ to balls: for $x\in [0,1]^d$ and $r>0$ one denotes $B(x,r)$ the closed ball of radius $r$ centered at $x$, and  defines $ \M_\mu(B(x,r))=\max\{\M_\mu (I_w):I_w\subset B(x,r)\}$. Then the multifractal analysis of $\M_\mu$ using the more intrinsic  logarithmic density $\frac{\log(\mu(B(x,r))}{\log(r)}$ to define  the local dimensions  of $\M_\mu$ is given by the multifractal analysis of~$\widetilde \M_\mu$.}

{The second reason is that  knowing the mutifractal nature of $\widetilde\M_\mu$ yields that of the sparse wavelets series  weighted by using the random sample $\widetilde\mu$ of $\mu$ (see \cite{JaffardPSPUM} for an account of multifractal analysis of functions)}. 

\mk 

 From now on, only $\widetilde  \M_\mu$ will be considered, and we denote it as $\M_\mu$.

\mk

The reader will check that our proofs to study the capacity defined by \eqref{defnewmmu}  are easily adapted to the case where the capacity is defined by \eqref{defwmu}. In fact, the case we study is a little bit more complicated, since it involves a control of all the immediate neighbors.

\subsection{New parameters}

\begin{definition}
\label{defetal}
The real number $\eta_\ell \in [0,\eta]$ is defined as 
 $$
 \eta_\ell= \begin{cases}   \ 0 & \mbox{ if }  \  0\leq D_\mu({H_{\min}})\leq d(1-\eta) \\  1- \frac{d(1-\eta)}{   {D_\mu}({H_{\min}})} & \mbox{ otherwise.}\end{cases}
 $$
 For all $\eta'\in [\eta_{\ell},\eta]$, there exists a unique ${{H_\ell}}(\eta')\in [{H_{\min}},{H_s}]$ such that  
 \begin{equation}
 \label{defHl}
D_\mu\big ({{H_\ell}}(\eta')\big)=\frac{d(1-\eta)}{1-\eta'}.
 \end{equation}
 \end{definition}

 See Figures \ref{fig3} and  \ref{figbetal} for a geometrical interpretation of ${{H_\ell}} (\eta')$, which makes it easier to understand.
By construction one has:
\begin{itemize}
\sk\item
 ${{H_\ell}}(\eta)={H_s}$,
 \sk\item
if $D_\mu({H_{\min}})\le d(1-\eta)$,   $\eta_\ell=0$ and ${{H_\ell}}(\eta_\ell)$ is the unique solution of $D_\mu(H)=d(1-\eta)$ in $[{H_{\min}},{H_s}]$,
  \sk\item  if $D_\mu({H_{\min}}) >d(1-\eta)$,   $\eta_\ell>0$ and ${{H_\ell}}(\eta_\ell)={H_{\min}}$.
 \end{itemize}

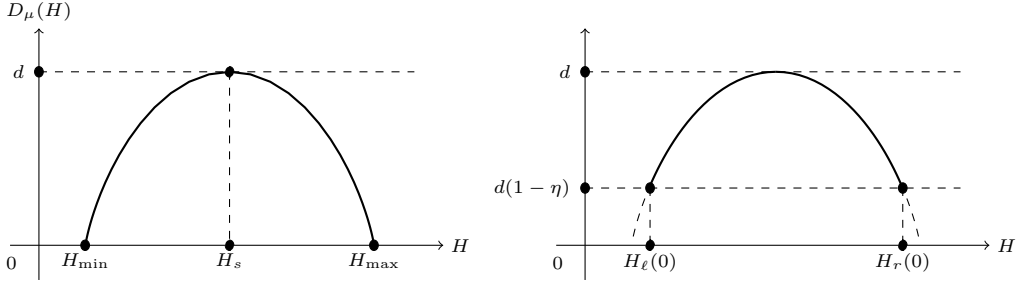
\begin{figure}   \begin{tikzpicture}[xscale=1.9,yscale=2.3]
    {\tiny
\draw [->] (0,-0.2) -- (0,1.25) [radius=0.006] node [above] {$D_\mu(H)$};
\draw [->] (-0.2,0) -- (2.8,0) node [right] {$H$};
\draw [thick, domain=0:5]  plot ({-(exp(\x*ln(1/5))*ln(0.2)+exp(\x*ln(0.8))*ln(0.8))/(ln(2)*(exp(\x*ln(1/5))+exp(\x*ln(0.8)) ) )} , {-\x*( exp(\x*ln(1/5))*ln(0.2)+exp(\x*ln(0.8))*ln(0.8))/(ln(2)*(exp(\x*ln(1/5))+exp(\x*ln(0.8))))+ ln((exp(\x*ln(1/5))+exp(\x*ln(0.8))))/ln(2)});
\draw [thick, domain=0:5]  plot ({-( ln(0.2)+ ln(0.8))/(ln(2)) +(exp(\x*ln(1/5))*ln(0.2)+exp(\x*ln(0.8))*ln(0.8))/(ln(2)*(exp(\x*ln(1/5))+exp(\x*ln(0.8)) ) )} , {-\x*( exp(\x*ln(1/5))*ln(0.2)+exp(\x*ln(0.8))*ln(0.8))/(ln(2)*(exp(\x*ln(1/5))+exp(\x*ln(0.8))))+ ln((exp(\x*ln(1/5))+exp(\x*ln(0.8))))/ln(2)});
  \draw[dashed] (0,1) -- (2.6,1);
\draw [fill] (-0.1,-0.10)   node [left] {$0$}; 
\draw [fill] (0,1) circle [radius=0.03] node [left] {$d \ $}; 
\draw  [fill] (0.32,0) circle [radius=0.03] node [below] {$H_{\min}$};
\draw  [fill] (2.32,0) circle [radius=0.03] node [below] {$H_{\max}$};
\draw  [fill] (1.32,1) circle [radius=0.03]  [dashed]   (1.32,1) -- (1.32,0)  [fill] (1.32,0) circle [radius=0.03]  node [below] {$H_{s}$};
}
\end{tikzpicture}
   \begin{tikzpicture}[xscale=1.9,yscale=2.3]
{\tiny
\draw [->] (0,-0.2) -- (0,1.25) [radius=0.006] node [above] { };
\draw [->] (-0.2,0) -- (2.8,0) node [right] {$H$};
\draw [dashed,  domain=1.9:5]  plot ({-(exp(\x*ln(1/5))*ln(0.2)+exp(\x*ln(0.8))*ln(0.8))/(ln(2)*(exp(\x*ln(1/5))+exp(\x*ln(0.8)) ) )} , {-\x*( exp(\x*ln(1/5))*ln(0.2)+exp(\x*ln(0.8))*ln(0.8))/(ln(2)*(exp(\x*ln(1/5))+exp(\x*ln(0.8))))+ ln((exp(\x*ln(1/5))+exp(\x*ln(0.8))))/ln(2)});
\draw [thick, domain=0:1.9]  plot ({-(exp(\x*ln(1/5))*ln(0.2)+exp(\x*ln(0.8))*ln(0.8))/(ln(2)*(exp(\x*ln(1/5))+exp(\x*ln(0.8)) ) )} , {-\x*( exp(\x*ln(1/5))*ln(0.2)+exp(\x*ln(0.8))*ln(0.8))/(ln(2)*(exp(\x*ln(1/5))+exp(\x*ln(0.8))))+ ln((exp(\x*ln(1/5))+exp(\x*ln(0.8))))/ln(2)});
\draw [dashed, domain=1.9:5]  plot ({-( ln(0.2)+ ln(0.8))/(ln(2)) +(exp(\x*ln(1/5))*ln(0.2)+exp(\x*ln(0.8))*ln(0.8))/(ln(2)*(exp(\x*ln(1/5))+exp(\x*ln(0.8)) ) )} , {-\x*( exp(\x*ln(1/5))*ln(0.2)+exp(\x*ln(0.8))*ln(0.8))/(ln(2)*(exp(\x*ln(1/5))+exp(\x*ln(0.8))))+ ln((exp(\x*ln(1/5))+exp(\x*ln(0.8))))/ln(2)});
\draw [thick, domain=0:1.9]  plot ({-( ln(0.2)+ ln(0.8))/(ln(2)) +(exp(\x*ln(1/5))*ln(0.2)+exp(\x*ln(0.8))*ln(0.8))/(ln(2)*(exp(\x*ln(1/5))+exp(\x*ln(0.8)) ) )} , {-\x*( exp(\x*ln(1/5))*ln(0.2)+exp(\x*ln(0.8))*ln(0.8))/(ln(2)*(exp(\x*ln(1/5))+exp(\x*ln(0.8))))+ ln((exp(\x*ln(1/5))+exp(\x*ln(0.8))))/ln(2)});
\draw[dashed] (0,1) -- (2.6,1);
  \draw[dashed] (0,0.33) -- (2.6,0.33) ; 
 \draw [fill]  (-0,0.33)  circle  [radius=0.03] node [left] {$d(1-\eta)\ $};
 \draw [fill]  (0,0.33)  circle  [radius=0.002]  ;
 \draw [fill]  (0.45,0.33)  circle  [radius=0.03]  ;
\draw[dashed] (0.45,0.33) -- (0.45,0) [fill] circle [radius=0.03] node [below] {${{H_\ell}}(0)$};
\draw[dashed] (2.20,0.33) -- (2.20,0) [fill] circle [radius=0.03] node [below] {${H_r}(0)$};
 \draw [fill]  (2.20,0.33)  circle  [radius=0.03]  ; 
 \draw [fill] (-0.1,-0.10)   node [left] {$0$}; 
 \draw [fill] (0,1) circle [radius=0.03] node [left] {$d \ $}; 
} 
\end{tikzpicture}   
\caption{{\bf Left:} Typical singularity spectrum of a Gibbs measure. {\bf Right:} Parameters ${{H_\ell}}(0)$ and ${H_r}(0)$.}
\label{fig3}
\end{figure}

\begin{definition}
\label{defgammal}
For $\eta'\in[\eta_\ell,\eta]\setminus\{0\}$, let 
\begin{eqnarray}
\nonumber{\Ht_\ell}(\eta') & = & \left( \frac{1}{\eta'}-1\right){{H_\ell}} (\eta') \\
\mbox{ and }  \ \ \ \  \ \ \  \ \ \ 
\label{defweta} \widetilde\eta & = & \mathrm{argmin}_{\eta'\in [\eta_\ell,\eta]\setminus\{0\}} \ {\Ht_\ell}(\eta').
\end{eqnarray}
 \end{definition}

Again, see Figure \ref{figbetal} for a geometrical interpretation of these parameters (in the case we discard at the moment, i.e. when $\mu$ is homogeneous, the function $\tau_\mu$ is linear so that $H_{\min}=H_s=H_{\max}$ and  $\widetilde\eta=\eta_\ell=\eta$). It is easily seen that by definition the value $\weta$ is so that the straight line passing through the points $(0,d(1-\eta))$ and $\left({H_\ell}(\weta), \frac{d(1-\eta)}{1-\weta}\right)$ is tangent to the singularity spectrum of $\mu$. This value always exists and  is unique. Since ${D_\mu}$ is strictly concave,   ${D_\mu}'(H_{\min}+)=\infty$ and ${D_\mu}'(H_s)=0$, one has $\weta \in(\eta_\ell,\eta)$.

\begin{definition} 
\label{degtaut}
Let ${q_{\widetilde \eta}}$ be the unique solution to the equation  
\begin{eqnarray}
\label{defqeta}
{{H_\ell}}(\widetilde\eta) & = & \tau_\mu'({q_{\widetilde \eta}})\\
\mbox{ and } \ \ \ \ \ \ \   \ \ \ 
\label{defqetap}
{q_{\eta_\ell}} & = & \sup \big\{q\ge 0: {\tau_\mu^*}(\tau_\mu'(q))\ge d(1-\eta)\big\}.
\end{eqnarray}

Let $\wtau:\R \to \R$ be the mapping  defined as
\begin{equation}
\label{formtaut2}
\widetilde\tau(q)=
\begin{cases}
\ \tau_\mu(q)+{\Ht_\ell}(\widetilde\eta)q&\text{if } q\le {q_{\widetilde \eta}},\\
\ \tau_\mu(q)+d(1-\eta)&\text{if }{q_{\widetilde \eta}}<q<{q_{\eta_\ell}},  \\
\ {{H_\ell}} (0) q&\text{if ${q_{\eta_\ell}}<\infty$ and $q\ge {q_{\eta_\ell}}$}.
\end{cases}
\end{equation}
\end{definition}

See Figure \ref{figtau} for a representation of $\wtau$.  

Observe that $q_{\eta_\ell} <+\infty $ if and only if $\eta_\ell =0$, and in this case $\tau_\mu'(q_{\eta_\ell})={H_\ell}(0)$.

\begin{definition}
\label{defetar}
The real number $\etar \in [0,\eta]$ is defined as 
 $$
 \etar= \begin{cases}   \ 0 & \mbox{ if }  \ 0\leq  {{D_\mu}}({H_{\max}})\leq d(1-\eta) \\  1- \frac{d(1-\eta)}{   {{D_\mu}}({H_{\max}})} & \mbox{ otherwise.}\end{cases}
 $$
 For all $\eta'\in [\etar,\eta]$, there exists a unique ${H_r}(\eta')\in [{H_s}, {H_{\max}}]$ such that  $$
{{D_\mu}}\big ({H_r}(\eta')\big)=\frac{d(1-\eta)}{1-\eta'}.
 $$
 \end{definition}
 By construction one has:
\begin{itemize}
\sk\item
 ${H_r}(\eta)={H_s}$,
 \sk\item
if ${{D_\mu}}({H_{\max}})\le d(1-\eta)$,   $\etar=0$ and ${H_r}(\etar)$ is the unique solution of ${{D_\mu}}(H)=d(1-\eta)$ in $[{H_s}, {H_{\max}}]$,
  \sk\item  if ${{D_\mu}}({H_{\max}}) >d(1-\eta)$,   $\etar>0$ and ${H_r}(\etar)={H_{\max}}$.
 \end{itemize}

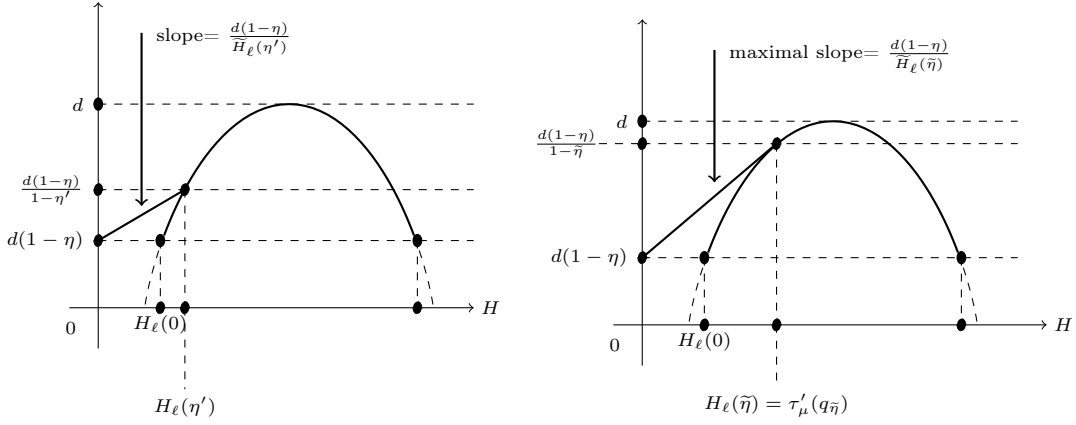
\begin{figure}
\begin{tikzpicture}[xscale=1.9,yscale=2.7]
{\tiny
\draw [->] (0,-0.2) -- (0,1.5) [radius=0.006] node [above] { };
\draw [->] (-0.2,0) -- (2.6,0) node [right] {$H$};
\draw [dashed,  domain=1.9:5]  plot ({-(exp(\x*ln(1/5))*ln(0.2)+exp(\x*ln(0.8))*ln(0.8))/(ln(2)*(exp(\x*ln(1/5))+exp(\x*ln(0.8)) ) )} , {-\x*( exp(\x*ln(1/5))*ln(0.2)+exp(\x*ln(0.8))*ln(0.8))/(ln(2)*(exp(\x*ln(1/5))+exp(\x*ln(0.8))))+ ln((exp(\x*ln(1/5))+exp(\x*ln(0.8))))/ln(2)});
\draw [thick, domain=0:1.9]  plot ({-(exp(\x*ln(1/5))*ln(0.2)+exp(\x*ln(0.8))*ln(0.8))/(ln(2)*(exp(\x*ln(1/5))+exp(\x*ln(0.8)) ) )} , {-\x*( exp(\x*ln(1/5))*ln(0.2)+exp(\x*ln(0.8))*ln(0.8))/(ln(2)*(exp(\x*ln(1/5))+exp(\x*ln(0.8))))+ ln((exp(\x*ln(1/5))+exp(\x*ln(0.8))))/ln(2)});
\draw [dashed, domain=1.9:5]  plot ({-( ln(0.2)+ ln(0.8))/(ln(2)) +(exp(\x*ln(1/5))*ln(0.2)+exp(\x*ln(0.8))*ln(0.8))/(ln(2)*(exp(\x*ln(1/5))+exp(\x*ln(0.8)) ) )} , {-\x*( exp(\x*ln(1/5))*ln(0.2)+exp(\x*ln(0.8))*ln(0.8))/(ln(2)*(exp(\x*ln(1/5))+exp(\x*ln(0.8))))+ ln((exp(\x*ln(1/5))+exp(\x*ln(0.8))))/ln(2)});
\draw [thick, domain=0:1.9]  plot ({-( ln(0.2)+ ln(0.8))/(ln(2)) +(exp(\x*ln(1/5))*ln(0.2)+exp(\x*ln(0.8))*ln(0.8))/(ln(2)*(exp(\x*ln(1/5))+exp(\x*ln(0.8)) ) )} , {-\x*( exp(\x*ln(1/5))*ln(0.2)+exp(\x*ln(0.8))*ln(0.8))/(ln(2)*(exp(\x*ln(1/5))+exp(\x*ln(0.8))))+ ln((exp(\x*ln(1/5))+exp(\x*ln(0.8))))/ln(2)});
  \draw[dashed] (0,1) -- (2.6,1);
 \draw [fill]  (-0,0.33)  circle  [radius=0.03] node [left] {$d(1-\eta)\ $}[dashed] (0,0.33) -- (2.6,0.33);
  \draw [fill]  (-0,0.58)  circle  [radius=0.03] node [left] {$\frac{d(1-\eta)}{1-\eta'}\ $} [dashed] (0,0.58) -- (2.6,0.58) ;
   \draw [fill]  (0,0.33)  circle  [radius=0.002]  ;
 \draw [fill]  (0.43,0.33)  circle  [radius=0.03]  ;
\draw[dashed] (0.43,0.33) -- (0.43,0) [fill] circle [radius=0.03] node [below] {$ {{H_\ell}}(0)$};
\draw[fill] (0.6,0.58) circle [radius=0.03] [dashed] (0.6,0.58) -- (0.6,-0.0) [fill] circle [radius=0.03]  [dashed] (0.6,0.0) -- (0.6,-0.4)  node [below]   {${{H_\ell}}(\eta')$};
\draw[dashed] (2.21,0.33) -- (2.21,0) [fill] circle [radius=0.03] node [below] {$ $};
 \draw [fill]  (2.21,0.33)  circle  [radius=0.03]  ; 
 \draw [fill] (-0.1,-0.10)   node [left] {$0$}; 
 \draw [fill] (0,1) circle [radius=0.03] node [left] {$d \ $}; 
   \draw [thick] (0,0.33) -- (0.6,0.58);
    \draw   (0.3, 1.2) node [above right] {\ slope= $ \frac{d(1-\eta)}{{\Ht_\ell}(\eta')}$}  [ ->,thick]    (0.3, 1.35)  -- (0.3,0.5);
 }
\end{tikzpicture}
 \   \begin{tikzpicture}[xscale=1.9,yscale=2.7]
{\tiny 
\draw [->] (0,-0.2) -- (0,1.5) [radius=0.006] node [above] { };
\draw [->] (-0.2,0) -- (2.8,0) node [right] {$H$};
\draw [dashed,  domain=1.9:5]  plot ({-(exp(\x*ln(1/5))*ln(0.2)+exp(\x*ln(0.8))*ln(0.8))/(ln(2)*(exp(\x*ln(1/5))+exp(\x*ln(0.8)) ) )} , {-\x*( exp(\x*ln(1/5))*ln(0.2)+exp(\x*ln(0.8))*ln(0.8))/(ln(2)*(exp(\x*ln(1/5))+exp(\x*ln(0.8))))+ ln((exp(\x*ln(1/5))+exp(\x*ln(0.8))))/ln(2)});
\draw [thick, domain=0:1.9]  plot ({-(exp(\x*ln(1/5))*ln(0.2)+exp(\x*ln(0.8))*ln(0.8))/(ln(2)*(exp(\x*ln(1/5))+exp(\x*ln(0.8)) ) )} , {-\x*( exp(\x*ln(1/5))*ln(0.2)+exp(\x*ln(0.8))*ln(0.8))/(ln(2)*(exp(\x*ln(1/5))+exp(\x*ln(0.8))))+ ln((exp(\x*ln(1/5))+exp(\x*ln(0.8))))/ln(2)});
\draw [dashed, domain=1.9:5]  plot ({-( ln(0.2)+ ln(0.8))/(ln(2)) +(exp(\x*ln(1/5))*ln(0.2)+exp(\x*ln(0.8))*ln(0.8))/(ln(2)*(exp(\x*ln(1/5))+exp(\x*ln(0.8)) ) )} , {-\x*( exp(\x*ln(1/5))*ln(0.2)+exp(\x*ln(0.8))*ln(0.8))/(ln(2)*(exp(\x*ln(1/5))+exp(\x*ln(0.8))))+ ln((exp(\x*ln(1/5))+exp(\x*ln(0.8))))/ln(2)});
\draw [thick, domain=0:1.9]  plot ({-( ln(0.2)+ ln(0.8))/(ln(2)) +(exp(\x*ln(1/5))*ln(0.2)+exp(\x*ln(0.8))*ln(0.8))/(ln(2)*(exp(\x*ln(1/5))+exp(\x*ln(0.8)) ) )} , {-\x*( exp(\x*ln(1/5))*ln(0.2)+exp(\x*ln(0.8))*ln(0.8))/(ln(2)*(exp(\x*ln(1/5))+exp(\x*ln(0.8))))+ ln((exp(\x*ln(1/5))+exp(\x*ln(0.8))))/ln(2)});
  \draw[dashed] (0,1) -- (2.6,1);
 \draw [fill]  (-0,0.33)  circle  [radius=0.03] node [left] {$d(1-\eta)\ $}[dashed] (0,0.33) -- (2.6,0.33);
 \draw (-0.2,0.89)  node [left] {$\frac{d(1-\eta)}{1-\widetilde\eta}$  $ $ }  [dashed] (-0.3,0.89) -- (0,0.89) [fill]  (-0,0.89)  circle  [radius=0.03]  [dashed] (0,0.89) -- (2.6,0.89) ;
  \draw [fill]  (0,0.33)  circle  [radius=0.002]  ;
 \draw [fill]  (0.43,0.33)  circle  [radius=0.03]  ;
\draw[dashed] (0.43,0.33) -- (0.43,0) [fill] circle [radius=0.03] node [below] {$ {{H_\ell}}(0)$};
\draw[fill] (0.93,0.89) circle [radius=0.03] [dashed] (0.93,0.89) -- (0.93,-0.0) [fill] circle [radius=0.03]  [dashed] (0.93,0.0) -- (0.93,-0.3)  node [below]  {${{H_\ell}}(\widetilde \eta)=\tau_\mu'({q_{\widetilde \eta}})$};
\draw[dashed] (2.21,0.33) -- (2.21,0) [fill] circle [radius=0.03] node [below] { };
 \draw [fill]  (2.21,0.33)  circle  [radius=0.03]  ; 
 \draw [fill] (-0.1,-0.10)   node [left] {$0$}; 
 \draw [fill] (0,1) circle [radius=0.03] node [left] {$d \ $}; 
 \draw [thick] (0,0.33) -- (0.93,0.89);
 \draw   (0.5, 1.2) node [above right] {\ maximal slope= $ \frac{d(1-\eta)}{{\Ht_\ell}(\widetilde \eta)}$   }  [ ->,thick]    (0.5, 1.35)  -- (0.5,0.7);
 }
\end{tikzpicture}
\caption{{\bf Left:} Parameters ${{H_\ell}}(\eta')$ and ${{\Ht_\ell}}(\eta')$. {\bf Right:} Optimal parameter $\weta$.}
\label{figbetal}
 \end{figure}

 The existence of ${H_\ell}$ and $H_r$ is ensured by the continuity of the Legendre  spectrum ${D_\mu}$  on its support.   
   
 \medskip
 
As ${\Ht_\ell}(\eta')$ was defined in Definition \ref{defgammal}, we can also define a parameter ${\Ht_r}(\eta')$ as follows: for every $\eta'\in[\eta_r,\eta]\setminus\{0\}$, let 
$$
{\Ht_r}(\eta')=\left( \frac{1}{\eta'}-1\right)H_r (\eta').
$$
The geometrical interpretation is the same as the one for ${\Ht_\ell}(\eta')$ (see Figure \ref{figbetal}), except now that everything is done on the decreasing part of the spectrum.

\sk

Finally, the following lemma provides us with another interpretation of the exponent  ${H_\ell}(\weta)$ (see \eqref{defweta}), which is useful to simplify some   formulas and to understand its role.

\begin{lemma}
\label{lem2s}
One has 
\begin{equation}\label{Hell}
{{H_\ell}}(\widetilde\eta) =  \mathrm{argmax}_H\left(\frac{{D_\mu}(H)}{H+{\Ht_\ell}(\widetilde\eta)}\right ).
\end{equation} 
\end{lemma}

\begin{proof} Due to the unimodal character of ${D_\mu}$, the maximum we seek for is reached at $H\in [H_{\min},H_s]$.  A rapid calculation shows that since ${D_\mu}$ is strictly concave and differentiable over $(H_{\min},H_s]$ with ${D_\mu}'(H_{\min}+)=+\infty$ and  ${D_\mu}'(H_s)=0$, then for any $\gamma>0$, $H\mapsto \frac{{D_\mu}(H)}{H+\gamma}$ reaches its maximum at a unique point of $(H_{\min},H_s)$.  Notice that from its definition the function $\eta'\mapsto {H_\ell}(\eta') $ is differentiable. 

 Let us introduce the function $\varphi(\eta')=\eta'{\Ht_\ell}(\eta') = (1-\eta'){H_\ell}(\eta')$.  Recall that by \eqref{defHl}, ${D_\mu}\big({H_\ell}(\eta') \big) = \dfrac{d(1-\eta)}{1-\eta'}$. So ${D_\mu}'\big({H_\ell}(\eta')\big ){H_\ell}'(\eta') = \dfrac{d(1-\eta)}{(1-\eta')^2} = \dfrac{{D_\mu}\big({H_\ell}(\eta')\big)}{1-\eta'}$. One deduces that
 $$
\varphi'(\eta')= -  {H_\ell}(\eta') +(1-\eta'){H_\ell}'(\eta') = -{H_\ell}(\eta') +\frac{{D_\mu}\big({H_\ell}(\eta')\big)}{ {D_\mu}'\big({H_\ell}(\eta')\big)}= -\dfrac{{D_\mu}^*\big({D_\mu}'\big({{H_\ell}}(\eta')\big)\big)}{{D_\mu}'\big({{H_\ell}}(\eta')\big)},
$$
since $ {D_\mu}^*(H) = H{D_\mu}' (H)- {D_\mu}(H)$.

On the other hand,  the derivative of $H\mapsto \dfrac{{D_\mu}(H)}{H+{\Ht_\ell}(\widetilde\eta)}$ vanishes at $\alpha= \mathrm{argmax}_H\left(\frac{{D_\mu}(H)}{H+{\Ht_\ell}(\widetilde\eta)}\right )$. This yields 
$$
 {D_\mu}'(\alpha) (\alpha+ {\Ht_\ell}(\weta)) - {D_\mu}(\alpha)= 0,$$
i.e.
$$ {\Ht_\ell}(\widetilde\eta) =  -\dfrac{{D_\mu}^*({D_\mu}'(\alpha ))}{{D_\mu}'(\alpha)}.
$$
Since $\weta$ is chosen so that ${\Ht_\ell}(\eta')$ is minimal at $\weta$, we have   $\Ht'_\ell(\weta)=0$.  This implies that  $\varphi'(\weta) = {\Ht_\ell}(\weta)$,   so finally
\begin{equation}
\label{eq2sss}
-\dfrac{{D_\mu}^*\big({D_\mu}'(\alpha )\big)}{{D_\mu}'(\alpha )}=  -\dfrac{{D_\mu}^*\big({D_\mu}'\big({{H_\ell}}(\widetilde \eta\big)\big)}{{D_\mu}'\big({{H_\ell}}(\widetilde \eta)\big)}.
\end{equation}
Recalling that ${D_\mu}$ is the Legendre transform of $\tau_\mu$, we know that $q\in \R^*_+\mapsto \tau_\mu'(q)$  is a bijection onto $(H_{\min},H_s)$. Hence,    since the  mapping $q>0\mapsto -\dfrac{\tau_\mu(q)}{q}$ is injective ($\tau_\mu$ being strictly concave),  the identification   $\Big(H,q, {D_\mu}^*\big({D_\mu}'(H)\big)\Big)= \Big(\tau_\mu'(q),{D_\mu}'(H),\tau_\mu(q)\Big)$ implies that $H \in (H_{\min},H_s) \mapsto -\dfrac{{D_\mu}^*\big({D_\mu}'(H)\big)}{{D_\mu}'(H)}$is injective as well.  Equation \eqref{eq2sss} yields finally $\alpha = {H_\ell}(\weta)$.   
\end{proof}
 
 {The previous definitions and discussion clarify the origin of the parameters introduced to state Theorem \ref{thm-0}}. The rest of the paper is devoted to the proof of the multifractal properties of $\Mm$ definied by \eqref{defnewMmu}.

 \section{Analysis of the surviving vertices}
 \label{surv}

 \subsection{Basic properties of the distribution of the surviving vertices}

Recall the Definition~\ref{def11} in which $\cjeta$ is defined, and recall that  $x_w$, defined by \eqref{defxw}, is the dyadic point corresponding to the projection of the finite word $w\in \Sigma_j$ to $\zu^d$. The first question concerns the distribution of the  points $x_w$,  for $w\in \cjeta$.

 \begin{definition}
 For every $j\geq 1$, and every finite word $W\in \Sigma^*$, we define
 $$  \mathcal{S}_j(\eta,W) = \{w\in  \mathcal{S}_j(\eta): I_w \subset I_W\}.$$    
 \end{definition}

The set $ \mathcal{S}_j(\eta,W) $ describes the {\em surviving} coefficients at generation $j$ included in $I_W$.

 Obviously, for every $J\leq j$, 
 $$ \mathcal{S}_j(\eta) =\bigcup _{W\in \Sigma_J}   \mathcal{S}_j(\eta,W) .$$

 \begin{lemma}
 \label{lem1}
There exists a positive sequence $(\ep_j)_{j\geq 1}$ converging to 0 such that, with probability 1,    for  every   $j$ large enough,  for every $W\in \Sigma_{\lfloor j(\eta-\ep_j)\rfloor}$,  $\mathcal{S}_j(\eta,W) \neq \emptyset$.  
 \end{lemma}
 
 In other words, every cylinder of generation $\lf j(\eta-\ep_j)\rf$ contains a surviving vertex $w$ of generation $j$.
\begin{proof}
Fix a positive sequence $(\ep_j)_{j\geq 1}$ converging to 0. For each $j\ge 1$ and $W \in \Sigma_{\lfloor j(\eta-\ep_j)\rfloor}$, the cylinder $[W]$ contains exactly $2^{j- \lfloor j(\eta-\ep_j)\rfloor}$ distinct cylinders $[w]$, with $w\in\Sigma_j$. Denote this set by $S(W)$. The probability of the event $\mathscr{E}(W)=\{\forall\, w\in S(W),\ p_{w}=0\}$ is given by $(1-2^{-j(1-\eta)})^{2^{j- \lfloor j(\eta-\ep_j)\rfloor}}$. Thus, 
\begin{eqnarray*}
\mathbb{P}\Big ( \bigcup_{W\in \Sigma_{\lfloor j(\eta-\ep_j)\rfloor}} \mathscr{E}(W)\Big )&\le& 2^{ \lfloor j(\eta-\ep_j)\rfloor}  (1-2^{-j(1-\eta)})^{2^{j- \lfloor j(\eta-\ep_j)\rfloor}}\\
&\le& 2^{ \lfloor j(\eta-\ep_j)\rfloor} \exp (-2\cdot 2^{j\ep_j}).
\end{eqnarray*}
If we choose $\ep_j =  (\log^2(j))/j$, we get $\sum_{j\ge 1}\mathbb{P}\Big ( \bigcup_{v\in \Sigma_{\lfloor j(\eta-\ep_j)\rfloor}} \mathscr{E}(W)\Big )<\infty$. So the Borel-Cantelli lemma yields that, with probability 1, for $j$ large enough, for all $W\in \Sigma_{\lfloor j(\eta-\ep_j)\rfloor}$, there exists $w\in \Sigma_j$ such that $I_w\subset I_W$ and $p_{w}=1$, i.e. $w\in\mathcal S_j(\eta,W)$.  
\end{proof}

The sequence $(\ep_j)_{j\geq 1}$ is now fixed.

 Lemma \ref{lem1} has the following consequence: Almost surely,  the set of points belonging to an infinite number of balls of the form $B(x_w, 2^{-\lfloor |w|(\eta-\ep_{|w|})\rfloor})$ with $p_w=1$  is exactly  the whole cube $\zud$, i.e.
 \begin{equation}
 \label{cover1}
 \zud = \limsup_{j\to +\infty} \ \bigcup_{w\in \mathcal{S}_j(\eta)} B(x_w, 2^{-\lfloor |w|(\eta-\ep_{|w|})\rfloor}).
 \end{equation}
  
  Next we obtain an upper bound for the cardinality of $\mathcal{S}_j(\eta,W)$ when $W \in \Sigma_{\lfloor \eta j \rfloor}$.
 
 \begin{lemma}
 \label{lem2}
  With probability one,  for  every large  $j$, for every $W \in \Sigma_{\lfloor \eta j \rfloor}$,   $ \#  \mathcal{S}_j(\eta,W) \leq j$.
 \end{lemma}
\begin{proof}
This is standard computations. Denote  for every $j\geq 1$ and every word $W \in \Sigma_{\lfloor \eta j \rfloor}$, the random variable
$${B}_{W} = \sum_{ w\in \Sigma_j: \, I_w\subset  I_W    } \ p_w$$
is equal to the cardinality of $  \mathcal{S}_j(\eta,W)$.

With this formulation, the $( {B}_{W})_{W\in \Sigma_{\lfloor \eta j \rfloor}}$ are i.i.d. random variables with common law the binomial law $ {B}(n_j,  \rho_j)$ of parameters  $n_j=2^{d(j-\lfloor \eta j \rfloor)} $ and $ \rho_j=2^{-dj(1-\eta)}$. We have 
\begin{eqnarray*}
\!\!\!\!\!\!\!\!    \!\!\!\!  &&\mathbb{P}\big ( {B}(n_j,  \rho_j) \geq  j \big)  =    \sum_{l =    j } ^{n_j}  \binom{n_j}{l}    ( \rho_j)^l (1- \rho_j)^{n_j-l} = \\
\!\!\!\!\!\!\!\!   \!\!\!\!  &&  \sum_{l =    j } ^{n_j} \frac{2^{d(j-\lfloor \eta j \rfloor)}(2^{d(j-\lfloor \eta j \rfloor)}-1) \ldots (2^{d(j-\lfloor \eta j \rfloor)}- (l-1)) }{l !}    2^{-djl(1-\eta)} (1-2^{-dj(1-\eta)})^{2^{d(j-\lfloor \eta j \rfloor)}-l}
\end{eqnarray*}

Observing that 
$$ \frac{2^{d(j-\lfloor \eta j \rfloor)} \ldots (2^{d(j-\lfloor \eta j \rfloor)}- (l-1)) }{  2^{djl(1-\eta)}} = (2^{d(\eta j  -[ \eta j ])}) \ldots (2^{d(\eta j  -[ \eta j ])}  - (l-1) 2^{-dj(1-\eta)}  )   , $$
this quantity is upper bounded by   $2^{dl} $. Finally, 
\begin{eqnarray*}
 \mathbb{P} \big ({B}(n_j, \rho_j) \geq   j \big)  \leq   \sum_{l =    j } ^{n_j} \frac{2^{dl}  }{l !}    (1-2^{-dj(1-\eta)})^{2^{d(j-\lfloor \eta j \rfloor)}-l}\ \leq   \sum_{l =    j } ^{+\infty} \frac{2^{dl}}{l !}     \leq  2^{-dj} 
   \end{eqnarray*}
for $j$ large enough.  We deduce that  $\sum_{j\geq 1} 2^{d\lfloor \eta j \rfloor} \mathbb{P} \big( {B}(n_j, \rho_j) \geq   j \big) <+\infty$. Then the Borel-Cantelli lemma yields that  almost surely there exists $J \geq 1$ such that for all $j\ge J$, for  all $W\in \Sigma_{\lfloor \eta j \rfloor}$, one has  ${B}_W  <   j$.
\end{proof}

As a conclusion, one keeps in mind the intuition that every cylinder $W\in \Sigma_{\lf \eta j  \rf}$ contains at least one, but not much more than one surviving vertex $w\in \mathcal{S}_j(\eta)$.

\subsection{Analysis of the values of $\mu$ at the surviving vertices}\label{anaval}

 The above lemmas give some hints about the possible values for $\mu(I_w)$ for $w\in \mathcal{S}_j(\eta)$. Indeed,  observe that   any word $w$ can be written as the concatenation $w = w_{|\lf \eta j  \rf} \cdot \sigma^{ \lf \eta j  \rf} w.$
Further,   by the almost multiplicativity property of $\mu$,  one has 
$$\mu(I_w) \approx \mu(I_{w_{|\lf \eta j \rf}}) \mu(I_{\sigma^{ \lf \eta j  \rf} w}).$$
 Lemmas \ref{lem1} and \ref{lem2} assert that  all the possible  values  for $  \mu(I_{w_{|\lf \eta j \rf}})$ are reached.  Hence, in order to describe the values of $\mu(I_w)$, it is necessary to investigate the possible values for $ \mu(I_{\sigma^{ \lf \eta j  \rf} w})$ when $w\in \mathcal{S}_j(\eta)$. A quick analysis could lead to the intuition that since most of the coefficients are put to zero, only the most frequent local dimension $H_s$ survive, i.e. $\mu(I_{\sigma^{ \lf \eta j  \rf} w}) \approx 2^{- \lf j \eta \rf H_s}$.

 The goal of this section is to prove that this intuition is neither true, nor absolutely false. In fact, we are going to explain that in order to investigate the values of $\mu(I_w)$ for $w\in \mathcal{S}_j(\eta)$, one needs to look at all the decompositions
 \begin{equation}
 \label{decompeta'}
 w = w_{|\lf \eta'j \rf} \cdot \sigma^{ \lf \eta'j \rf} w,
 \end{equation}
and to use that 
$$\mu(I_w) \approx \mu(I_{w_{|\lf \eta'j\rf}}) \mu(I_{\sigma^{ \lf \eta'j  \rf} w}),$$
for all possible values of $\eta'\in [\etal,\eta]\cup[\etar,\eta]$, and that the most frequent behaviors for $\mu(I_{\sigma^{ \lf \eta'j  \rf} w})$ are related to the local dimensions ${H_\ell}(\eta')$ and $H_r(\eta')$, $H_s$ corresponding to ${H_\ell}(\eta)=H_r(\eta)$.

These considerations lead to the following definition.

\begin{definition}
\label{defRP}
Let $\alpha,\ep \geq 0$ be two real numbers, and let $\eta'\in [0,\eta]$. 

When $w\in \Sigma_j$, the prefix $
 w_{|\lf \eta'j  \rf} $ is referred to as the {\em $\eta'$-root} of $w$, and the suffix $ \sigma^{ \lf \eta'j  \rf} w$ is the {\em $\eta'$-tail} of $w$.

 \mk
 
We introduce the following sets:

\begin{itemize}

\sk\item
$\Rmu(j,\eta',\alpha\pm\ep)$ is the set of those finite words $w\in \Sigma_j$ whose $\eta'$-root   $w_{|\lf \eta' j\rf} $ belongs to $ \mathcal{E}_\mu(\lf \eta'j \rf,\alpha\pm\ep)$, i.e. 
$$  \frac{\log_2 \mu(I_{w_{|\lf \eta' j\rf}})}{- \lf \eta' j\rf } \in [\alpha-\ep,\alpha+\ep].$$
\sk\item
When  $W \in  \Sigma^*$,  $\mathcal{T}_{\mu, \ell} (j, \eta',\ep,W )$  is the set of those finite words $w\in \Sigma_j$ satisfying $I_w\subset I_W$ and   whose $\eta'$-tail   $\sigma^{ \lf \eta'j  \rf} w $  belongs to $ \mathcal{E}_\mu(j-\lf \eta'j\rf,{H_\ell}(\eta')\pm\ep)$, i.e. 
\begin{equation}
\label{tail}
  \frac{\log_2  \mu(I_{\sigma^{ \lf \eta'j  \rf} w})}{ j - \lf \eta' j\rf } \in [{H_\ell} (\eta')-\ep, {H_\ell}(\eta')+\ep ].
  \end{equation}
 
\sk\item
the sets $\Tl (j,\eta',\ep)$ is the set of all finite words $w\in \Sigma_j$ satisfying \eqref{tail}, so for every $J\leq j$,
$$ \Tl (j,\eta',\ep) = \bigcup_{W\in \Sigma_J} \mathcal{T}_{\mu, \ell} (j, \eta',\ep,W ) .$$

\sk\item
$\Tr (j,\eta',\ep,W)$ and $\Tr (j,\eta',\ep)$ are defined as $\Tl (j,\eta',\ep,W)$ and $\Tl (j,\eta',\ep)$ by replacing ${H_\ell} (\eta')$ by $H_r (\eta')$.
\sk\item
$\Tt (j,\eta',\ep)=\Tl (j,\eta',\ep)\cup\Tr (j,\eta',\ep)$.

\end{itemize}
\end{definition}

\begin{center}
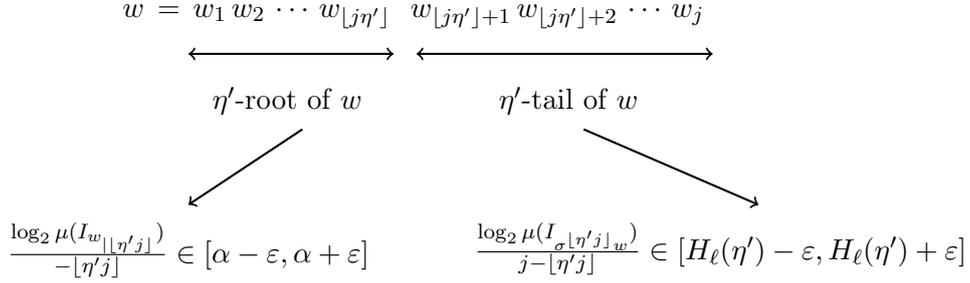
\begin{figure}
\begin{tikzpicture}[xscale=1.,yscale=1.]
\draw  [thick,|-|,color=white] (1.2,1) -- (1.4,0.2)  ;

\draw[fill] (0,0) node  [right]  {$w \, = \, w_1  \, w_2   \,\cdots  \,w_{\lf j \eta'\rf}  \ \  w_{\lf j \eta'\rf +1 }   \,w_{\lf j \eta'\rf +2 } \, \cdots   \, w_j$} ;

\draw [thick,<->] (4,-0.5) -- (7.9,-0.5) ;
\draw [fill] (6,-0.8)   node [below] {$\eta'$-tail of $w$};
\draw [ thick,->]  (6.2,-1.5) -- (8.5,-2.5);
\draw [fill] (8,-2.6)   node [below] {$  \frac{\log_2  \mu(I_{\sigma^{ \lf \eta'j  \rf} w})}{ j - \lf \eta' j\rf } \in [{H_\ell} (\eta')-\ep, {H_\ell}(\eta')+\ep ]$};

\draw [thick,<->] (1,-0.5) -- (3.7,-0.5) ;
\draw [fill] (2.3,-0.8)   node [below] {$\eta'$-root of $w$};
\draw [ thick,->]  (2.5,-1.5) -- (1,-2.5);
\draw [fill] (1,-2.6)   node [below]   {$ \frac{\log_2 \mu(I_{w_{|\lf \eta' j\rf}})}{- \lf \eta' j\rf } \in [\alpha-\ep,\alpha+\ep]$};

\end{tikzpicture}
\caption{Decomposition of a word $ w\in  \Rmu(j,\eta',\alpha\pm\ep) \cap \Tl (j,\eta',\ep)$ into its $\eta'$-tail and its $\eta'$-root.}
\label{figdecomp}
\end{figure}\end{center}

Recall the decomposition \eqref{decompeta'} of any finite word $w$. 
The idea, illustrated by Figure~\ref{figdecomp},  is that the sets $\Rmu(j,\eta',\alpha\pm\ep)$ describe the scaling behavior of the {$\eta'$-root} $w_{|\lf \eta'j \rf}$ of the word $w\in \Sigma_j$, while $\Tt(j,\eta',\ep)$ describe the scaling behavior of the {$\eta'$-tail}  $ {\sigma^{\lf \eta'j \rf}w}$ of $w$. Observe that we focus on the cases where the $\eta'$-tail of $w$ behaves with a local dimension close to some ${H_\ell}(\eta')$ or $H_r(\eta')$. Indeed, these specific behaviors of the $\eta'$-tail  will turn out to be central to explain the structure of the local dimensions of $\M_\mu$, and  Propositions \ref{p1}, and \ref{p2} to \ref{discretization} will exhibit essential properties related to them.

\mk

Observe   that the knowledge of which sets $ \Rmu (j,\eta',\alpha\pm\ep)  $ and $\Tt (j,\eta', \ep)$ a given word $w$ belongs to, yields  {$\mu(I_w)$ up to a multiplicative factor of order $2^{\pm \varepsilon j}$.}

 The first proposition gives upper and lower bounds for the  possible values of $\mu(I_w)$, when $w$ survives after sampling.

 \begin{proposition}\label{p1}
  Almost surely, there exists a positive sequence $(\ep^1_j)_{j\ge 1}$ converging to 0 such that for $j$ large enough, for all $w\in\mathcal{S}_j(\eta)$, one has 
 $$
j({{H_\ell}}(\eta_\ell)-\ep^1_j)\le  -\log_2 \mu(I_w)\le j({H_r}(\eta_r)+\ep^1_j).
$$ 
 \end{proposition}

 \begin{proof} This is a consequence of the large deviations properties of Gibbs measures.   Fix an integer $p\geq 1$. Consider the interval $I_p= [0,{H_\ell}(\eta_\ell)- 2^{-p}] \cup [ H_r(\etar)+2^{-p}, +\infty)$.  By definition of ${H_\ell}$ and $H_r$,    one has $\sup \{ {D_\mu}(h): h\in  I_p \} < d(1-\eta)$.  Let us call $\xi_p = d(1-\eta) - \sup \{ {D_\mu}(h): h\in  I_p \} $.
 
  By  item (5) of Proposition \ref{fm}, there exists a generation   $J_p$ such that $j \geq J_p$ implies
$$\left| \frac{ \log  \# \mathcal{E} _ \mu(j, I_p)   }{-j}  { \log 2^{j}} -   \sup_{h\in I_p } 
 {D_\mu}(h) \right| \leq  \xi_p/2 . $$
 Using the definition of $\xi_p$, this rephrases as
 $$ \# \mathcal{E} _ \mu(j, I_p)   \leq 2^{j( \sup_{h\in I_p } {D_\mu}(h) +\xi_p/2)} \leq 2^{j(d(1-\eta) -\xi_p/2)} .  $$
 
 Let us compute the probability  of the event $\mathscr{A}_j^p = \left\{ \mathcal{S}_j(\eta) \cap \mathcal{E} _ \mu(j, I_p)  \neq \emptyset \right\}$.   One has
 \begin{eqnarray*}
\forall \ j\geq J_p, \ \ \  \mathbb{P}(\mathscr{A}_j^p)   &   \leq  &  1 - (1-2^{-dj(1-\eta)}) ^{\#\mathcal{E} _ \mu(j, I_p ) } \\
  &\leq &  1 - (1-2^{-dj(1-\eta)}) ^{ 2^{j(d(1-\eta) -\xi_p/2)}}\\
  & \leq & 2^{ -j\xi_p/4}. 
  \end{eqnarray*}
The Borel-Cantelli lemma implies that, almost surely, $\mathscr{A}_j^p$ is not realized  when $j$ becomes greater than some integer $J'_p \geq J_p$. In the construction, one can ensure that $J_{p+1}$ is always strictly greater than $J_p$, for all integers $p$.

Choosing now $\ep^1_j = 2^{-p}$ for $ j\in [J_p, J_{p+1})$ yields the result. 
 \end{proof}

The next proposition is complementary to the previous one: it  precisely estimates  the number of surviving vertices with a given $\mu$-local dimension. We state it for the sake of completeness but it will not be used. 
 \begin{proposition}\label{p1prime}
  Almost surely,  for every  $\ep>0$ and every $H\in [{H_\ell}(\etal), H_r(\etar)]$,  there exists $ \widetilde \ep>0$ with $\widetilde\ep \to 0$ as $\ep\to 0$ such that for $j$ large enough one has
 $$
2^{j\big({D_\mu}(H)-d(1-\eta) -\widetilde\ep\big)}\le \#  \left(\mathcal{E}_\mu(j,H\pm\ep) \cap\mathcal{S}_j(\eta)  \right)\leq  2^{j \big ({D_\mu}(H)-d(1-\eta) +\widetilde\ep\big)}.
$$ 
 \end{proposition}
\begin{proof}
The proof is close to that of Proposition \ref{p1} and is omitted.
\end{proof}
Next proposition is crucial. It shows that the parameters $\eta'$, ${H_\ell}(\eta')$ and $H_r(\eta')$ play a special role in our problem. The  underlying idea  is the following: The almost multiplicativity property  implies that for every word $w\in \Sigma$, for every $\eta' \in [\etal,\eta] \cup [\etar,\eta]$, one has 
  $$ \mu(I_w) \approx \mu(I_{w_{| \lfloor \eta' j\rfloor}}) \mu(I_{w_{\sigma ^{\lfloor \eta' j\rfloor} w}}) .$$
But   if a vertex $w$ survives after sampling, i.e. if $ w\in \Sigma_j(\eta)$, then we are going to prove that  $\mu(I_w)$   can be decomposed  as
  $$ \mu(I_w) \approx  \mu(I_{w_{| \lfloor \eta' j\rfloor}})  2^{-(j-\lfloor \eta' j\rfloor)  {H_\ell}(\eta')} \ \mbox{ or } \   \mu(I_w) \approx  \mu(I_{w_{| \lfloor \eta' j\rfloor}})2^{-(j-\lfloor \eta' j\rfloor) H_r(\eta')} , $$
 for some  suitable choice of $\eta'$ (depending on $w$). So we have an explicit formula for its $\eta'$-tail. We will then establish a complementary information (Proposition \ref{p'2}): $\eta'$ being fixed,  with probability one for $j$ large enough,  for $ W \in \Sigma_{\lfloor \eta'j\rfloor}$, there is necessarily at least  one word $w\in\mathcal{S}_j(\eta,W)$ such that   the above decomposition holds.

\begin{center}
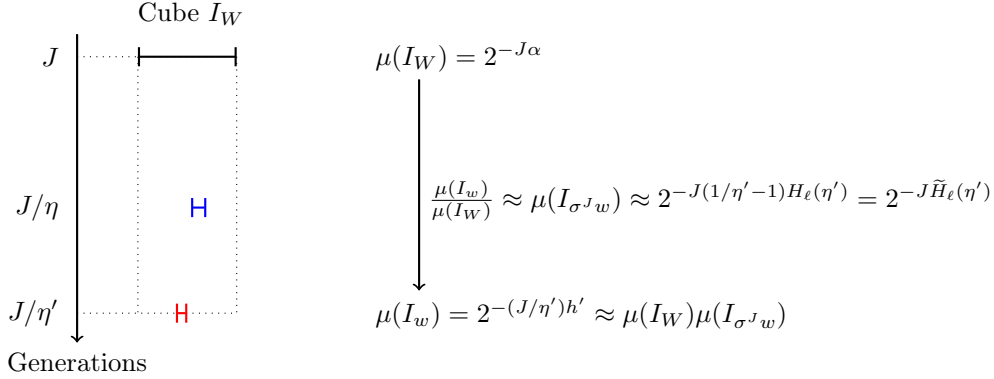
\begin{figure}
\begin{tikzpicture}[xscale=1.,yscale=1.]{\small
\draw [thick,->] (-0.3,2.5) -- (-0.3,-1.59) node [below] {Generations};
\draw [dotted] (-0.3,2.2) -- (1.8,2.2)   ; 
\draw [fill] (-0.3,2.2)  node   [left]   {$J$ \ }  [thick,|-|] (0.5,2.2) -- (1.8,2.2) ;
\draw [fill] (1.2,2.5)  node [above]  {Cube $ I_{W}$}  [fill] (4.7,2.5)  ;
\draw [fill] (3.5,2.2)  node  [right] {$\mu(I_{W }) = 2^{-J  \alpha}$} ;
\draw  [thick,|-|,color=blue] (1.2,0.2) -- (1.4,0.2)  ;
\draw [thick,->] (4.2,1.9) -- (4.2,-0.9)  ;
\draw [fill] (4.2,0.3)  node  [right] {$\frac{\mu(I_{w})}{\mu(I_{W})}  \approx \mu(I_{\sigma^J w})\approx  2^{-J (1/{\eta'}-1) {{H_\ell}}( {\eta'})}= 2^{-J {\Ht_\ell}(\eta')}$} ;
\draw [fill] (-0.3,.2)  node   [left]   {$J /{\eta}$ \ };
\draw [dotted] (0.5,2.2) -- (0.5,-1.2)  [dotted] (1.8,2.2) -- (1.8,-1.2)   ; 
\draw [dotted] (-0.3,-1.2) -- (1.8,-1.2)   ; 
\draw [fill] (-0.3,-1.2)  node   [left]   {$J /{\eta'}$ \ }  [thick,|-|,color=red] (1.0,-1.2) -- (1.15,-1.2)  ;
\draw [fill] (3.5,-1.2)  node  [right] {$\mu(I_{w}) = 2^{-(J /{\eta'}) h'}\approx \mu(I_W)\mu(I_{\sigma^J w})$} ;}
\end{tikzpicture}
\caption{Behavior of the surviving vertices inside a cube $I_W$}
\label{figure_decoup}
\end{figure}\end{center}
 \begin{proposition}\label{p2} 
  With probability one, there exists a positive sequence $(\ep^2_j)_{j\ge 1}$ converging to 0 such that for all $w\in \mathcal{S}_j(\eta)$, there exists $\eta'\in [\eta_{\ell},\eta] \cup [\etar,\eta]$ such that $w\in \Tt (j,\eta', \ep^2_j)$. \end{proposition}
 \begin{proof}  
 We fix $w\in \mathcal{S}_j(\eta)$, and we look for a suitable $\eta'$. See Figure \ref{figure_decoup}.

  Let us denote, for all $j\geq 1$, $\alpha_j := \displaystyle -\frac{\log_2 \mu(I_w) }{j}$, and for all $\eta'\in [0,\eta]$, $\displaystyle\alpha_j(\eta')= -\frac{\log_2 \mu(I_{w_{|\lfloor \eta' j\rfloor}})}{\lfloor \eta' j\rfloor} $ and  $\displaystyle H_j(\eta') = \frac{ -\log_2 \mu(I_{\sigma^{\lfloor \eta' j\rfloor}w})}{j-\lfloor \eta' j\rfloor}$. By the almost multiplicativity property of $\mu$, we have 
 \begin{equation}\label{decomp}
 \alpha_j  j=\alpha_j(\eta')\lfloor \eta' j\rfloor+H_j(\eta') (j-\lfloor \eta' j\rfloor)+O(1),
 \end{equation}
 where $O(1)$ is bounded independently on $w$, $j$ and $\eta'$ (it depends only on the constant $C$ involved in \eqref{quasib}).

 On the other hand, for $\eta,'\eta''\in[0,\eta]$ we have 
  $$H_j(\eta'') (j-\lfloor \eta'' j\rfloor) - H_j(\eta') (j-\lfloor \eta' j\rfloor) = -\log_2 \mu(I_{\sigma^{\lfloor \eta' j\rfloor}w})+\log_2 \mu(I_{\sigma^{\lfloor \eta'' j\rfloor}w}),$$
  which is  bounded above by  $c|\lfloor \eta' j\rfloor -\lfloor \eta''j\rfloor| $ for some constant $c>0$ by \eqref{quasib}. Also, by item (6) of Proposition \ref{fm},    $H_j (\eta')$ and ${H_j}(\eta'')$ are bounded by a constant $K>0$ independently of  $j$, $w$ and $\eta'$. Subsequently, 
\begin{eqnarray*}
| H_j (\eta'')  - H_j (\eta')  |   \!\!  &  \leq  &  \!\!  \left| H_j (\eta'')  - H_j (\eta') \frac{j-\lfloor \eta' j\rfloor}{j-\lfloor \eta'' j\rfloor}  \right|  + H_j (\eta') \left|1- \frac{j-\lfloor \eta' j\rfloor}{j-\lfloor \eta'' j\rfloor}  \right|   \\
  \!\! &  \leq  &   \!\!  (c+K)\frac{|\lfloor \eta' j\rfloor -\lfloor \eta''j\rfloor|  }{j-\lfloor \eta'' j\rfloor}    \leq (c+K) \frac{|\eta''-\eta'|+1/j}{1-\eta}.\end{eqnarray*}
   
  From this inequality, one deduces that there exists a continuous function $\widetilde H_j:[0,\eta] \to \R^+$ such that 
$$s_j=\sup\{|H_j (\eta')-\widetilde H_j(\eta')|:\eta'\in [0,\eta]\}=O(1/j)$$
 independently of $w$ as $j\to\infty$, and  \eqref{decomp} holds with $\widetilde H_j$ instead of  $H_j$.    
 
 \sk
 
$\bullet$ Suppose that $\widetilde H_j(\eta)= H_s $. Since ${{H_\ell}}(\eta) =H_s$,  Proposition \ref{p2} is proved with $\eta'=\eta$.

\sk

$\bullet$ Suppose now that $\widetilde H_j(\eta)< H_s={H_\ell}(\eta)$.  

\sk

- Suppose first that $\eta_\ell=0$. Recall that ${H_\ell}(0)=H_{\min}$.

If $\widetilde H_j(0 ) \le {{H_\ell}}(0) $, then we see that $j\alpha_j  =j \widetilde H_j(0)+O(1)\le  j({{H_\ell}}(0)+O(1/j))$, which due to Proposition \ref{p1} implies that $ {{H_\ell}}(0)-\ep^1_j\le \widetilde{H}_j(0)+O(1/j)\le  {{H_\ell}}(0)+O(1/j)$. Hence \eqref{eq3sss} is satisfied with $\eta'=0$.

 If $\widetilde{H}_j(0)>{{H_\ell}}(0)$, observe that the  mapping $\eta'  \mapsto (\widetilde{H}_j-{{H_\ell}} )(\eta')$ is continuous, positive at $\eta'=0$, negative at $\eta'=\eta$. The continuity  ensures the existence of  $\eta'\in (0,\eta)$ such that $\widetilde{H}_j(\eta')={{H_\ell}}(\eta')$, and \eqref{eq3sss} is satisfied with this $\eta'$.

\sk

-  Suppose now  that $\eta_\ell>0$ and recall that ${{H_\ell}}(\eta')$ ranges in $[{{H_\ell}}(\eta_\ell), H_s]$.  Notice that for any $\eta'$,  $j-\lfloor \eta' j\rfloor\ge j-\lfloor \eta j\rfloor$ which tends to $+\infty$ when $j\to +\infty$. Hence, by Proposition~\ref{p1}, for $j$ large enough we have ${H_j}(\eta')\ge {H}_{\ell}(\etal)-\ep^1_{j-\lfloor \eta' j\rfloor}$, so that  for all $\eta'\in [\eta_\ell,\eta]$,
$$\widetilde{H}_j (\eta')\ge {{H_\ell}}(\eta_\ell)-\ep^1_{j-\lfloor \eta' j\rfloor}-s_j.$$
 By continuity of $\eta'  \mapsto (\widetilde{H}_j-{{H_\ell}} )(\eta')$, there exists $\eta'\in [\eta_\ell,\eta]$ such that  $|\widetilde{H}_j(\eta')-{{H_\ell}}(\eta')|\le \ep^1_{j-\lfloor \eta' j\rfloor}+s_j$. In all   cases, we found $\eta'\in [\eta_\ell,\eta]$ such that $|\widetilde{H}_j(\eta')-{{H_\ell}}(\eta')|\le\ep^1_{j-\lfloor \eta' j\rfloor}+2s_j+O(1/j)$. 
Since $H_j$ and $\widetilde{H}_j$ differ by $o(1)$, the result follows.
\sk

$\bullet$ Finally suppose that $\widetilde{H}_j(\eta)> H_s$. Similar arguments as above yield  $\eta'\in [\eta_r,\eta]$ such that $|{H_j}(\eta')-{H_r}(\eta')|\le \ep^1_{j-\lfloor \eta' j\rfloor}+2s_j+O(1/j)$. We let the reader check the details.

\sk

Since the bound  $\ep^1_{j-\lfloor \eta' j\rfloor}+2s_j+O(1/j)$ tends to 0 uniformly in $\eta'\in [0,\eta]$ as $j\to +\infty$, the sequence $(\ep^2_j := \ep^1_{j-\eta j} +2s_j +O(1/j))_{j\geq 1}$  fulfills the conditions of Proposition \ref{p2}.
  \end{proof}

The previous proposition tells us that every surviving vertex  $w\in \mathcal{S}_j(\eta)$ is such that  its $\eta'$-tail satisfies  either for some $\eta'\in [\etal,\eta]$ (depending on $w$), 
\begin{equation}
 \label{eq3sss}
  (j-\lfloor \eta' j\rfloor)({{H_\ell}}(\eta')-\ep^2_j) \le -\log_2 \mu(I_{\sigma^{\lfloor \eta' j\rfloor}w})\le (j-\lfloor \eta' j\rfloor)({{H_\ell}}(\eta')+\ep^2_j),
\end{equation}
 or for some $\eta'\in [\eta_{r},\eta]$  (also depending on $w$) that
 $$
  (j-\lfloor \eta' j\rfloor)({H_r}(\eta')-\ep^2_j)\le - \log_2 \mu(I_{\sigma^{\lfloor \eta' j\rfloor}w})\le (j-\lfloor \eta' j\rfloor)({H_r}(\eta')+\ep^2_j).
$$
  Next proposition claims that $\eta'$ being fixed in $[\eta_\ell,\eta]$, almost surely, for $j$ large enough, for all $W \in \Sigma_{\lfloor \eta' j\rfloor}$, there is a surviving vertex  $w\in\mathcal{S}_{j}(\eta,W)$  such that   $\dfrac{-\log_2 \mu(I_{\sigma^{\lfloor \eta' j\rfloor}w} )}{j-\lf \eta'j \rf} \approx  {H_\ell}(\eta')$. This property shall be understood as a renewal property for  the local dimensions ${H_\ell}(\eta')$. 
See Figure \ref{figdecomp} for an illustration of this decomposition.

\begin{proposition} 
\label{p'2}
Given $\eta'\in [\eta_\ell,\eta)$, there exists a positive sequence $(\ep^3_j)_{j\ge 1}$ converging to 0 such that, with probability 1,  for $j$ large enough,   for all $W \in \Sigma_{\lfloor \eta' j\rfloor}$, $
 \mathcal{S}_j(\eta,W)  \cap \Tl(j,\eta', \ep^3_j) \neq \emptyset $. 
\end{proposition}

Of course, the same holds true for $  \Tr(j,\eta', \ep^3_j)$, but we do not need this second property.
\begin{center}
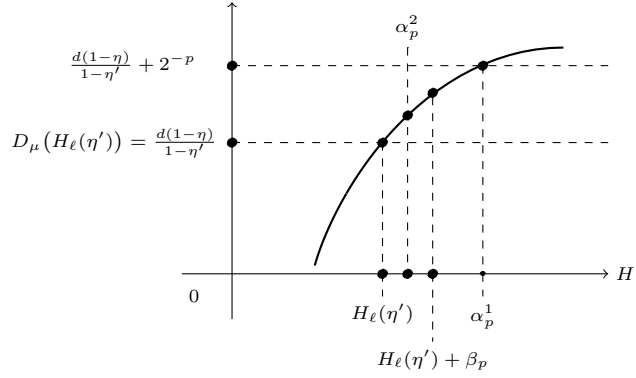
\begin{figure}
\begin{tikzpicture}[xscale=3.3,yscale=3.0]
{\tiny
\draw [->] (0,-0.2) -- (0,1.2) [radius=0.006] node [above] { };
\draw [->] (-0.2,0) -- (1.5,0) node [right] {$H$};
\draw [thick,  domain=1.9:4]  plot ({-(exp(\x*ln(1/5))*ln(0.2)+exp(\x*ln(0.8))*ln(0.8))/(ln(2)*(exp(\x*ln(1/5))+exp(\x*ln(0.8)) ) )} , {-\x*( exp(\x*ln(1/5))*ln(0.2)+exp(\x*ln(0.8))*ln(0.8))/(ln(2)*(exp(\x*ln(1/5))+exp(\x*ln(0.8))))+ ln((exp(\x*ln(1/5))+exp(\x*ln(0.8))))/ln(2)});
\draw [thick, domain=0:1.9]  plot ({-(exp(\x*ln(1/5))*ln(0.2)+exp(\x*ln(0.8))*ln(0.8))/(ln(2)*(exp(\x*ln(1/5))+exp(\x*ln(0.8)) ) )} , {-\x*( exp(\x*ln(1/5))*ln(0.2)+exp(\x*ln(0.8))*ln(0.8))/(ln(2)*(exp(\x*ln(1/5))+exp(\x*ln(0.8))))+ ln((exp(\x*ln(1/5))+exp(\x*ln(0.8))))/ln(2)});

  \draw [fill]  (-0,0.92)  circle  [radius=0.02] node [left] {$\frac{d(1-\eta)}{1-\eta'} +2^{-p} \  \ \ \ $} [dashed] (0,0.92) -- (1.5,0.92) ;
  \draw [fill]  (-0,0.58)  circle  [radius=0.02] node [left] {${D_\mu}\big({H_\ell}(\eta')\big)=\frac{d(1-\eta)}{1-\eta'}\ $} [dashed] (0,0.58) -- (1.5,0.58) ;

\draw[fill] (0.6,0.58) circle [radius=0.02] [dashed] (0.6,0.58) -- (0.6,-0.0) [fill] circle [radius=0.02]  [dashed] (0.6,0.0) -- (0.6,-0.1)  node [below]   {${{H_\ell}}(\eta')$};
\draw[fill] (0.8,0.8) circle [radius=0.02] [dashed] (0.8,0.78) -- (0.8,-0.0) [fill] circle [radius=0.02]  [dashed] (0.8,0.0) -- (0.8,-0.3)  node [below]   {${{H_\ell}}(\eta')+\beta_p $};

\draw[fill] (0.7,0.7) circle [radius=0.02] [dashed] (0.7,1) -- (0.7,-0.0) [fill] circle [radius=0.02]  [dashed] (0.7,1) -- (0.7,1)  node [above]   {$\alpha^2_p$};

\draw[fill] (1,0.92) circle [radius=0.02] [dashed] (1,0.92) -- (1,-0.0) [fill] circle [radius=0.01]  [dashed] (1,0.0) -- (1,-0.1)  node [below]   {$\alpha^1_p$};
  
 \draw [fill] (-0.1,-0.10)   node [left] {$0$};

 }
\end{tikzpicture}
\caption{Relative positions of ${{H_\ell}}(\eta')$, ${{H_\ell}}(\eta')+\beta_p$, $\alpha^1_p$, $\alpha^2_p$.}
\label{figalphap}
\end{figure}\end{center}

\begin{proof}
 
Fix  $\eta'\in [\eta_\ell,\eta)$, which implies that  ${D_\mu}({{H_\ell}}(\eta'))<d=\| {D_\mu}\|_\infty$. For every    integer $p\geq 1$ so large that ${D_\mu} \big({H_\ell}(\eta')\big) +2^{-p} <d$,   let $ \alpha^1_p ,  \alpha^2_p  ,\beta_p$  be such that 
\begin{itemize}
\item 
$\alpha^1_p$ is the  unique real number  in $[{H_\ell}(\etal),H_s]$ such  that ${D_\mu}(\alpha^1_p )={D_\mu}({{H_\ell}}(\eta'))+ 2^{-p} $, 
\item
$\beta_p = (\alpha^1_p - {H_\ell}(\eta'))/2$,
\item  $ \alpha^2_p $ is   such that  $ {H_\ell}(\eta') <  \alpha^2_p < {H_\ell}(\eta')+\beta_p=  \alpha^1_p -\beta_p  $.   
\end{itemize}
Observe that $ {D_\mu}({H_\ell}(\eta' )) < {D_\mu}( \alpha^2_p) < {D_\mu}( \alpha^1_p -\beta_p) < {D_\mu}( \alpha^1_p  )$ (see Figure \ref{figalphap}).

For every integer $p\geq 1$, due to the large deviations properties of $\mu$ (part (5)  of Proposition \ref{fm} and equation \eqref{eq1sss}), we can fix  an integer $j_{p}$ such that for all $j\ge j_{p}$, 

$$ \#\mathcal{E} _\mu(j,\alpha^1_p \pm \beta_p)\ge 2^{j {D_\mu}(\alpha^2_p )}.$$

Using the definition of our parameters,  this implies that
$$\# \mathcal{E}  _\mu(j,{{H_\ell}}(\eta')\pm \widetilde \ep_p )\ge 2^{j \big ({D_\mu}   ({{H_\ell}}(\eta')  )+\widehat\ep_p \big)},$$
 where $\widetilde \ep_p=  3\beta_p   $ and $\widehat\ep_p=  {D_\mu}( \alpha^2_p) -  {D_\mu}({H_\ell}(\eta') ) >0$.
  
  It is clear from the continuity and monotonicity of ${D_\mu}$ that $(\widetilde \ep_p)_{p\geq 1} $ and $(\widehat\ep_p)_{p\geq 1} $ are two positive decreasing sequences, and that $ \lim_{ p \to +\infty} \widetilde \ep_p = \lim_{p\to+\infty} \widehat\ep_p=0$.

For $j\ge j_p/(1-\eta')$ (hence so that $j-\lfloor \eta'j\rfloor\ge j_p$)  and $W \in \Sigma_{\lfloor \eta'  j\rfloor}$,  consider the event 
$$
\mathscr{A}(\eta', \widetilde \ep_p , W )=\Big\{\forall\, w'\in \mathcal{E}_\mu  \big(j-\lfloor \eta'  j\rfloor,{{H_\ell}}(\eta' )\pm \widetilde \ep_p \big), \   p_{Ww' }=0\Big\}.
$$

One has   
\begin{align}
\nonumber\mathbb P \Big(\mathscr{A}(\eta',\widetilde \ep_p, W) \Big)&=(1-2^{-d(1-\eta)j})^{\# \mathcal{E}_\mu \big(j-\lfloor \eta'  j\rfloor,{{H_\ell}}(\eta' ) \pm\widetilde \ep_p \big)}\\
\nonumber&\le \exp\big( {-2^{-d(1-\eta)j}\# \mathcal{E}_\mu(j-\lfloor \eta'  j\rfloor,{{H_\ell}}(\eta' )\pm  \widetilde \ep_p) \big)}\\
\nonumber&\le \exp\big( {-2^{-d(1-\eta)j+(j -\lf \eta' j\rf)\big({D_\mu}  ({{H_\ell}}(\eta')  )+\widehat\ep_p \big)}\big)}.
\end{align}
Recalling that  ${D_\mu} \big({H_\ell}(\eta' )\big)=\frac{d(1-\eta)}{1-\eta' }  $, we get
$$\mathbb P \Big(A( \eta', \widetilde \ep_p,W) \Big)\le\exp (  -2^{(j-\lfloor \eta'j\rfloor ) \widehat\ep_{p}} + O(1)) \leq C \exp (  -2^{(1-\eta')j\widehat\ep_{p}} ).
$$
We choose the sequence $(\ep^3_j)_{j\geq 1}$  as follows:  we first build some sequences of integers by induction. Pick an integer  $p_0$   so large that the previous inequality holds true for   $j\ge j_{p_0}/(1-\eta')$.   Also,  choose $\widetilde j_{p_{0}}>  j_{p_0}$  so large  that for all $j\ge \widetilde j_{p_{0}}/(1-\eta')$,   one has  $C \exp (  -2^{(1-\eta')j\widehat\ep_{p_{0}}} )\le 2^{-d j}$.

Then, assume that   integers     and  $  \widetilde j_{p_0}, ...  , \widetilde j_{p_0+m}$ are found such that   for $n=1,...,m$:
\begin{itemize}
\item 
$\widetilde j_{p_{0}+n}> \max( j_{p_0+n } , \widetilde j_{p_{0}+n-1})$, 

\item
 for $j(1-\eta')\ge \widetilde j_{p_0+n}$ one has $C \exp (  -2^{(1-\eta')j\widehat\ep_{p_0+n}} )\le 2^{-d j}$.
\end{itemize}
Then we choose   $\widetilde j_{p_{0}+{m+1}}>  \max( j_{p_0+m+1 } , \widetilde j_{p_{0}+m})$  so large  that for all $j\ge \widetilde j_{p_{0}+{m+1}}/(1-\eta')$,   one has  $C \exp (  -2^{(1-\eta')j\widehat\ep_{p_{0}+{m+1}}} )\le 2^{-d j}$.

Finally, for every  $j \geq  \widetilde j_{p_0}/(1-\eta')$,  there is a unique integer $m_j$  such that   
\begin{equation}
\label{eq4sss}
\widetilde j_{p_0+m_j}/(1-\eta')\le j<\widetilde j_{p_{0}+{m_j+1}}/(1-\eta'),
\end{equation}
and we set $\ep^3_j = \widetilde \ep_{p_0+m_j}$. By construction we obtain
\begin{equation}
\label{eq1s}
\mathbb P \Big(\mathscr{A}(\eta',\ep^3_j , W) \Big)  \leq    C \exp (  -2^{(1-\eta')j\widehat\ep_{p_0+m_j}} )\le 2^{-dj}.
\end{equation}

Subsequently, 
\begin{eqnarray*}
&& \mathbb P\Big(\left\{\exists\, j\ge \widetilde j_{p_0}/(1-\eta') \mbox{ and } \exists \, W\in \Sigma_{\lfloor \eta' j\rfloor}:  \mathscr{A}(\eta',\ep^3_j ,W) \text{ holds} \right\}\Big )\\
 &\leq &  \sum_{j\ge \widetilde j_{p_0}/(1-\eta')}  \sum_{W\in  \Sigma_{\lfloor \eta' j\rfloor}}   \mathbb P \Big (\mathscr{A}(\eta',\ep^3_j, W) \Big) \\
 & \leq & \sum_{j\ge \widetilde j_{p_0}/(1-\eta')} 2^{d\lf\eta'j\rf}2^{-d j}  <+\infty.
\end{eqnarray*}
We conclude thanks to the Borel-Cantelli lemma. 
\end{proof}

Last proposition can be realized simultaneously on several $\eta' \in [\etal,\eta]$.

\begin{corollary} 
\label{coroetank}
   For all integers $N\ge 1$ and $0\le k\le N-1$, let $\eta_{N,k}=\eta_\ell+ \frac{k}{N}(\eta-\eta_\ell)$. There exists a positive sequence $(\ep^{4,N}_j)_{j\ge 1}$ converging to 0 when $j$ tends to infinity, such that with probability 1, for $N\ge 2$ and  $j$ large enough, for all $0\le k\le N-1$ and all $W\in \Sigma_{\lfloor \eta_{N,k} j\rfloor}$, $ \mathcal{S}_j(\eta,W)  \cap  \Tl(j,\eta_{N,k}, \ep^{4,N}_j) \neq \emptyset $. 
   \end{corollary}

\begin{proof} 
 Fix $N\ge 2$. For each $  k \in \{0,...,N-1\}$, we apply Proposition \ref{p2}, so that we get a sequence $(\ep^3_j(k))_{j\geq 1}$ and a sequence $(\widetilde j_p(k))_{p\geq p_0(k)}$, such that \eqref{eq1s} holds true, i.e. for every $j\geq \widetilde j_{p_0(k)}/(1-\eta')$ and $W \in \Sigma_{\lfloor \eta_{N,k} j\rfloor}$, 
\begin{equation}
\label{eq2s}
\mathbb P\Big (A(\eta_{N,k}, \ep^3_j(k) , W) \Big)  \leq    2^{-d j}.
\end{equation}

Observe that if $0< \ep<\ep'$, $\mathscr{A}(\eta', \ep ' ,W) \subset \mathscr{A}(\eta', \ep ,W)$. Hence,  we can choose  the integer $p=\max(p_0(0),...,p_0(N-1))$, and  the sequences $\ep^{4,N}_j := \max (\ep^3_j(0),...,\ep^3_j(N-1))$ and   $\widetilde j_p := \max( \widetilde j_p(0), ..., \widetilde j_p(N-1) )$, so that  we have the following property:  for all $0\le k\le N-1$, for all $j\ge \widetilde j_{p_0}/(1-\eta)$, for all $ W \in \Sigma_{\lfloor \eta_{N,k} j\rfloor}$,  \eqref{eq2s} holds true with $\ep^{4,N}_j $ instead of $\ep^{3}_j(k) $.

Thus,  
\begin{align*}
& \mathbb P\Big(\left\{\exists\, j\ge \widetilde  j_{p_0}/(1-\eta), \ \exists \, k\in \{0,..., N-1\},\ \exists \, W \in \Sigma_{\lfloor \eta_{N,k} j\rfloor}:  A(\eta_{N,k},\ep^{4,N}_j ,W) \text{ holds} \right\}\Big )\\
&\leq   \sum_{k=0}^{N-1}   \mathbb P\Big(\left\{\exists\, j\ge \widetilde j_{p_0}/(1-\eta) \mbox{ and } \exists \,W \in \Sigma_{\lfloor \eta_{N,k} j\rfloor}:  A(\eta_{N,k},\ep^3_j(k) ,W) \text{ holds} \right\}\Big )  \\
&\leq   \sum_{k=0}^{N-1}    \sum_{j\ge \widetilde j_{p_0}/(1-\eta)}   \sum_{W \in \Sigma_{\lfloor \eta_{N,k} j\rfloor}} \mathbb{P}\Big (  A(\eta_{N,k},\ep^3_j(k) ,W )    \Big  )  \\
&\leq   \sum_{k=0}^{N-1}   2^{-   (d-\eta_{N,k}) \widetilde j_{p_0}/(1-\eta) }     <+\infty.
\end{align*}
The result follows again by the Borel-Cantelli lemma. 
\end{proof}

Next proposition completes the previous corollary by showing (roughly speaking), that for a fixed $W\in \Sigma_J$ with $J$ large enough, for $\eta'$ in some interval $[\eta_0, \eta]$ fixed in advance, the probability to find $w\in  \bigcup_{J/\eta\le j\le J/\eta_0}\mathcal{S}_j(\eta,W)$ with a $\eta'$-tail having a local dimension smaller than ${H_\ell}(\eta')$ decreases exponentially with $J$.  
\begin{proposition}
\label{discretization}
Let 
$\displaystyle \eta_0= \begin{cases}  \frac{H_{\min}}{H_{\min}+{\Ht_\ell}(\widetilde\eta)}  & \mbox{ if } \eta_\ell=0\\ \eta_0=\eta_\ell& \mbox{ if } \eta_\ell>0. \end{cases}$

   For all integers $N\ge 1$ and $k\in \{-1,0..., N-1\}$,  set $\widetilde \eta_{N,k}=\eta- (\eta-\eta_0)\frac{k}{N}$.

For $J\ge 1$ and $W\in\Sigma_J$, consider  the event   $\mathscr{C}(N,J,W)$  defined as 
$$
\mathscr{C}(N,J,W)=\left\{
\begin{split}
&\exists \,   k\in \{-1,0,..., N-1\},   \ \exists \,  j \in [J/\weta_{N,k},J/\weta_{N,k+1}],   \ \\
&\ \  \exists  \, w\in \mathcal{S}_{j}(\weta_{N,k},W) \mbox{ such that } \  \mu(I_{\sigma^{J}w})> 2^{-J ({\Ht_\ell}(\weta_{N,k})+\ep_N)}
\end{split}\right\} ,
$$
with the convention that  $ {\Ht_\ell}(\widetilde\eta_{N,-1})={\Ht_\ell}(\widetilde\eta_{N,0})={\Ht_\ell}(\eta)$. 

With probability one, there exists a positive sequence $(\ep_N)_{N\ge 1}$ converging to 0 such that  for all $N\ge 1$, $J\ge 1$  and $W\in\Sigma_{J}$,  we have $\mathbb{P}\big (\mathscr{C}(N,J,W)\big)\le 2^{-J\ep_N}$.
\end{proposition}

The proof  uses arguments  similar to those developed earlier, and is left to the reader.

\mk

Proposition \ref{p'2} asserts that    for all $W \in \Sigma_{\lfloor \eta' j\rfloor}$, 
$ \mathcal{S}_j(\eta,W)  \cap \Tl(j,\eta', \ep^3_j) $ is not empty when $j$ becomes large. The last proposition of this section shows that its cardinality cannot be very large. This fact will be interpreted geometrically as a  {\em weak redundancy} property from the viewpoint of ubiquity theory \cite{BSubiquity2,BSubiquity3} and has nice geometric consequences for our study. 

\begin{proposition}\label{p4} 
 \begin{enumerate} 
 \item  For all $\eta' \in [\etal,\eta]\setminus \{0\}$, for all $\ep\in(0,1)$, there exists $\beta>0$ such that with probability 1,  for every $j$ large enough and  all $W\in\Sigma_{\lfloor \eta'j \rfloor}$,
\begin{equation}
\label{eq4s}
1  \, \leq \,   \# \big (\,  \mathcal{S}_j(\eta,W)  \cap \Tl(j,\eta',  \beta)  \,\big)  \, \leq \,  2^{\eta'j  \ep}.
\end{equation}
 
 \item
The same holds true for  $\eta'   \in [\etar,\eta]\setminus \{0\}$ and the sets $ \mathcal{S}_j(\eta,W)  \cap \Tr(j,\eta', \beta)  $.
\end{enumerate}
\end{proposition}

\begin{proof}

(1) It is clear that it is enough to get the conclusion for $\ep$ small enough. Fix $\ep\in (0,1)$ and $\eta'\in \eta' \in [\etal,\eta]\setminus \{0\}$.  Due to the almost multiplicativity property of $\mu$, and equation \eqref{eq1sss},  there exists $\beta >0$ and $J_0$ such that for $j\ge J_0$, for each $W \in \Sigma_{\lfloor \eta'j \rfloor}$, 
\begin{equation}\label{eqbeta}
\#  \Tl(j,\eta',  \beta,W)   \leq  2^{ \big ({D_\mu}({{H_\ell} }(\eta')  )+d\ep^2\big)(j-\lfloor \eta'j  \rfloor)}.
\end{equation}
Notice that the cardinality $ n_j=\# \Tl(j,\eta',  \beta,W) $ is independent of $W$. Since ${D_\mu}({{H_\ell} }(\eta'))=d(1-\eta)/(1-\eta')\le d$, $\ep< 1$ and $\eta' \le \eta<1$, for $j\ge J_0$ we have 
$$
n_j\le 2^{ \big ({D_\mu}({{H_\ell} }(\eta')  )+d\ep^2\big)(j-\lfloor \eta'j  \rfloor)}\le 2^{d(1-\eta)j}2^{d\ep^2 j+d}.
$$

 By definition, we have 
 $$ \# \big (\,  \mathcal{S}_j(\eta,W)  \cap \Tl(j,\eta', \ep^3_j)  \,\big)  =\sum_{w\in   \mathcal{T}_{\mu, \ell} (j, \eta',\beta,W) } p_w.$$
 Denote this random variable by $B(j,\eta',\beta,W)$. Its law is  a binomial law of parameters $(n_j, 2^{-d(1-\eta)j} )$.  Thus 
\begin{eqnarray*}
\mathbb P \Big (B(j,\eta',\beta,W)\ge 2^{\ep\eta' j} \Big) &  \le  & \sum_{2^{\ep\eta' j}\le l\le n_j} \binom{n_j}{l}(2^{-d(1-\eta)j} )^l  \le  \sum_{2^{\ep\eta' j}\le l\le n_j} \frac{(n_j2^{-d(1-\eta)j} )^l}{l!}\\ & \le & \sum_{2^{\ep\eta' j}\le l\le n_j} \frac{2^{dj\ep^2 l+dl}}{l!}\le \sum_{l\ge 2^{\ep\eta' j}} \left (\frac{e2^{dj\ep^2 +d}}{l} \right )^l 
\end{eqnarray*}
for $j$ large enough by Stirling's formula. Then, if  $ \ep \leq\eta'/(4d)$,   there is another integer $J'_0$ such that for $j \geq J'_0$, for all  $l\ge2^{\ep\eta' j}$, we have $\dfrac{e2^{dj\ep^2 +d}}{l}\le 2^{-\ep\eta' j/2}\le 1/2$, hence 
$$
 \mathbb P\Big(B(j,\eta',\beta,W)\ge 2^{\ep\eta' j} \Big)\le 2\cdot 2^{-\lfloor 2^{\ep\eta' j}\rfloor\ep\eta' j/2},
 $$
 and 
$$
\sum_{j\ge J'_0} \  \sum_{W \in\Sigma_{\lfloor \eta'j \rfloor}}\mathbb P\Big(B(j,\eta',\beta,W)\ge 2^{\ep\eta' j} \Big)\le  \sum_{j\ge J'_0} 2^{d\lfloor \eta'j \rfloor}2\cdot 2^{-\lfloor 2^{\ep\eta' j}\rfloor\ep\eta' j/2}<\infty.
$$
The desired conclusion follows from the Borel-Cantelli lemma.  
 
 \mk
 
 \noindent 
 (2) The computations are identical for $\eta'\in [\eta_r,\eta]\setminus\{0\}$ and $\# \big (\,  \mathcal{S}_j(\eta,W)  \cap \Tl(j,\eta', \beta)  \,\big) $.
\end{proof}
\section{Upper bound for the singularity spectrum of $\M_\mu$}\label{sec_upper}

Section~\ref{upb3} derives the sharp upper bound provided by Theorem~\ref{thm-0} for the decreasing part of $D_{\M_\mu}$; this bound comes rather directly after the preparation achieved in Section~\ref{anaval}. Next,  Section~\ref{upb1} examines in detail the possible local behaviors of the surviving coefficients $\mu(I_w)$ which contribute to a given set $\underline E_{\M_\mu}(H)$ and provides a first expression for the upper bound of  the increasing part of $D_{\M_\mu}$. This bound is then simplified in Section~\ref{upb2} into the formula given by Theorem~\ref{thm-0}.  Also, precious information are pointed out in preparation of next Section~\ref{sec_lower}, which deals with the lower bound for   $D_{\M_\mu}$.

  \begin{remark}
  This section is rather long and technical, but it is key to understand what phenomena rule the local behavior of $\Mm$ at a point $x\in \zu^d$.
  Let us mention that there is a slightly faster way to obtain the upper bound for the multifractal spectrum of $\Mm$, using the multifractal formalism and a lower bound for the $L^q$-spectrum $\tau_{\Mm}$ obtained in Section~\ref{sec92}. Nevertheless, we choose to keep up with the first method, mainly for two reasons:
  \begin{itemize}
  \item
  The second method is ``blind'', since it does not give any clue on how to obtain the lower bound for the spectrum. Indeed, it will appear soon that for every possible local dimension $H$, there is a favorite scenario which leads a point $x$ to satisfy $\underline{\dim}(\Mm,x)=H$. This can absolutely not be guessed without the precise study achieved in the following pages.
  \item
  Using the multifractal formalism to get an upper bound for the multifractal spectrum is efficient only when the multifractal formalism is satisfied by the object under consideration. Fortunately, this is the case for $\Mm$, and the bound is sharp. But there are closely related sampling processes (not developed in this paper) not satisfying the multifractal formalism, and in this case this method is useless.   
  \end{itemize}
  
  \end{remark}

\subsection{Upper bound for the decreasing part of $D_{\M_\mu}$}\label{upb3}

It turns out that finding an upper bound in the decreasing part of the singularity spectrum $D_{\M_\mu}$, i.e. for $H\geq H_s+{\Ht_\ell}(\weta)$, is much easier than in the increasing~one.

\mk

We start with quite a direct  upper bound for all the local dimensions of $\M_\mu$.
\begin{proposition}
\label{majexp}
Almost surely, for every $x\in [0,1]^d$, 
$$
 {\underline \dim_\locloc}(\M_\mu,x) \le  {\overline \dim_\locloc}(\M_\mu,x) \leq \ {\overline \dim_\locloc}(\mu,x)  +{\Ht_\ell} (\widetilde \eta).
$$
As a consequence, for every $x\in \zu^d$, $ {\underline \dim_\locloc}(\M_\mu,x) \leq H_{\max}+{\Ht_\ell}(\weta)$.
\end{proposition}

\begin{proof}
Let $x\in [0,1]^d$. Due to Proposition~\ref{p'2} applied with $\eta'=\widetilde\eta$, for each $j$ large enough we have 
$$\M_\mu (I_{\lfloor j\widetilde \eta\rfloor}(x))\ge C^{-1} \mu (I_{\lfloor j\widetilde \eta\rfloor}(x)) 2^{-(j-\lfloor j\widetilde \eta\rfloor)({{H_\ell}}(\widetilde\eta)+\ep^3_j)}.$$
Taking logarithm on both sides,   dividing by $-\lfloor j\widetilde \eta\rfloor\log(2)$, and  taking the $\liminf$ as $j\to\infty$ yields the desired conclusion. 

Since for every $x\in \zu^d$, $ {\overline \dim_\locloc}(\mu,x) \leq H_{\max}$, the result follows.
\end{proof}

Using this upper bound for the local dimensions, {and by anticipation the lower bound given by Lemma~\ref{lempointgauche} below}, one deduces that the domain of $D_{\M_\mu}$ is included in $[{H_\ell}(\etal), H_{\max} +{\Ht_\ell}(\weta)]$. Also, we get an upper bound for the decreasing part of the singularity spectrum.

\begin{proposition}
For all $H\in[ {H_s}+ {\Ht_\ell} (\widetilde \eta), H_{\max}+{\Ht_\ell}(\weta)]$, one has 
$$D_{\M_\mu} (H)\le {D_\mu}\big (H- {\Ht_\ell} (\widetilde \eta)\big),$$
and for all $H> H_{\max}+{\Ht_\ell}(\weta)$, 
$$\dim   \underline E^{\geq}_{\M_\mu}(H')=-\infty.$$
\end{proposition}

\begin{proof}
Let $H\ge {H_s}+{\Ht_\ell} (\widetilde \eta)$. By Proposition~\ref{majexp}, if $x\in \underline{E}_{\M_\mu}(H)$, then $ {\overline \dim_\locloc}(\mu,x)   \ge    H-{\Ht_\ell} (\widetilde \eta)$, hence $\underline{E}_{\M_\mu}(H)\subset  \overline{E}^{\geq}_\mu({  H-{\Ht_\ell} (\widetilde \eta)})$. Using Part (3) of Proposition \ref{fm}, one deduces that 
$$\dim \underline{E}_{\M_\mu}(H)\le {D_\mu}\big( H-{\Ht_\ell} (\widetilde \eta)\big),$$
 since  $ H-{\Ht_\ell} (\widetilde \eta) \ge {H_s}$ (this corresponds to  the decreasing part of $D_\mu$). 
\end{proof}

\subsection{Upper bound for the increasing part of $D_{\M_\mu}$}\label{upb1}

Let us start with  the lower bound for the left end-point of the support of the singularity spectrum of $\M_\mu$.

\begin{lemma}
\label{lempointgauche}
With probability 1, for every $x\in \zu^d$, ${\underline  \dim}(\M_\mu,x)\geq {H_\ell}(\etal)$.
\end{lemma}

\begin{proof}
By Proposition \ref{p1}, with probability 1, for $j$ large enough, the surviving vertices $w\in \mathcal{S}_j(\eta)$ all satisfy $\mu(I_w) \leq 2^{-j ({H_\ell}(\etal)-\ep^1_j)}$. Hence, for every large integer $J$ and every word $W\in \Sigma_{J}$, $\M_\mu(I_{W}) \leq 2^{-J({H_\ell}(\etal)-\ep^1_J)}$, since 
$\M_\mu(I_{W})$ is the  maximum of $\mu(I_w)$ over all surviving words $w$ such that $I_w\subset I_{W}$.  Subsequently, for every $x$, ${\underline  \dim}(\M_\mu,x)\geq {H_\ell}(\etal)$.
\end{proof}

Further, we are going to  provide a first expression for the sharp upper bound of $\dim \underline E_{\M_\mu}$ when $H_\ell(\etal)\leq H\le H_s+{\Ht_\ell}(\widetilde\eta)$ in Proposition~\ref{ubip}. 
It is based on the following definition and Proposition~\ref{propmajs}, which describe the possible scenarii leading to the property ${\underline \dim_\locloc}(\M_\mu,x) \le H$.  

\begin{definition}
\label{defbigsets}
 For each  $j\ge 1$, $k\ge 1$,   $\alpha\in\R_+$, $\eta'\in [\eta_\ell,\eta]\setminus\{0\}$ and $\delta\ge 1$, let 
\begin{eqnarray*}
\vspace{2mm}&& {F}_{\mu,\ell}(j,\alpha,\eta',\delta ,k) = \vspace{2mm}\\
 \vspace{2mm}&&\left\{ \,  x\in \zu^d:  \begin{cases}  \   \exists  \, w \in \mathcal{S}_j(\eta) \, \cap \,  \Rmu (j,\eta',\alpha\pm1/k) \, \cap  \, \Tl (j,\eta', 1/k) \vspace{1mm} \\ \      \mbox{such that }   \max( 2^{-j},d(x, I_w))\le  2^{-\eta'j  \delta}   \end{cases} \right\}.
\end{eqnarray*}

 \noindent 
Let  $\mathcal{P}_\ell$ be a countable set of parameters $(\alpha,\eta',\delta) $ dense in $[H_{\min},H_{\max}]\times  (\etal,\eta] \times [1,+\infty)$.

\noindent Then, for $H\ge 0$, let 
\begin{equation}
\label{defFmul}
 F_{\mu,\ell} (H)=\bigcap_{\ep\in(0,1)}  \ \bigcap_{k\ge 1} \ \bigcup_{\substack { (\alpha,\eta',\delta) \in \mathcal{P}_\ell: \\
 \delta \in [1,1/\eta'],\ 
 \frac{\alpha+{\Ht_\ell}(\eta')}{\delta}\le H+\ep } }\limsup_{j\to+ \infty} F_{\mu,\ell}(j,\alpha,\eta',\delta -\ep,k)  .
\end{equation}

\medskip

The sets  $F_{\mu,r}(j,\alpha,\eta',\delta,k) $,   $\mathcal{P}_r$ and $ {F}_{\mu,r} (H)$ are similarly defined.

\end{definition}

 The definition of the sets $F_{\mu,\ell}(j,\alpha,\eta',\delta,k)$  and   $F_{\mu,r}(j,\alpha,\eta',\delta,k) $ is rather long and difficult to handle with at first sight. Nevertheless, it is the mathematical counterpart of the following intuition: $F_{\mu,\ell}(j,\alpha,\eta',\delta,k)$ contains those points $x$ which are $2^{-\eta'j  \delta}$ close to some cube  $I_w$  (i.e.  $ \max( 2^{-j},d(x, I_w))\le  2^{-\eta'j  \delta}  $)  associated with a surviving vertex $w$ whose {$\eta'$-root} and {$\eta'$-tail} have a prescribed behavior with respect to the capacity $\mu$. 

The fact that these sets play a key role is illustrated by the following proposition, which is the  main result of this section.

\begin{proposition}
\label{propmajs}
With probability 1, for all $H\ge 0$,  
$$\underline{E}_{\M_\mu}(H)\subset F_{\mu,\ell}(H)\cup F_{\mu,r}(H).$$
\end{proposition}

\begin{proof}
Fix $H\ge 0$ and $x\in \underline E_{\M_\mu}(H)$. We denote by $\M_j(x) $ the coefficient $\M_\mu(I_{x_{|j}})$, and $\mathcal{N}_{j}(x)$ the set $\mathcal{N}(x_{|j})$ (recall Definition \ref{defnewsets}). 
Let $\ep\in (0,1)$.

By definition, since $\underline\dim (\Mm,x) =H$ there is an infinite number of integers $J_n$ such  that   $ \frac{-\log_2 \M_{J_n}(x)}{J_n}\le H+\ep/2$.   We write $W_n= x_{|J_n}$. By definition of the quantity $\M_{J_n} (x)  = \M_\mu(I_{W_n})$, there exists a surviving vertex   $w_n \in \mathcal{S}_{j_n}(\eta)$, $j_n \ge J_n$, such that  $w_n \in \mathcal N_{J_n}(x)$  and  $\M_{J_n}(x)=\mu(I_{w_n})$. One can assume that $I_{w_n}\subset I_{W_n}$, the other cases, i.e. when  $I_{w_n}$ is included in a neighboring cube $I_{w'}$ of $I_{W_n} $  of generation $J_n$, are absolutely similar by switching   $W_n$ and $w'$).

By Proposition~\ref{p2}, there are $\eta_{j_n} \in[\etal,\eta]\cup[\etar,\eta]$ and  $\alpha_{\lfloor j_n \eta_{j_n} \rfloor}\in[{H_{\min}},{H_{\max}}]$ such that  
\begin{enumerate}
\sk\item either $\eta_{j_n} \in[\etal,\eta]$, and 
\begin{align}
\label{r1}-\log_2 \mu(I_{w_n})&=\alpha_{\lfloor j_n \eta_{j_n} \rf } \lfloor j_n\eta_{j_n} \rfloor+ (j_n-\lfloor j_n\eta_{j_n} \rfloor) H_{ \ell}(\eta_{j_n}) +o(j_n),\\
\label{r2}-\log_2  \mu(I_{\sigma^{\lfloor j_n\eta_{j_n} \rfloor}w_n})&= (j_n-\lfloor j_n \eta_{j_n } \rfloor) {H_{\ell }}(\eta_{j_n}) + o(j_n),
\end{align}  
 
\sk
\item or 
$\eta_{j_n} \in[\etar,\eta]$, and the same holds with $H_r(\eta_{j_n})$ instead of ${H_\ell}(\eta_{j_n})$.
\end{enumerate}

Obviously,  $\max(2^{-j_n},d(x,I_{w_n}))\le  2 \cdot 2^{-{J_n}}$ (recall that we work with the $\|.\|_\infty$ norm).

\begin{center}
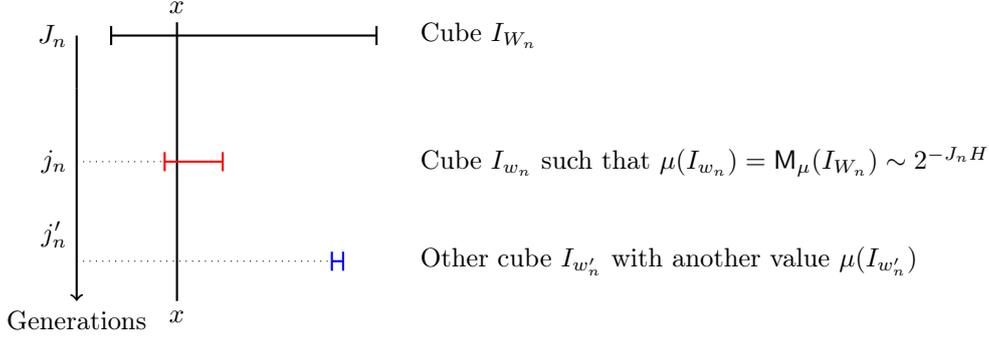
\begin{figure}
  \begin{tikzpicture}[xscale=0.88,yscale=0.88]
{\small
\draw[thick] (-0.5,4) -- (-0.5,3.2);
\draw [thick,->] (-0.5,3.2) -- (-0.5,-0.) node [below] {Generations};
\draw [fill] (1,4.2)  node   [above]   {$x$ }   [thick,-] (1,4.2) -- (1,-0.) node [below] {$x$};
\draw    [thick,|-|]   (0,4) -- (4 ,4)  ;
 \draw [fill] (-0.5,2.1)  node   [left]   {{$j_n $} }     [dotted] (-0.5,2.1) --  (.8,2.1)    ; 
\draw   [thick,color=red,|-|] (.8,2.1) -- (1.7,2.1)  ;
\draw   [thick,color=blue,|-|] (3.3,0.6) -- (3.5,0.6)  ;
\draw [fill] (4.5,4)  node  [right] {Cube $I_{W_n}$} ;
\draw [fill] (-0.5,1.)  node   [left]   {$j'_n$ }     [dotted] (-0.5,0.6) --  (3.4,0.6)    ; 
\draw [fill] (4.5,2.1)  node  [right] { Cube $I_{w_n}$ such that $\mu(I_{w_n}) = \M_\mu(I_{W_n}) \sim 2^{-J_n H} $};
\draw [fill] (4.5,0.6)  node  [right] {Other cube $I_{w'_n}$ with another value $ \mu(I_{w'_n})$   } ;  
 \draw [fill] (-0.5,4)  node   [left]   {$J_n$ } ;      
}
\end{tikzpicture}    
\caption{Competition between  surviving vertices  }
\label{figure_decoup2}
\end{figure}
\end{center}

$\bullet$ Assume that  the subsequence $(\eta_{j_n})_{n\geq 1}$ has an accumulating point in $(0,\eta]$. 

Without loss of generality we can assume that the whole  sequence $(\eta_{j_n})_{n\geq 1}$ converges to some $\eta'\in [\etal,\eta]\setminus\{0\}$ and that we are in the above situation above (1) for infinitely many values of $j_n$, along which   $H_{\ell }(  \eta_{j_n}  )$ converges to  $H_{\ell }(  \eta' )$. The other situation, absolutely symmetric,  is that  the    sequence $(\eta_{j_n} )_{n\geq 1}$ converges to some $\eta'\in [\etar,\eta]\setminus\{0\}$, the situation (2) holds for $n$ large enough,  and  $H_{r }(  \eta_{j_n}  )$ converges to $H_{r}(\eta')$. 

\begin{lemma}
One has
$\lim_{n\to+\infty} j_n \eta_{j_n}=+\infty$, and  for $n$ large enough, one has
\begin{equation}
\label{eq7s}
j_n \eta_{j_n} \le J_n(1+\ep/6).
\end{equation}
\end{lemma}

\begin{proof}
The first part $\lim_{n\to\infty} j_n \eta_{j_n }= +\infty$ is obvious since $\lim_{n\to\infty}\eta_{j_n} = \eta' >0$.

Recall  Corollary \ref{coroetank} and the notations therein. For $N\ge 1$, let us write $\eta_{j_n}(N) $  for the  unique $\eta_{N,k}$ such that  \begin{equation}
\label{eq8s}
  \eta_{j_n}  \in \big [ \eta_{N,k} , \,  \eta_{N,k} +(\eta-\etal)/N \big).
\end{equation}
  We apply Corollary \ref{coroetank}  with an integer  $ j '_n >J_n$ such that  $\lf {j'_n}   \eta_{j_n}(N) \rf  =   J_n $, to the word $W_n   \in \Sigma_{ J_n }$: there exists $w'_n \in \mathcal{S}_{j'_n} (\eta)$ such that  $I_{w'_n}\subset I_{W_n }$  and  $w'_n \in \Tl(j'_n, \eta_{j_n}(N), \ep^{4,N}_{j'_N})$, i.e.
$$
2^{-(j'_n- J_n ) \big ({{H_\ell}}   ( \eta_{j_n}(N)   )+\ep^{4,N}_{j'_n }  \big )}\le \mu(I_{\sigma^{\lfloor {j'_n } \eta_{j_n} (N)   \rfloor}  w'_n })\le 2^{-( j' _n - J_n    ) \big ({{H_\ell}}(\eta_{j_n}(N)  )-\ep^{4,N}_{j'_n  }  \big )},
$$
see Figure \ref{figure_decoup2}. 
 Observe  that $I_{w'_n} \subset I_{W_n}$  (hence $(w'_n)_{|\lf{j'_n } \eta_{j_n} (N) \rf} =(w'_n)_{|J_n} = W_n$), so 
\begin{eqnarray}
\label{eq5sss}
\mu(I_{w'_n }) &\geq & C^{-1}\mu(I_{ {w'_n }_{|J_n }})\mu (I_{\sigma^{J_n } w'_n})  \geq  C^{-1} \mu(I_{W_n  })  2^{-( {j'_n} - J_n) \big({{H_\ell}}(\eta_{j_n }(N)   )+\ep^{4,N}_{  {j'_n} } \big  )}.
\end{eqnarray}

Assume  towards contradiction that \eqref{eq7s} is not true. We are going to  prove that $  \mu(I_{  w'_n})>\mu(I_{w_n})$, contradicting the maximality of $\mu(I_{w_{n}})$ and the  fact that $\M_{J_n} (x)=\mu(I_{w_n})$. When \eqref{eq7s} does not hold, one has
\begin{eqnarray*}
  {j'_n}   - J_n     &  \leq     \dfrac{J_n+1}{ \eta_{j_n}(N)}  - J_n  = J_n \left(  \dfrac{1+1/J_n}{ \eta_{j_n} (N)}  - 1\right ) <  \dfrac{ {j_n} \eta_{j_n}  }{ 1+\ep/6}  \dfrac{1+1/J_n -\eta_{j_n}(N) }{ \eta_{j_n}(N)}  .
\end{eqnarray*}
Moreover, since $\eta_{j_n}$ tends to $\eta'$ and $|\eta_{j_n}(N)-\eta_{j_n}|\le 1/N$, we  choose $N$ so large that   that when $n$  becomes large   we have 
$$
  {j'_n}   -J_n  \le \frac{{j_n} -\lfloor j_n \eta_{j_n}\rfloor}{1+\ep/12}.
$$   
Observe that $ j_n \eta_{j_n}  \geq J_n(1+\ep/6)  $ also  yields  $I_{ {w_n}_{|\lfloor  j_n \eta_{j_n} \rfloor}} \subset I_{ W_n}$. This, together with  the last inequality and  \eqref{eq5sss} yields
\begin{eqnarray*}
\mu(I_{ w'_n}) & \ge  &   
 C^{-1} \mu(I_{ {w_n}_{|\lfloor  j_n \eta_{j_n} \rfloor}}) 2^{- \frac{j_n-\lfloor j_n \eta_{j_n} \rfloor}{1+\ep/12}  \big({H_\ell}( \eta_{j_n}(N))+\ep^{4,N}_{  {j'_n} }  \big)}.
\end{eqnarray*}

One chooses now the integers $N$ and $n$ so large that:
\begin{itemize}
\item
 $|{H_\ell}( \eta_{j'_n} ) - {H_\ell}( \eta')|< {H_\ell}( \eta_{j_n} )  \ep/96$,
 
 \item
 $|{H_\ell}( \eta_{j_n}(N))-{H_\ell}( \eta_{j_n})|<{H_\ell}( \eta_{j_n} )  \ep/96$,

\item
$\ep^{4,N}_{  {j'_n} }  <  {H_\ell}( \eta_{j_n} )  \ep/96$.

\end{itemize} 
This is possible since  $\eta_{j'_n} \to \eta' >0$, ${H_\ell}(\eta')>0$,  and ${H_\ell}$ is differentiable as function of $\eta'\in (\eta_\ell, \eta)$.
With these choices, one sees that  
\begin{eqnarray*} \frac{j_n-\lfloor j_n \eta_{j_n} \rfloor}{1+\ep/12}  \big({H_\ell}( \eta_{j_n}(N))+\ep^{4,N}_{  {j'_n} }  \big)  & \leq & (j_n-\lfloor j_n \eta_{j_n} \rfloor   ){H_\ell}( \eta_{j_n} )  (1-\ep/48).
 \end{eqnarray*}

Finally, we use \eqref{r2} to get, for $n$ large enough 
\begin{eqnarray*}
\mu(I_{ w'_n}) & \ge  &   
 C^{-1} \mu(I_{ {w_n}_{|\lfloor  j_n \eta_{j_n} \rfloor}}) 2^{-(j_n-\lfloor j_n \eta_{j_n} \rfloor   ){H_\ell}( \eta_{j_n} )  (1-\ep/48)}\\
 & \ge  &   
 C^{-2 } \mu(I_{ {w_n}_{|\lfloor  j_n \eta_{j_n} \rfloor}}) \mu(I_{\sigma^{\lfloor j_n\eta_{j_n} \rfloor}w_n}) 2^{(j_n-\lfloor j_n \eta_{j_n} \rfloor   ){H_\ell}( \eta_{j_n} )  \ep/48 +o(j_n)} \\
   & \ge  & C' \mu(I_{w_n})  2^{(j_n-\lfloor j_n \eta_{j_n} \rfloor   ){H_\ell}( \eta_{j_n} )  \ep/96}.
\end{eqnarray*}

When $n$ becomes large, we conclude that  $  \mu(I_{  w'_n})>\mu(I_{w_n})$, hence a contradiction.
 \end{proof}
 
Now we prove that $x$ is well-approximated by some cubes $I_w$ whose $\eta'$-tail and $\eta'$-root have a controlled local behavior with respect to $\mu$.

By construction and the last lemma,  since $\eta_{j_n}$ tends to $\eta'$ when $n \to +\infty$, we have 
$$2^{-j_n} \leq \max(2^{-j_n},d(x,I_{w_n}))\le  2\cdot 2^{- J_n }\le 2\cdot 2^{- \lfloor j_n \eta_{j_n} \rfloor/(1+\ep/6)}\le 2^{-  j_n\eta' (1-\ep/3)}.$$
Consequently, for $n$ large enough  we can write   $\max(2^{-j_n},d(x,I_{w_n}))=  2^{-j_n \eta'(\delta_{j_n}-\ep/2)}$ for some $\delta_{j_n}\in [1,1/\eta']$.  

Up to extraction of a subsequence, the sequences $(\delta_{j_n})_{n\geq 1}$ and $(\alpha_{\lfloor j_n \eta_{j_n} \rf })_{n\ge 1}$  can be assumed to converge to some $\delta\in [1, 1/\eta']$ and  $\alpha\in [{H_{\min}},{H_{\max}}]$. In addition, by construction we have $\M_{\lfloor j_n  \eta' \delta_{j_n} \rfloor}(x)=\M_{J_n}(x)=\mu(I_{w_n} )$ and $J_n \le \lfloor j_n \eta '\delta_{j_n}+ 1\rfloor  $. Thus, recalling \eqref{r1}, one finally gets 
$$
\frac{\alpha_{\lfloor{j_n} \eta_{j_n} \rf } \lfloor {j_n}\eta_{j_n} \rfloor+ ({j_n}-\lfloor {j_n}\eta_{j_n} \rfloor) H_{ \ell}(\eta_{j_n}) +o({j_n}) }{\eta _{j_n} {j_n}  \delta_{j_n}}  \le  \frac{-\log_2  \M_{J_n}(x)}{J_n}\le H+\ep/2,
$$
which can also be written
$$
\alpha_{\lfloor{j_n}  \eta_{j_n} \rf }\frac{ \lfloor {j_n}\eta_{j_n} \rfloor}{\eta _{j_n} {j_n}  \delta_{j_n}}  +  \frac{({j_n}-\lfloor {j_n}\eta_{j_n} \rfloor)  }{\eta _{j_n} {j_n}  \delta_{j_n}} (H_{ \ell}(\eta_{j_n}) +o(1))\le   H+\ep/2.
$$
Finally,   fix an integer $k\geq 1$. Choosing $\ep$ such that  $\ep<1/(4k)$  and  a triplet $(\widetilde\alpha, \widetilde\delta,\widetilde\eta')$ in the dense set $\mathcal{P}_\ell$  so that  for $n$ large  $\|(\alpha_{j_n}, \delta_{j_n},\eta_{j_n}) - (\widetilde\alpha, \widetilde\delta,\widetilde\eta') \|_\infty \leq \ep/4$,  we see that  when  $n$ becomes large, we ensured that:
\begin{itemize}
\item
$\widetilde\delta \in [1,1/\eta']$,  
\item
$\frac{\widetilde\alpha+ {\Ht_\ell}(\weta')}{\widetilde\delta} =\frac{\widetilde\alpha+ (1/\widetilde\eta'-1) {H_\ell}(\weta')}{\widetilde\delta} \le H+\ep$,

\item
$\max(2^{-j_n},d(x,I_{w_n}))=  2^{-j_n\widetilde \eta'(\widetilde\delta -\ep)}$,
\item

\item
$\left| \dfrac{-\log_2 \mu(I_{(w_n)_{|\lf  {j_n }\widetilde \eta'\rf} })}{\lfloor {j_n} \widetilde \eta'  \rfloor} - \widetilde \alpha   \right |\le 1/k$, hence $w_n \in \Rmu(j, \widetilde \eta',\widetilde \alpha\pm 1/k)$,

\item
$ \left | \dfrac{-\log_2 \mu \left (I_{\sigma^{\lfloor  {j_n}\widetilde \eta' \rfloor} w_n} \right)}{j_n -\lfloor j_n\widetilde \eta' \rfloor} -{{H_\ell}}(\widetilde \eta') \right|\le 1/k $, hence $w_n \in \Tl(j, \widetilde \eta', 1/k)$.
  \end{itemize}
Recalling Definition \ref{defbigsets} and equation \eqref{defFmul}, this precisely  shows that  $x$ belongs to the set $ {F}_{\mu,\ell}(j_n,\widetilde\alpha,\widetilde\eta',\widetilde\delta-\ep,k) $ with $\widetilde \delta \in[1,1/\eta']$ and $\frac{\widetilde\alpha+ {\Ht_\ell}(\weta')}{\widetilde\delta}  \le H+\ep$, as claimed in the initial statement.  
 
 \mk
 
 $\bullet$ Assume now  that $(\eta_{j_n})_{n\geq 1}$ converges to $0$ (this implies that $\etal=0$). Fix a small $\eta'>0$, and rewrite \eqref{r1} and \eqref{r2} as
\begin{align}
\nonumber-\log_2 \mu(I_{w_n} )&=\alpha_{j_n} \lfloor j_n \eta'\rfloor+ (j_n-\lfloor j_n \eta'\rfloor) {{H_\ell}}(\eta') +o(j_n)+j_n \xi_1(j_n,\eta'),\\
\nonumber -\log_2 \mu(I_{\sigma^{\lfloor j_n \eta'\rfloor}w_n })&= (j_n -\lfloor j_n \eta'\rfloor) {{H_\ell}}(\eta') +o(j_n)+j_n \xi_2(j_n ,\eta'),
\end{align}
where both $\xi_1(j_n,\eta')$  and $\xi_2(j_n ,\eta')$ are $O\left(\left|{{H_\ell}}(\eta')-{{H_\ell}}(\eta_{j_n})\right|+\left|\eta'\right|\right)$. If one chooses $\eta'$ so small that $\limsup_{n\to +\infty} |{{H_\ell}}(\eta')-{{H_\ell}}(\eta_{j_n} )|+\eta'\le |{{H_\ell}}(\eta')-{{H_\ell}}(0)|+\eta'\le \ep^2$, we are back to the previous situation. 
\end{proof}

Using Proposition~\ref{propmajs}, we are now able to find an upper bound for the Hausdorff dimension of any $\underline{E}_{\M_\mu}(H)$. 
\begin{proposition}\label{ubip}
For $H>0$ and $\ep\ge 0$ let 
\begin{equation}\label{eq110s}
D(H,\ep)=\sup\left \{ \frac{{D_\mu}(\alpha)}{\delta}:
\begin{cases} 
\alpha\in [{H_{\min}},{H_{\max}}],\\
\exists\, i\in\{\ell,r\},\   \ \eta'\in [\eta_i,\eta]\setminus\{0\},\\
1\le \delta\le 1/\eta',\ \  \dfrac{\alpha+ \Ht_i(\eta')}{\delta}\le H+\ep
\end{cases}\right\},
\end{equation}
with the convention $\sup\emptyset=-\infty$.
For all $H>0$, we have 
\begin{equation}
\label{eq11s}
\dim \underline E_{\M_\mu}(H)\le D(H):=\lim_{\ep\to 0^+} D(H,\ep).
\end{equation}
\end{proposition}

\begin{proof}
Recalling Proposition~\ref{propmajs}, it is enough to find an upper bound for the Hausdorff dimensions of $F_{\mu,\ell}(H) $ and $ F_{\mu,r}(H)$. But for $F_{\mu,\ell}(H) $ (the same hold true for $F_{\mu,r}(H) $),  one needs only to focus on  $\ \limsup_{j\to\infty} {F}_{\mu,\ell}(j,\alpha,\eta',\delta-\ep,k)$, for the suitable values of the parameters $\alpha,\delta,\eta',\ep$ described in \eqref{defFmul}. We are going to prove  that   the weak redundancy (described in Proposition~\ref{p4}) implies that  for $\ep\in (0,1)$, for $k$ large enough, uniformly in $(\alpha,\eta',\delta)$ (under the constraint $\delta \in [1,1/\eta']$),   one has
\begin{equation}\label{dimFi}
\dim  \ \limsup_{j\to\infty} {F}_{\mu,\ell}(j,\alpha,\eta',\delta -\ep,k)\le \dfrac{{D_\mu}(\alpha)}{\delta}+ \theta(\ep),
\end{equation}
where $\lim_{\ep\to 0}\theta(\ep)=0$. Then, the result follows by taking the supremum over $(\alpha,\eta',\delta)$ in $\mathcal{P}_\ell$ and letting $\ep$ tend to~$0$. 

\mk

We prove \eqref{dimFi}. Fix $\ep>0$.

For every  $\alpha\in[{H_{\min}},{H_{\max}}]$, using Proposition \ref{fm}  and the large deviations properties \eqref{eq1sss} of $\mu$, there exists $\beta_\alpha>0$ and $j_\alpha\in\mathbb N$ such that for $j\ge j_\alpha$, 
\begin{equation}
\label{eq9s}
\#\mathcal{E}_\mu(j,\beta_\alpha,\alpha) \leq   2^{j({D_\mu}(\alpha)+\ep/2)}.
\end{equation}
 Moreover, ${D_\mu}$ being continuous over $[{H_{\min}},{H_{\max}}]$, one can choose $\beta_\alpha$ so small  that ${D_\mu}(\alpha)\le {D_\mu}(\alpha')+\ep/2$ for all $\alpha'\in [\alpha -\beta_{\alpha},\alpha+\beta_{\alpha}]$. Using the compactness of $[{H_{\min}},{H_{\max}}]$, there exists finitely many real numbers  $\alpha_1,\ldots,\alpha_p \in[{H_{\min}},{H_{\max}}]$,  such that for the associated numbers $\beta_1, ...,\beta_p>0$ and integers $j_1,...,j_p$, one has  $[{H_{\min}},{H_{\max}}]\subset \bigcup_{i=1}^p[\alpha_i-\beta_i,\alpha_i+\beta_i]$ and \eqref{eq9s} holds  for every $j\geq j_i$.

\sk

Pick  an integer $k\ge 1$ such that $1/k\le \min\{\beta_i:1\le i\le p\}$. For all $j\ge J^1_\ep=\max\{j_{i}:1\le i\le p\}$, for every $\alpha\in [{H_{\min}},{H_{\max}}]$,  there exists $1\le i\le p$ such that $\alpha\in [\alpha_i-\beta_i,\alpha_i+\beta_i]$, hence because of \eqref{eq9s}, one has  
\begin{equation}
\label{eq10s}
\#\mathcal{E}_\mu(j,1/k,\alpha) \leq   2^{j({D_\mu}(\alpha_i)+\ep/2)} \le 2^{j({D_\mu}(\alpha)+\ep)} .
\end{equation}

\sk

Consider  one triplet $(\alpha,\eta',\delta)\in \mathcal{P}_\ell$ with $\delta\in [1,1/\eta']$.  We choose the integer $k$ such that  $1/k \leq \beta$. 
Observe that with these   parameters, 
$$ {F}_{\mu,\ell}(j,\alpha,\eta',\delta-\ep,k) \subset   \bigcup_{W \in  \Sigma_{\lf j \eta'\rf}} T_{\mu, \ell} (j, \eta',\beta,W).$$
We use  Proposition~\ref{p4}  to find an upper bound for the cardinality $T_{\mu, \ell} (j, \eta',\beta,W)$ . Applying part (1) of this proposition, one can find  $\beta>0$ and  an integer $  J^2_\ep$ such that for all $j\geq   J^2_\ep$,   equation \eqref{eq4s} holds true.  

For each $j\ge J_\ep :=\max(J^1_\ep/\eta', J^2_\ep)$,  every  point $x$   belonging to $ {F}_{\mu,\ell}(j,\alpha,\eta',\delta- \ep,k) $ is close at distance at most     $2^{-{\lfloor \eta'j \rfloor}(\delta-\ep)}$  from some cube  $I_w$  such that $w\in \mathcal{S}_j(\eta) \cap \Rmu(j,\eta',\alpha\pm1/k) \cap \Tl(j,\eta',1/k)$, with $\delta\in [1,1/\eta']$.

One combines now two facts:

\begin{itemize}

\sk \item By \eqref{eq10s}, the cardinality of  those words $W \in \Sigma_{ \lfloor \eta'j\rfloor}$ satisfying 
 $$ \left |\dfrac{ -\log_2 \mu(I_{W})}{ \lf j \eta'\rf} -\alpha\right| \leq 1/k$$  is bounded from above by 
$$\#\mathcal{E}_\mu(\lfloor \eta'j  \rfloor,\alpha,1/k) \leq 2^{ \lf j \eta'\rf ({D_\mu}(\alpha)+ \ep )}.$$ 
This applies  to the words  $W=w_{\lf\eta'j\rf}$ when $w\in \Rmu(j,\eta',\alpha\pm1/k)$

\sk\item
 If $W \in \Sigma_{ \lfloor \eta'j\rfloor}$ is fixed, we know by  \eqref{eq4s} that  the cardinality  the words  $w \in \Tl(j,\eta', 1/k)$    is bounded from above by $2^{\eta'j  \ep}$.  
\end{itemize}

The first item allows us to control the number of possible  $\eta'$-roots $w_{\lf\eta'j\rf}$ of $w$, and the second item the number of possible $\eta'$-tails $w_{\sigma^{j-\lf \eta'j\rf}w}$.

We deduce  that for all   $J\ge   J_\ep $,  the set ${F}_{\mu,\ell}(j,\alpha,\eta',\delta-\ep,k)$ is covered by at most $ 2^{ \lf j \eta'\rf ({D_\mu}(\alpha)+ \ep) }$ times $2^{\eta'j  \ep}$ cubes of diameter  $2^{{\lfloor \eta'j \rfloor}(\delta-\ep)}$. 

Fix a real number $s> \dfrac{{D_\mu}(\alpha)+2\ep}{\delta}$. The $s$-Hausdorff measure $\mathcal{H}^s$ of the limsup set $F:=\limsup_{j\to+ \infty} {F}_{\mu,\ell}(j,\alpha,\eta',\delta-\ep,k) $   satisfies for every integer $J\ge J_\ep $
\begin{eqnarray*}
\mathcal{H }^s(F)  & \leq  & \sum_{j\geq J} 2^{ \lf j \eta'\rf ({D_\mu}(\alpha)+ \ep)  }2^{\eta'j  \ep}2^{- \delta{\lfloor \eta'j \rfloor}(\delta-\ep)}<+\infty.
\end{eqnarray*}
Hence $\dim F\leq  \dfrac{{D_\mu}(\alpha)+2\ep}{\delta}$. Since $\delta\ge 1 $, $\dim F\leq  \dfrac{{D_\mu}(\alpha)}{\delta}+2\ep$, and \eqref{eq9s} is proved.
\end{proof}

\subsection{Simplification of the upper bound formula \eqref{eq11s} of Section~\ref{upb1}}\label{upb2}

Summarizing the previous information, we already established that, with probability 1,   the domain of the singularity spectrum $D_{\M_\mu}$  of $\M_\mu$ is included in $[{H_\ell}(\etal), H_{\max}+{\Ht_\ell}(\weta)]$ (Lemma \ref{lempointgauche} and Section \ref{upb3}), and for all $H\in [{H_\ell}(\etal), H_{\max}+{\Ht_\ell}(\weta)]$ we   found  the upper bound $D(H)$ for  $D_{\M_\mu}(H)$ in Proposition~\ref{ubip}. 

Formula \eqref{eq11s} defining $D(H)$ is useless for the moment for at least two reasons: it is not tractable as it it is written, and we ignore yet whether this upper bound is optimal.
In this section we investigate the optimization problem raised by the expression of $D(H)$ and prove the following proposition, which 
  yields the  upper bound announced in Theorem~\ref{thm-0} for the increasing part of  $D_{\M_\mu}$.

\begin{proposition}
\label{prop46sss}
For every $H \in [{H_\ell}(\etal),{H_s}+{\Ht_\ell}(\widetilde\eta)]$, one has 
$$D(H) =
\begin{cases}\sk
\ {{D_\mu}}(H)-d(1-\eta) &\text{if }  \ \ \ \ \ \  \ \  \,  \ \ {{H_\ell}} (\eta_\ell) \le H \le {{H_\ell}}(\widetilde\eta),\\ \sk
 \  H \cdot \dfrac{  \widetilde\eta \,  {{D_\mu}}( {{H_\ell}}(\widetilde\eta))}{{{H_\ell}   } (\widetilde\eta)} &\text{if }  \ \ \ \ \ \ \ \ \ \ \ {{H_\ell}}(\widetilde\eta) \leq H \leq  {{H_\ell}}(\widetilde\eta)+{\Ht_\ell}(\widetilde\eta),\\\sk
 \ {{D_\mu}}(H-{\Ht_\ell}(\widetilde\eta))&\text{if } \  {{H_\ell}}(\widetilde\eta)+{\Ht_\ell}(\widetilde\eta) \leq H \leq H_s+{\Ht_\ell}(\widetilde\eta).
\end{cases}$$
 In particular $H\longmapsto D(H)$ is strictly increasing on the interval $ [{H_\ell}(\etal),{H_s}+{\Ht_\ell}(\widetilde\eta))$, $D({H_\ell}(\etal))={D_\mu}({H_\ell}(\etal))-d(1-\eta)$ and   $D({H_s}+{\Ht_\ell}(\widetilde\eta)) = d$.
\end{proposition}

 \begin{remark}
 For every $H\in ({H_\ell}(\etal), H_s+{\Ht_\ell}(\weta)]$, we are going to show that $D(H)=D(H,0)$ and there exists a unique triplet $(\alpha_H,\eta_H, \delta_H)$ which realizes the supremum in $D(H,0)$. Lemma \ref{lemcritical} is particularly important, since it describes the precise scenario leading to the optimal upper and lower bounds for the dimensions of the level sets.  In particular,  Lemma \ref{lemcritical}  distinguishes this triplet, which yields the simplified formula and is key for the next  Section \ref{sec_lower}, where we prove that the upper bound $D(H)$ is a lower bound for  $D_{\M_\mu}(H)$ as well, when $H\in [{H_\ell}(\etal), H_s+{\Ht_\ell}(\weta)]$.
 
The rest of the section is less essential and more technical, since it involves only computations based on Legendre transforms.
\end{remark}

\begin{figure}
        \begin{tikzpicture}[xscale=2.2,yscale=2.9]{\small
   \draw [->] (0,-0.2) -- (0,1.2) [radius=0.006] node [above] {$D_{\M_\mu}(h)$ };
\draw [->] (-0.2,0) -- (3.2,0) node [right] {$h$};
\draw [thick, domain=0.6:1.9,color=black]  plot ({-(exp(\x*ln(1/5))*ln(0.2)+exp(\x*ln(0.8))*ln(0.8))/(ln(2)*(exp(\x*ln(1/5))+exp(\x*ln(0.8)) ) )} , {(-0.33)-\x*( exp(\x*ln(1/5))*ln(0.2)+exp(\x*ln(0.8))*ln(0.8))/(ln(2)*(exp(\x*ln(1/5))+exp(\x*ln(0.8))))+ ln((exp(\x*ln(1/5))+exp(\x*ln(0.8))))/ln(2)});
\draw [thick, domain=0:0.6, color=cyan]  plot ({(0.47)-(exp(\x*ln(1/5))*ln(0.2)+exp(\x*ln(0.8))*ln(0.8))/(ln(2)*(exp(\x*ln(1/5))+exp(\x*ln(0.8)) ) )} , {-\x*( exp(\x*ln(1/5))*ln(0.2)+exp(\x*ln(0.8))*ln(0.8))/(ln(2)*(exp(\x*ln(1/5))+exp(\x*ln(0.8))))+ ln((exp(\x*ln(1/5))+exp(\x*ln(0.8))))/ln(2)});
  \draw[dashed] (0,1) -- (1.82,1);
\draw  [fill]  (0.45,0) circle [radius=0.03] node [below] {${{H_\ell}}(0)$};
 \draw[fill] (0.93,0.56) circle [radius=0.03] [dashed] (0.93,0.56) -- (0.93,-0.0) [fill] circle [radius=0.03]  [dashed] (0.93,0.0) -- (0.93,-0.1)  node [below]  {${{H_\ell}}(\widetilde \eta)$};
 \draw[fill] (1.41,0.89) circle [radius=0.03] [dashed] (1.41,0.89) -- (1.41,-0.0) [fill] circle [radius=0.03]  [dashed] (1.41,0.0) -- (1.41,-0.3)  node [below]  {${{H_\ell}}(\widetilde \eta)+{\Ht_\ell}(\widetilde \eta)$};
 \draw[fill] (1.82,1) circle [radius=0.03] [dashed] (1.82,1) -- (1.82,-0.0) [fill] circle [radius=0.03]  [dashed] (1.82,0.0) -- (1.82,-0.0)  node [below]  { \ \ \ \ \ \ \ \ \  ${H_s}+{\Ht_\ell}(\widetilde \eta)$};
 \draw [thick,color=blue]    (0.93,0.56) -- (1.41,0.89); 
\draw [dotted,color=blue]   (0,0) --  (0.93,0.57) ;
 \draw [fill] (-0.1,-0.10)   node [left] {$0$}; 
 \draw [thick, fill] (1.82,1) -- (2.82,1); 
 \draw [fill] (0,1) circle [radius=0.03] node [left] {$d \ $}; 
  \draw [<-] (0.5,0.25) -- (-0.7,0.25) node [left] {${D_\mu}(H)-d(1-\eta)$};
 \draw [<-] (1.1,0.75) -- (-0.8,0.75) node [left] {$ H \cdot \dfrac{  \widetilde\eta \,  {{D_\mu}}( {{H_\ell}}(\widetilde\eta))}{{{H_\ell}  } (\widetilde\eta)} $};
  \draw [<-] (1.6,0.9) -- (3.3,0.5) node [right] {${{D_\mu}}(H-{\Ht_\ell}(\widetilde\eta))$};
   } 
\end{tikzpicture} 
\caption{  The mapping $H\mapsto D(H)$}
\end{figure}

For the proof of Proposition \ref{prop46sss}, we will use the following facts:
\begin{enumerate}
\mk\item[{\bf Fact 1:}] By Lemma \ref{lem2s}, 
the mapping $\alpha\mapsto \frac{{D_\mu}(\alpha)}{\alpha+{\Ht_\ell}(\widetilde \eta)}$  is increasing when $\alpha \leq {H_\ell}(\weta)$, decreasing when $\alpha \geq {H_\ell}(\weta)$.

\mk\item[{\bf Fact 2:}]  For all $\eta'\in [\etal,\eta]\setminus\{0\}$, one has $ \eta'=  \frac{ {{H_\ell}} (\eta')+ {\Ht_\ell}(\eta')}{{{H_\ell}} (\eta')} $.

\mk\item[{\bf Fact 3:}]  For all $\eta'\in [\etal,\eta]$, one has $ {D_\mu}({H_\ell}(\eta')) =   \dfrac{d(1-\eta)} { 1-\eta'}  $.   

\end{enumerate}

\begin{proof}[Proof of Proposition \ref{prop46sss}]
We start with some direct observations.

Observe first that $D(H,\ep)$  and $D(H)$ are non-decreasing mappings with respect to the variable $H$. Hence we only need to deal  with $H>{H_\ell}(\eta_\ell)$, since it will follow from our computations that $ D({H_\ell}(\etal)) = \lim_{H\to {H_\ell}(\etal)^+} D(H) = {D_\mu}({H_\ell}(\etal))-d(1-\eta) $.

In addition,   for all $\eta'\in [\eta_r,\eta]\setminus\{0\}$, one has ${\Ht_r}(\eta')\ge {\Ht_\ell} (\eta)\ge {\Ht_\ell}(\widetilde\eta)$. Hence it is enough to consider ${\Ht_\ell}(\eta')$ and $\eta'\in [\eta_\ell,\eta]\setminus\{0\}$ to obtain the greatest upper bound in $D(H,\ep)$. One deduces that  \eqref {eq110s} reduces to
$$
D(H,\ep)=\sup\left \{ \frac{{D_\mu}(\alpha)}{\delta}:
\begin{cases} 
  \,  \alpha\in [{H_{\min}},{H_{\max}}], \    \ \eta'\in [\eta_\ell,\eta]\setminus\{0\},\\
1\le \delta\le 1/\eta',\ \  \dfrac{\alpha+{\Ht_\ell}(\eta')}{\delta}\le H+\ep
\end{cases}\right\}.
$$
We call $\mathcal{P}_{H,\ep}$ the domain of admissible values appearing in the right hand-side above, so that
$$
D(H,\ep)=\sup\left \{ \frac{{D_\mu}(\alpha)}{\delta}: (\alpha,\eta',\delta)\in \mathcal{P}_{H,\ep} \right\}.$$

\mk

One starts by finding the expected lower bounds for $D(H)$.

\begin{lemma} 
\label{lemcritical}
\begin{enumerate}
\sk\item
If  $H\in [{{H_\ell}}(\widetilde\eta)+{\Ht_\ell}(\widetilde\eta),{H_s}+{\Ht_\ell}(\widetilde\eta))$,   the triplet 
\begin{equation}
\label{triplet1}
  \big(\alpha_H:=H-{\Ht_\ell}(\weta), \weta ,1\big)
  \end{equation}
  belongs to the domain $\mathcal{P}_{H,0}$, and
  \begin{equation}\label{DH01}
D(H) \geq \frac{{D_\mu}(\alpha_H)}{1} = {D_\mu}\big(H- {\Ht_\ell}(\widetilde\eta)\big) .
  \end{equation}

\sk\item  If  $H\in [ {{H_\ell}}(\widetilde\eta)    , {{H_\ell}}(\widetilde\eta)+{\Ht_\ell}(\widetilde\eta))$,  the triplet 
\begin{equation}
\label{triplet2}
  \left({{H_\ell}}(\widetilde\eta), \weta,\ \frac{{H_\ell}(\widetilde\eta)+{\Ht_\ell}(\widetilde \eta)}{H_\ell (\widetilde \eta)}\right)
   \end{equation} 
belongs to the domain $\mathcal{P}_{H,0}$ and
  \begin{equation}\label{DH02}
    D(H) \geq  H\frac{{D_\mu}({{H_\ell}}(\widetilde\eta) )}{{{H_\ell}}(\widetilde\eta)+{\Ht_\ell}(\widetilde \eta)}=H\cdot \dfrac{\weta \cdot {D_\mu}( {{H_\ell}}(\widetilde\eta)) }{{{H_\ell}}(\widetilde\eta)} .
  \end{equation}

\sk\item
If  $H\in ({{H_\ell}}(\etal),   {{H_\ell}}(\widetilde\eta)   )$, let       $\eta_H\in (\eta_\ell,\widetilde\eta)$ be the unique real number such that  
\begin{equation}
\label{defetaH}
H={{H_\ell}}(\eta_H).
\end{equation}
 The triplet   
\begin{equation}
\label{triplet3}
({H_\ell}(\eta_H), \eta_H, 1/\eta_H)
 \end{equation}
  belongs to $\mathcal{P}_{H,0}$ and  
  \begin{equation}\label{DH03}
   D(H) \geq \frac{ {D_\mu}({H_\ell}(\eta_H)) }{1/\eta_H}    = {D_\mu}(H) - d(1-\eta) .
 \end{equation}
 \end{enumerate}
  \end{lemma}

\begin{proof}  
Observe that   $\mathcal{P}_{H,0} \subset \bigcap_{\ep>0} \mathcal{P}_{H,\ep}$. So, if a triplet $(\alpha,\eta',\delta) $ belongs to $\mathcal{P}_{H,0}$, $D(H,\ep) \geq \frac{{D_\mu}(\alpha)}{\delta}$ for every $\ep$, and one necessarily has $D(H) \geq \frac{{D_\mu}(\alpha)}{\delta}.$

\sk

The fact that the three triplets belong to the associated domains     $\mathcal{P}_{H,0}$ is a simple calculation.

\sk

Part (1) is immediate.

\sk

Part (2) follows from Fact 2.

\sk

Concerning Part (3), one observes that
\begin{eqnarray*}
\frac{ {D_\mu}({H_\ell}(\eta_H)) }{1/\eta_H}   & =   & \eta_H \dfrac{d(1-\eta)}{1-\eta_H}  =  \frac {d(1-\eta)}{1-\eta_H}-d(1-\eta)  \\
& =&  {D_\mu}({H_\ell}(\eta_H)) - d(1-\eta) = {D_\mu}(H) - d(1-\eta) .
\end{eqnarray*}
\end{proof}

From the above lower bounds, one deduces that 
$
D(H)=\lim_{\ep\to 0^+}D(H,\ep)=D(H,0).
$
Indeed, take any $H\in ({H_\ell}(\etal),{H_s}+{\Ht_\ell}(\widetilde\eta)]$. Obviously $d\geq D(H,\ep) \geq D(H) >0$, for every $\ep\geq 0$. Recalling  that  $D(H,\ep)$ is a supremum of quantities $\frac{{D_\mu}(\alpha)}{\delta}$ and ${D_\mu}(\alpha) \leq d$, it is enough to consider triplets of the form $(\alpha,\eta',\delta)$ with $\delta \leq d/D(H,\ep) $. As a consequence, in this range of triplets, recalling the constraint $\dfrac{\alpha+{\Ht_\ell}(\eta')}{\delta}\le H+\ep$, the quantity ${\Ht_\ell}(\eta') $ must be bounded to contribute to the value of $D(H,\ep)$. This means that $\eta'$ is bounded from below by some positive constant depending on $H$ only.

We know now that the domain of suitable parameters leading to $D(H,\ep)$ can be chosen to be compact, and independent of $\ep>0$. The continuity of all the functions involved in the domain $\mathcal{P}_{H,\ep}$, $\ep\geq 0$, allows us to conclude that $D(H,0) = \lim_{\ep\to 0^+}D(H,\ep)$, and 
\begin{equation}
\label{dhdh0}
D(H)= D(H,0) = \max \left \{ \frac{{D_\mu}(\alpha)}{\delta}:
\begin{cases} 
\exists \,  \alpha\in [{H_{\min}},{H_{\max}}], \  \exists \ \eta'\in [\eta_\ell,\eta]\setminus\{0\},\\
1\le \delta\le 1/\eta',\ \  \dfrac{\alpha+{\Ht_\ell}(\eta')}{\delta}\le H 
\end{cases}  \!\!\!\! \right\},
\end{equation}   
where we know that the maximum is effectively reached (it is not only a supremum as in formula \eqref{eq110s}).

\mk

 Now let us make some general remarks on $D(H,0)$.

Suppose the maximum  in \eqref{dhdh0} is realized at  some triplet $(\alpha,\eta',\delta)$. Observe that if $\alpha $ were strictly greater than ${H_s}$,  one could  improve the bound   ${D_\mu}(\alpha)/\delta$ by replacing $\alpha $ by $H_s$ and not changing the value of the other parameters. This contradicts the maximality of  ${D_\mu}(\alpha)/\delta$. As a conclusion,  
$\alpha\leq{H_s}.$

\medskip

Suppose now that $\dfrac{\alpha+{\Ht_\ell}(\eta')}{\delta}< H$. Then necessarily $\delta=1$,   otherwise one could improve the upper bound $\frac{{D_\mu}(\alpha)}{\delta}$ by slightly decreasing $\delta$. When $\delta=1$, one has  $\alpha+{\Ht_\ell}(\eta') < H \le {H_s}+{\Ht_\ell}(\widetilde\eta)$.   Observe that  it is necessary to have $\alpha <H_s$, since ${\Ht_\ell}(\eta')\ge {\Ht_\ell}(\widetilde\eta)$. By taking    $\alpha' \in (\alpha, H_s)$ still satisfying $\alpha'+{\Ht_\ell}(\eta') < H$, one   gets a larger value for $\frac{{D_\mu}(\alpha')}{\delta}$, which contradicts again the maximality of $\frac{{D_\mu}(\alpha)}{\delta}$. So,   for the optimal triplet, one has necessarily the equality 
\begin{equation}
\label{eq12s}
 \dfrac{\alpha+{\Ht_\ell}(\eta')}{\delta} =H.
\end{equation}

\mk

Finally,  observe that  \eqref{eq12s} implies that  
\begin{align}
 D(H,0)& \label{up2}\le
H \cdot \max\left \{\frac{{D_\mu}(\alpha)}{\alpha+{\Ht_\ell}(\eta')}:
\begin{cases} 
\alpha\in [{H_{\min}},{H_{\max}}],\  \eta'\in [\eta_\ell,\eta]\setminus\{0\},\\
\ \exists \ 1\le \delta\le 1/\eta',\ 
\dfrac{\alpha+{\Ht_\ell}(\eta')}{\delta} = H
\end{cases}\right\}.
\end{align}

We now  distinguish three cases.

 \medskip 
\noindent
{\bf $\bullet$ First case:  $H\in ({{H_\ell}}(\widetilde\eta)+{\Ht_\ell}(\widetilde\eta),{H_s}+{\Ht_\ell}(\widetilde\eta))$.}
\sk

Recall the triplet \eqref{triplet1}, which satisfies \eqref{eq12s}. 

If  $\alpha > \alpha_H$,   Fact 1 yields
$$H\frac{{D_\mu}(\alpha)}{\alpha+{\Ht_\ell}(\eta')}\le H\frac{{D_\mu}(\alpha)}{\alpha+{\Ht_\ell}(\widetilde \eta)}\le H \frac{{D_\mu}(\alpha_H)}{\alpha_H+{\Ht_\ell}(\widetilde \eta)}={D_\mu}(\alpha_H)$$
  
   If $\alpha<\alpha_H$, we have ${D_\mu}(\alpha)/\delta<{D_\mu}(\alpha_H)/1 = {D_\mu}(\alpha_H)$. 
   
   Consequently, the maximum is reached necessarily at $(\alpha_H, \weta, 1)$, and it equals \eqref{DH01}. 
\medskip\medskip 

\noindent
{\bf $\bullet$ Second case: $H\in [ {{H_\ell}}(\widetilde\eta)    , {{H_\ell}}(\widetilde\eta)+{\Ht_\ell}(\widetilde\eta))$.  }

\sk
For any $\alpha\in [{H_{\min}},{H_{\max}}]$ and any $\eta'$, by Fact 1 we have 
$$H\frac{{D_\mu}(\alpha)}{\alpha+{\Ht_\ell}(\eta')}\le H\frac{{D_\mu}(\alpha)}{\alpha+{\Ht_\ell}(\widetilde \eta)} \le H\frac{{D_\mu}({{H_\ell}}(\widetilde\eta ))}{{{H_\ell}}(\widetilde\eta)+{\Ht_\ell}(\widetilde \eta)},$$
which is the value obtained in \eqref{DH02} with the triplet   \eqref{triplet2}, which satisfies  \eqref{eq12s}.

\medskip \medskip 
\noindent
{\bf $\bullet$ Third case:} $H\in ({{H_\ell}}(\etal),   {{H_\ell}}(\widetilde\eta)   )$. 

 \medskip

The optimization process here is more intricate. Indeed, in the first two cases,  only one amongst the three parameters involved in the optimal triplet depends on $H$. In this range of local dimensions, the three parameters of the optimal triplet  are functions of $H$.

\begin{lemma}
\label{lem3s}
When $H\in ( {{H_\ell}}(\etal),   {{H_\ell}}(\widetilde\eta)   )$,  the value of $D(H,0)$ is reached at some triplet $(\alpha,\eta',1/\eta')$, for some $\eta' \in  [\etal,\weta)\setminus \{0\}$. 
\end{lemma}
\begin{proof}  
 
Suppose that $D(H,0)$ is reached at some triplet $(\alpha,\eta',\delta)$, with $1< \delta<1/\eta'$. 

If $\eta'\neq\widetilde\eta$, we can perturb slightly $(\alpha,\eta',\delta)$  into a new triplet $(\alpha'',\eta'',\delta'')$ such that  $\alpha '' \geq \alpha$,  ${\Ht_\ell}(\eta'')<{\Ht_\ell}(\eta')$,  $1 < \delta''<\delta' $ and  $\delta''<1/\eta''$, and such that \eqref{eq12s} still holds for the new triplet. This contradicts the maximality of ${D_\mu}(\alpha)/\delta$.
One deduces  that $\eta'=\weta$, and the optimal  triplet  is in fact $(\alpha,\weta,\delta)$  with  $1< \delta<1/\weta$.

\sk

Now, the constraints \eqref{eq12s} and $H<{H_\ell}(\widetilde\eta)$ imply that $\alpha<{H_\ell}(\widetilde\eta)$, since ${H_\ell}(\widetilde\eta)=\widetilde\eta({H_\ell}(\widetilde\eta)+{\Ht_\ell}(\widetilde\eta))$. Then, optimizing  in $\alpha$ the ratio ${D_\mu}(\alpha)/\delta$ under the constraint  \eqref{eq12s}  amounts to studying the mapping
$$ \alpha \mapsto H \frac{{D_\mu}(\alpha)}{ \alpha+{\Ht_\ell}(\weta)}$$
over $[\alpha_{\min},{H_\ell}(\widetilde\eta))$. Fact 1 ensures that it is increasing. If the  optimal  triplet  is  $(\alpha,\weta,\delta)$  with  $1< \delta<1/\weta$, we can slightly increase the values of  $\alpha$ and $\delta$ in $\alpha'$ and $\delta'$, preserve \eqref{eq12s}, $\alpha'<{H_\ell}(\widetilde\eta)$ and $1< \delta'<1/\weta$, and get the contradiction $H \frac{{D_\mu}(\alpha')}{ \alpha'+{\Ht_\ell}(\weta)}>D(H,0)= \frac{{D_\mu}(\alpha)}{ \alpha+{\Ht_\ell}(\weta)}$.   

Next suppose that $D(H,0)$ is reached at some triplet $(\alpha,\eta',1)$, for some $\eta' \in [\etal,\eta]\setminus\{0\}$. This imposes    $H=\alpha+{\Ht_\ell}(\eta')$, and  one looks for the maximum of ${D_\mu}(\alpha)$. Necessarily $\alpha  \in I=[ {{H_\ell}}(\etal) -{\Ht_\ell}(\eta'),   {{H_\ell}}(\widetilde\eta)  -{\Ht_\ell}(\eta')  )$, which is an interval included in the increasing part of the spectrum ${D_\mu}$. However ${D_\mu}$ does not reach a maximum over $I$, hence a new contradiction.

The previous cases leading to a contradiction, we deduce that  necessarily $D(H,0)$ is reached at some triplet $(\alpha,\eta',1/\eta')$, for some $\eta'\in [\etal,\eta]\setminus\{0\}$. Hence we have
\begin{equation}
\label{eq15s}
 H=\eta'(\alpha+{\Ht_\ell}(\eta'))  \ \mbox { and } \ D(H,0)=\eta'{D_\mu}(\alpha),
 \end{equation}
  with $\eta'\in [ \eta_\ell,\eta]\setminus\{0\}$. 

\mk

Further, we  prove that it is enough to consider $\eta'\in [\etal,\weta)$. 

Suppose that  \eqref{eq15s} holds for some  $\eta'\in[\widetilde\eta ,\eta]$. 
Since $H< \widetilde\eta({{H_\ell}}(\widetilde\eta)+{\Ht_\ell} (\widetilde\eta))$ and ${\Ht_\ell}(\eta')\ge {\Ht_\ell}(\widetilde\eta)$, we have $\alpha< {{H_\ell}}(\widetilde\eta)$. Consequently, there exists $\alpha'\in (\alpha,{{H_\ell}}(\widetilde\eta))$ such that $H=\widetilde\eta(\alpha'+{\Ht_\ell} (\widetilde\eta))$. For the triplet $(\alpha', \weta,1/\weta)$, one has
$$\widetilde\eta {D_\mu}(\alpha')= H \dfrac{ {D_\mu}(\alpha')}{\alpha'+{\Ht_\ell} (\widetilde\eta)} > H \dfrac{ {D_\mu}(\alpha)}{\alpha+{\Ht_\ell} (\widetilde\eta)} , $$ 
where we used  Fact 1 and ${H_\ell}(\weta) >\alpha'>\alpha$. Finally, since ${\Ht_\ell} (\widetilde\eta)  \le {\Ht_\ell} (\eta') $, one sees that 
 $$\widetilde\eta {D_\mu}(\alpha' ) > H \dfrac{ {D_\mu}(\alpha)}{\alpha+{\Ht_\ell} (\eta')} =D(H,0),$$
 which is a (last) contradiction.
\end{proof}

From last Lemma, the maximum $D(H,0)$  reduces to
$$
 D(H,0)=\max\Big \{\eta' {{D_\mu}(\alpha)} :
\alpha\in [{H_{\min}}, H_s], \ \eta'\in [\eta_\ell,\weta)\setminus\{0\},\     \eta'({\alpha+{\Ht_\ell}(\eta'))}=H  \Big\}.
$$
 
This is standard optimization under constraints. The maximum is reached  when
$$ {D_\mu}'(\alpha )(\alpha +\Ht'_\ell(\eta'))+ {D_\mu}(\alpha ) = 0 ,$$
or equivalently when
$$  \Ht'_\ell(\eta') =  - \frac{D^*_\mu( {D_\mu}'(\alpha))}{ D'_\mu(\alpha)}.$$

When $\eta'$ is fixed, this happens if and only if $\alpha={H_\ell}(\eta')$.  This means that   $H=\eta'({H_\ell}(\eta') +{\Ht_\ell}(\eta')) $, so $H={H_\ell}(\eta')$. This leads to the choice \eqref{defetaH} for $\eta'$, and to  the triplet \eqref{triplet3}, which gives the value \eqref{DH03} for $D(H)$. 
\end{proof}

\begin{remark}
\label{rk1s}
A key observation for the following is that actually we proved a little more than what we announced. Indeed, as a direct by-product of the proof, not only we know that when $H\leq H_s+{\Ht_\ell}(\widetilde\eta)$, $D_{\M_\mu}(H) (=\dim \underline E_{\M_\mu}(H) ) \, \leq D(H)$, we also have that
\begin{equation}
\label{majdimsss}
\dim \underline E^\leq _{\M_\mu}(H)  \, \leq D(H).
\end{equation}

This inequality is useful in the following section.
\end{remark}

\begin{remark}
\label{rk2s}There is no chance for   $D(H)$ to  be an optimal bound in the decreasing part of the singularity spectrum of $\M_\mu$, since the mapping $H\mapsto D(H)$ is non-decreasing.
\end{remark}

 \section{Lower bound for the singularity spectrum}
\label{sec_lower}

For each admissible  local dimension $H$, we are going to exhibit an auxiliary probability measure $\nu$ (which depends on $H$) such that $\nu \big(\underline E_{\M_\mu}(H) \big)=1$, and such that the dimension of $\nu$ equals the announced value for $D_{\Mm}(H)$: i.e. $D(H)$  when $H\leq H_s+{\Ht_\ell}(\weta)$, and    ${D_\mu}(H-{\Ht_\ell}(\weta))$ when $H> H_s+{\Ht_\ell}(\weta)$. 

These auxiliary measures do not always have the same nature, depending on $H$. They can be taken as a Gibbs measures when $H\in [ H_\ell(\weta)+\Ht_\ell(\weta), H_{\max}+\Ht_\ell(\weta))$,  but not for the other values of $H$.  

We introduce  two  families of measures in Section~\ref{fammeas}, whose properties are established in Section \ref{sec5proof}. Then we obtain the sharp lower bound for $D_{\M_\mu}$ in Sections~\ref{sharplow1} to \ref{sharplow3}.

 \subsection{Two families of measures}\label{fammeas}

The first family will be used to obtain a sharp lower bound for $D_{\M_\mu}(H)$ when $H\in[{H_\ell}(\weta)+{\Ht_\ell}(\weta), H_{\max}+{\Ht_\ell}(\weta)]$. It is based on the following result.

Recall  Proposition~\ref{discretization} in which  the event $\mathscr{C}(N,J,W)$ is defined.  

\begin{theorem}\label{cantor} With probability 1, for all $\alpha\in [H_{\min},H_{\max}]$, there exists an exact dimensional Borel probability measure $\nu_\alpha$ of Hausdorff dimension ${D_\mu}(\alpha)$ supported on $\widetilde E_\mu(\alpha)$ (i.e. $\nu_\alpha(\widetilde E_\mu(\alpha))=1$), such that:
\begin{enumerate}
\item
for all $\delta>1$ we have
$$
\nu_\alpha \left (\bigcap_{J\ge 1} \ \bigcup_{j\ge J}\   \bigcup_{w\in\mathcal{S}_j(\eta)}  B \big(x_w,(2\cdot 2^{-\lfloor \eta j \rfloor})^\delta\big)\right )=0.
$$ 

\medskip
\item  for all $N>1/\eta$,  for $\nu_\alpha$-almost every $x$, there exists an integer $J_{N,\alpha,x}\ge 1$ such that for all $J\geq J_{N,\alpha,x}$, the event $\mathscr{C}(N,J,x_{|J})$   is not realized.
\end{enumerate}
\end{theorem}
Theorem \ref{cantor} is proved at the end of this Section (Section \ref{sec5proof}). Observe that the result holds simultaneously for all $\alpha\in[H_{\min}, H_{\max}]$.

In the first item, the limsup set contains those points $x\in \zu^d$  that are very close to the surviving coefficients, i.e. those $x$ satisfying for some $\delta>1$
$$|x-x_w| < 2 \cdot 2^{- \lf |w|\eta\rf \delta}$$
for infinitely many surviving words $w$. We know  by the covering Lemma \ref{lem1} that when $\delta <1$, every $x\in  \zu^d$ satisfies the last inequality infinitely many times. Part (1) of Theorem \ref{cantor} states that this is no longer true when $\delta >1$, in the sense that the $\nu_\alpha$-measure of these sets of points is always 0.

\mk

 The second part of the Theorem is technical, and used in the proofs below.

\mk\mk

The second family of measures allows us to compute the value of $D_{\M_\mu}(H)$ when  $H \in[{H_\ell}(\eta_\ell),{H_\ell}(\weta)+{\Ht_\ell}(\weta)]$.  These measures are built thanks to the theory of heterogeneous ubiquity theory, developed in \cite{BSubiquity1,BSubiquity2,BSubiquity3,Durand}, whose main results can be resumed as follows.

\begin{theorem}
\label{ubi}
Let $\mathcal{F} = ((x_n,r_n))_{n\geq 1}$ be a sequence of couples such that $(x_n)_{n\geq 1}$ is a sequence of points in $\zu^d$, and $(r_n)_{n\geq 1}$ is a positive  sequence converging to zero. Assume that 
\begin{equation}
\label{eq17s}
(0,1)^d \subset \limsup _{n\to+\infty} B(x_n,r_n).
\end{equation}
Let $\alpha\in(H_{\min},H_{\max})$. Recall that the Gibbs measure  $\mu_\alpha$ was defined in Proposition~\ref{fm}(4). 
For every $\delta\geq 1$ and any positive sequence $\widetilde \beta:=(\widetilde \beta_n)_{n\geq 1} $ converging to zero, define\begin{equation}
\label{defU}
U_\mu \big (\alpha,\delta, \mathcal{F}, \widetilde\beta\big ):=   \bigcap_{N\geq 1  }   \ \bigcup_{ \substack{n\geq N :  \\ (r_n)^{\alpha+\widetilde\beta_n }\leq {\mu_\alpha(B(x_n,r_n))} \leq (r_n)^{\alpha-\widetilde\beta_n} }} B\Big (x_n, (r_n)^\delta\Big) .
\end{equation}

For every $\delta\geq 1$, there exists a Borel probability measure $\nu_{\alpha,\delta}$ and a positive sequence $\widetilde \beta:= (\widetilde \beta_n)_{n\geq 1} $ converging to zero such that 
$$\nu_{\alpha,\delta}\left( U_\mu \big (\alpha,\delta, \mathcal{F}, \widetilde\beta\big )\right) = 1,$$
and $\nu_{\alpha,\delta}(E)=0$ for every set $E$ such that $\dim\, E<{D_\mu}(\alpha)/\delta$.

In particular, one has 
$$ \dim U_\mu \big (\alpha,\delta, \mathcal{F}, \widetilde\beta\big )  \geq \dim\nu_{\alpha,\delta} \geq  \frac{ {D_\mu}(\alpha)}{\delta}.$$

\mk

Moreover, if $(\alpha^{(p)},\mathcal{F}^{(p)},\delta^{(p)})_{p \geq 1}$ stands for a sequence of parameters satisfying the above conditions, there exists a measure $\widetilde \nu $ and sequences $\widetilde \beta^{(p)}: =(\widetilde \beta^{(p)}_{n})_{p\ge 1,n\geq 1}$ converging to zero satisfying
$$\widetilde \nu\left( \bigcap_{p\geq 1}  U_\mu\big (\alpha^{(p)},\delta^{(p)},\mathcal{F}^{(p)}, \widetilde\beta^{(p)}\big)  \right) = 1,$$
and $\widetilde \nu(E)=0$ for every set $E$ such that $\displaystyle \dim\, E<\inf_{p\geq 1 } \frac{{D_\mu}(\alpha^{(p)})}{\delta^{(p)}} $.

In particular,  
$$
\dim  \bigcap_{p\geq 1}  U_\mu\big (\alpha^{(p)},\delta^{(p)},\mathcal{F}^{(p)}, \widetilde\beta^{(p)}\big)       \geq \inf_{p\geq 1}  \frac{{D_\mu}(\alpha^{(p)})}{\delta^{(p)}}.
$$
\end{theorem}

The last property is due to the fact that the sets $U_\mu\big(\alpha,\delta,\mathcal{F}, \widetilde\beta\big) $ enjoy the large intersection property, i.e. when intersecting a countable number of them, the Hausdorff dimension of the resulting set is at least the infimum of all the dimensions, see \cite{BSubiquity3,Durand}.

\mk

We are going to apply Theorem \ref{ubi} with specific families $(x_n,r_n)_{n\geq 1}$:

\mk

{Let   $\mathcal{D}_\ell$ be a dense countable subset of $[\eta_\ell,\eta]\setminus\{0\}$, such that $\widetilde\eta\in\mathcal D_\ell$.} With probability~1, for all $\eta'\in \mathcal D_\ell$, Proposition~\ref{p'2} proves  the existence of   words $w\in \mathcal{S}_j(\eta,W) \cap \Tl(j,\eta',\ep^3_j)$,  for $j$ large enough, for all $W \in \Sigma_{\lfloor \eta' j\rfloor}$. For such a surviving word $w$, we set $r_w= 2\cdot 2^{-\lfloor \eta'j\rfloor}$. The  sequence of  couples $(x_w,r_w)$ obtained in this way is  denoted 
$$\mathcal{F}_{\eta'}:= \big(x_n(\eta'), r_n(\eta')  \big)_{n\geq 1}$$
 after being re-ordered so that the sequence of radii $ (r_n(\eta'))_{n\geq 1}$ is non-increasing. By construction, the covering property \eqref{eq17s} is satisfied for the family $\mathcal{F}_\eta'$, so that the second part of Theorem \ref{ubi} can be applied  with the countable number of families $\big(\mathcal{F}_{\eta'}\big)_{\eta'\in\mathcal{D}_\ell}$.

\subsection{The right part of the spectrum $D_{\M_\mu}$ }\label{sharplow1} \ 

Recall that for $ H\in \big [{{H_\ell}}(\widetilde\eta)+{\Ht_\ell}(\widetilde\eta), {H_{\max}}+{\Ht_\ell}(\widetilde\eta)\big ]$ we defined $\alpha_H=H-{\Ht_\ell}(\widetilde\eta)$. Also recall Theorem~\ref{cantor} in which the measure $\nu_{\alpha_H}$ is defined.    
%
\begin{lemma}
\label{lemmino}
With probability 1, for all $H\in  [{{H_\ell}}(\widetilde\eta)+{\Ht_\ell}(\widetilde\eta), {H_{\max}}+{\Ht_\ell}(\widetilde\eta)]$, there exists a set $G_H  \subset \widetilde E_\mu(\alpha_H)$ such that:
\begin{itemize}
\item
$\nu_{\alpha_H}(G_H) =1$
\item
 for all $x\in G_H$, for all integers $N> 1/\eta$, there exists $J_N(x) \ge 1$ such that for all $J\ge J_N(x)$, for all $ J \le j < J/(\weta_{N,-1})$, one has 
 $$ \bigcup_{W\in \mathcal{N}_J(x)} \mathcal{S}_{j }(\eta,W)=\emptyset.$$
 \end{itemize}

\end{lemma} 

\begin{proof}
Notice that for all $N\ge 1 $ we have $1/\weta_{N,-1}<1/\eta$. Assume towards contradiction  that with positive probability,  there exists $H\in [{{H_\ell}}(\widetilde\eta)+{\Ht_\ell}(\widetilde\eta), {H_{\max}}+{\Ht_\ell}(\widetilde\eta)]$, a set  \red{$F_H$} of positive $\nu_{\alpha_H}$-measure, and $N\ge 1$  such that 
$$
F_H\subset\bigcap_{J\ge 1} \ \bigcup_{j\ge J} \   \bigcup_{w\in\mathcal{S}_j(\eta) }  B\big (x_w,(2\cdot 2^{-\lfloor \eta j \rfloor})^\delta\big)$$
for all  $\delta\in (1,\weta^{-1}\weta_{N,-1})$.  
This contradicts Theorem~\ref{cantor}. 

Consequently we get that, with probability 1, for all $H\in [{{H_\ell}}(\widetilde\eta)+{\Ht_\ell}(\widetilde\eta), {H_{\max}}+{\Ht_\ell}(\widetilde\eta)]$ a set $G_H $ such that the two items of the statement hold. Moreover,   $ G_H$ can be taken a subset of $\widetilde E_\mu(\alpha_H)$  since  $\nu_{\alpha_H}\big(\widetilde E_\mu(\alpha_H)\big) =1$.
\end{proof}

Now, we prove that $D_{\Mm}(H) \geq D_{\mu} \big(H-\Ht(\weta)\big)$.

\mk Consider a set $\Omega'$ of  probability 1 over which the conclusions of Theorem \ref{cantor}  and Lemma~\ref{lemmino} hold true.

\begin{lemma} For all $\omega\in\Omega'$, for all $H\in [{{H_\ell}}(\widetilde\eta)+{\Ht_\ell}(\widetilde\eta), {H_{\max}}+{\Ht_\ell}(\widetilde\eta)]$, one has  $G_H\subset    \underline E_{\M_\mu}(H) $.

\end{lemma} 

\begin{proof}
Take  $\omega\in \Omega'$, and fix an integer $N>1/\eta$. Fix $x\in G_H$.
 We focus on  the values of $M_\mu(I_J(x))$. We   analyse the values of $\mu(I_w) $ when  $w\in \mathcal{S}_j(\eta)$ is a surviving vertex such that $I_w$ is included in the neighborhood $\mathcal{N}_J(x)$ of $x$. For this, we apply   the quasi-Bernoulli  property \eqref{quasib} to get
$$\mu(I_w) \approx  \mu(I_{w_{|J}}) \mu(I_{\sigma^Jw}),$$
and  we use some of the inequalities we proved above. 

First, combining part (2) of Theorem \ref{cantor}  and Lemma \ref{lemmino}, for all $J$ large enough, for all $W\in \mathcal N_J(x)$, one has:
\begin{itemize}
\item for all $J\le j \le J/\weta_{N,-1}$, $ \bigcup_{W\in \mathcal{N}_J(x)} \mathcal{S}_{j }(\eta,W)=\emptyset;$

\item
for all $-1\le k\le N-1$, for all $J/\weta_{N,k}\le j \le J/\weta_{N,k+1}$, for all $w\in \mathcal{S}_j(\eta,W)$,    
$$\mu(I_{\sigma^J w})\le 2^{-J({\Ht_\ell}(\eta_{N,k})-\ep_N)} \le 2^{-J (\eta_{N,k}^{-1}-1)\ep_N}  2^{- J{\Ht_\ell}(\weta)} ;$$
\item
 if $j>J/\eta_0$, for all $w\in \Sigma_{j}$ such that $I_w\subset I_{W}$, we have  
 $$\mu(I_{\sigma^J w})\le 2^{-J (\eta_0^{-1}-1)(H_{\min}-\epsilon_N)} \leq \begin{cases} 2^{- J (\eta_0^{-1}-1)(H_{\min} - \epsilon_N)}\le 2^{   -J \big({\Ht_\ell}(\weta) - (\eta_0^{-1}-1)\epsilon_N\big)} & \mbox{if } \eta_\ell=0\\ 2^{-j(\eta_\ell^{-1}-1)({H_\ell}(\eta_\ell)-\epsilon_N)}\le 2^{ -J \big( {\Ht_\ell}(\weta)-   (\eta_\ell^{-1}-1)\epsilon_N\big)}& \mbox{if } \eta_\ell>0.   \end{cases}.
 $$
\end{itemize}

Secondly, since $x\in G_H \subset  \widetilde E_\mu(\alpha_H)$, there exists a sequence $(\widetilde\ep_J)_{J\geq 1}$ (depending on $x$) tending to 0 as $J \to +\infty$ such that  $x_{|J} \in \mathcal{E}_\mu(J,\alpha_H\pm \widetilde\ep_J)$. In particular, one has for $I_W\in \mathcal{N}_J(x)$, 
$$    \mu(I_W) \leq 2^{-J (\alpha_H -\wep_J)}.$$

When $\etal=0$, combining  the previous inequalities,   we get 
\begin{eqnarray*}
\M_\mu(I_J( x))& =  & \max \{\mu(I_w) :w\in \mathcal{S}_j(\eta,W),  \ W  \in \mathcal{N}_J(x)\}\\
& \leq & C \cdot \max\{\mu(I_{W}): I_W \in\mathcal N_J(x)\} \cdot  \max \{\mu(I_{\sigma^J w}) : w\in \mathcal{S}_j(\eta,W),  \ W  \mathcal{N}_J(x) ) \}\\
&\le & C 2^{-J\big(\alpha_H-\widetilde \ep_J + {\Ht_\ell}(\weta) -(\eta_{0}^{-1}-1)\ep_N\big)}.
\end{eqnarray*}
Consequently,  
$$
{\underline \dim_\locloc}(\M_\mu,x)  \ge {\alpha_H+{\Ht_\ell}(\weta)- (\eta_{0}^{-1}-1)\ep_N=H-(\eta_{0}^{-1}-1)\ep_N}.
$$
This holds for all $N>1/\eta$ hence ${\underline \dim_\locloc}(\M_\mu,x) \ge H$, for every $x\in G_H$.  The same estimate are true when $\etal>0$ by replacing $\eta_0$ by $\etal$.

On the other side, by Proposition~\ref{majexp} we know that $  {\overline \dim_\locloc}(\M_\mu,x)  \le   {\overline \dim_\locloc}(\mu,x)  +{\Ht_\ell}(\weta)=\alpha_H+{\Ht_\ell}(\weta)=H$, hence $\underline  \dim(\M_\mu,x)=H$ {(in fact we obtained that $G_H\subset E_{\M_\mu}(H)$)}. 
\end{proof}

We can now conclude. Recall that  with probability 1, simultaneously  for all $H\in [{{H_\ell}}(\widetilde\eta)+{\Ht_\ell}(\widetilde\eta), {H_{\max}}+{\Ht_\ell}(\widetilde\eta)]$, we have $\nu_{\alpha_H}(G_H)=1$, so one has $\dim G_H \geq {D_\mu}(\alpha_H)={D_\mu}(H-{\Ht_\ell}(\widetilde\eta))$. 
Finally, since $G_H\subset    \underline E_{\M_\mu}(H) $, one has 
$$\dim\,  \underline E_{\M_\mu}(H) \ge \dim\,  G_H \ge  {D_\mu}(H-{\Ht_\ell}(\widetilde\eta)) .$$

\begin{remark}\label{remmin}
Observe that the previous arguments give also a lower bound for the Hausdorff dimension of the level sets of the limit local dimension:    for  any $H\in  [{H_{\min}}+{\Ht_\ell}(\widetilde\eta), {H_{\max}}+{\Ht_\ell}(\widetilde\eta)]$,   $\dim\, E_{\M_\mu}(H)\ge {D_\mu}(H-{\Ht_\ell}(\widetilde\eta))$.
\end{remark}

\subsection{The middle part of the spectrum $D_{\M_\mu}$}\label{sharplow2}    

\ 

\mk

 Let $H \in  [ {{H_\ell}}(\widetilde\eta), {{H_\ell}}(\widetilde\eta)+{\Ht_\ell}(\widetilde\eta)]$. We apply Theorem ~\ref{ubi}   with the parameters:
\begin{itemize}
\sk\item
$\eta'=\widetilde\eta$, 
\sk\item 
the family $\mathcal{F}_{\weta} = (x_n(\weta), r_n(\weta))_{n\geq 1}$,
\sk\item
$\alpha={H_\ell}(\widetilde\eta)$,
\sk\item
 $\delta= {{H_\ell}}(\widetilde\eta)/ (\widetilde\eta H)$ (which does belong
 to $[1,1/\weta]$).
\end{itemize}

There exists a sequence $\widetilde\beta:=(\widetilde\beta_n)_{n\geq 1}$ and a Borel probability measure $\nu_{\alpha,\delta}$ supported on the  set  $U_\mu\big ({H_\ell}(\weta),\delta,\mathcal{F}_{\weta} , \widetilde\beta \big )$  and such that 
$$ {\dim}\  \nu_{\alpha,\delta}  \ge\frac{\dim \mu_{{{H_\ell}}(\widetilde\eta)}   }{\delta} =  \weta H \frac{{D_\mu}({H_\ell}(\weta))}{{H_\ell}(\weta)} = D(H) .$$

\begin{lemma}
\label{lemminor2}
One has  $ U_\mu \big({H_\ell}(\weta),\delta,\mathcal{F}_{\weta} , \widetilde\beta \big) \subset  \underline E^\leq _{\M_\mu}(H)$.
\end{lemma}

\begin{proof}
Let $x\in U_\mu \big ({H_\ell}(\weta),\delta,\mathcal{F}_{\weta} , \widetilde\beta\big)$. By definition of this limsup set, there is an increasing  sequence of integers $(j_k)_{k\ge 1}$ and words $w_k\in \mathcal{S}_{j_k}(\eta) \cap \Rmu (j_k, \widetilde \eta, {H_\ell}(\weta)\pm \widetilde \beta_{j_k}) \cap \Tl (j_k, \widetilde \eta,\ep^3_{j_k}) $ such that for each $k\ge 1$, $x\in B\big(x_{w_k },(2\cdot 2^{-\lfloor \widetilde\eta j_k\rfloor})^\delta\big)$. In other words,   $w_k$ satisfies  
$$
\begin{cases}
 \ \ 2^{-\lfloor \widetilde \eta j_k\rfloor({H_\ell}(\weta)+\widetilde \beta_{\lfloor \widetilde \eta j_k\rfloor})} \ \  \, \le\  \mu(I_{{w_k }_{|\lfloor \widetilde \eta j_k\rfloor}}) \ \le  \ 2^{-\lfloor \widetilde \eta j_k\rfloor({H_\ell}(\weta)-\widetilde \beta_{\lfloor \widetilde \eta j_k\rfloor})}\\
  \ \ 2^{-(j_k-\lfloor \weta j_k\rfloor)({{H_\ell}}(\widetilde \eta)+\ep^3_{j_k})}\le  \ \mu(I_{\sigma^{\lfloor\widetilde \eta j_k\rfloor} {w_k }})\le  \ 2^{-(j_k-\lfloor \widetilde \eta j_k\rfloor)({{H_\ell}}(\widetilde \eta)-\ep^3_{j_k})}.
\end{cases}
$$

Consider for each $k\ge 1$ the  largest integer $J_k$ such that $2^{-J_k} \geq    
(2\cdot 2^{-\lfloor \widetilde \eta j_k\rfloor})^{\delta}$.  With such a choice, one has $I_{w_k}\subset \mathcal{N}_{J_k}(x)$, so that $\M_{J_k}(x)    \geq   \mu(I_{w_k })$. Since $J_k = \delta \widetilde \eta j_k +o(1/k)$, one concludes that 
\begin{eqnarray*}
\M_{J_k}(x)    \geq    \mu(I_{w_k }) \geq C^{-1}2^{-\lfloor \widetilde \eta j_k\rfloor({H_\ell}(\weta)+\widetilde \beta_{\lfloor \widetilde \eta j_k\rfloor})}2^{-(j_k-\lfloor \weta j_k\rfloor)({{H_\ell}}(\widetilde \eta)+\ep^3_{j_k})}\geq 2^{-\frac{J_k}{\delta}( {H_\ell}(\weta)+  \widehat{\beta}_k) },
\end{eqnarray*}
for some sequence $\widehat{\beta}_k$  converging to 0 as $k\to +\infty$. Taking  the liminf as $k\to + \infty$ on both sides yields ${\underline \dim_\locloc}(\M_\mu,x) \leq H$. 
\end{proof}

From the previous lemma, we deduce that 
$$ \nu_{\alpha,\delta} \Big (E^\leq _{\M_\mu}(H)  \Big) \geq \nu_{\alpha,\delta} \big (U_\mu \big({H_\ell}(\weta),\delta,\mathcal{F}_{\weta} , \widetilde\beta \big) >0.$$

But for any $H'<H$, applying formula \eqref{majdimsss} of Remark \ref{rk1s}, one knows that   $\dim E^\leq _{\M_\mu}(H') \leq D(H') < D(H)$. In addition, Theorem \ref{ubi} asserts that $  \nu_{\alpha,\delta} \big (E^\leq _{\M_\mu}(H') \big) =0$.   We deduce that 
$$  \nu_{\alpha,\delta}   \left(E^\leq _{\M_\mu}(H) \setminus \bigcup_{n\geq 1} 
  \underline E^\leq_{\M_\mu} (H-1/n) \right) =1.$$
  Since $\underline E_{\M_\mu}(H) = E^\leq _{\M_\mu}(H) \setminus \bigcup_{n\geq 1} 
 \underline E^\leq_{\M_\mu}(H-1/n) $, we conclude that  $  \nu_{\alpha,\delta} \big (\underline E_{\M_\mu}(H)\big) =1$, i.e.  $D_{\M_\mu}(H)  \geq D(H)$. Since we already proved the converse inequality, equality holds.

\subsection{The   left part of the spectrum $D_{\M_\mu}$}\label{sharplow3}

$\ $
\mk

 Let $H \in   \big [ {{H_\ell}}(\eta_\ell),{{H_\ell}}(\widetilde\eta) \big )$. Recall the definition \eqref{defetaH} of $\eta_H$ when $H\in ({{H_\ell}}(\eta_\ell),{{H_\ell}}(\widetilde\eta))$.
When $H=  {{H_\ell}}(\eta_\ell)$, we set  $\eta_{{H_\ell}(\eta_\ell)}=\eta_\ell$.
   Let $(H^{(p)})_{p\ge 1}$ be a decreasing sequence of  real numbers in the interval  $ \big ({{H_\ell}}(\eta_\ell),{{H_\ell}}(\widetilde\eta) \big)$ converging to $H$, with the constraint that $\eta_{H^{(p)}} \in \mathcal{D}_\ell$. For each $p\ge 1$, we consider  any sequence $(\delta^{(p)})_{p\ge 1}$ converging to $ 1/ \eta_H $ as $p\to + \infty$,  and such that  the sequence of real numbers $\displaystyle \left ( \frac{{D_\mu}(H^{(p)})}{ \delta^{(p)}} \right )_{p\ge 1}$ is non increasing.
 
We apply the second part of Theorem \ref{ubi}:  there exist a collection of  positive sequences $ \big (\widetilde\beta^{(p)}:=(\widetilde\beta^{(p)}_n )_{n\ge 1} \big )_{p\ge 1}$    converging to 0, such that the  set $\bigcap_{p\ge 1} U_{\mu}\big (H^{(p)},\delta^{(p)}, \mathcal{F}_{\eta_{H^{(p)}}}, \widetilde\beta^{(p)} \big ) $ supports a measure  $\widetilde \nu_H$, whose dimension is greater than   or equal to
$$\inf_{p\geq 1} \frac{{D_\mu}(H^{(p)})}{\delta^{(p)}} = \eta_H  {D_\mu}(H)= D(H).$$
 
\mk

{Also, similarly to what was done in Section \ref{sharplow2}, we can get 
$$\bigcap_{p\geq 1}U_{\mu}\big (H^{(p)},\delta^{(p)}, \mathcal{F}_{ \eta_{H^{(p)} }}, \widetilde\beta^{(p)} \big )  \ \subset \   \underline E^\leq _{\M_\mu}(H)$$
and $ \widetilde \nu_H \left( \underline E^\leq_{\M_\mu}(H-1/n) \right) =0$ for all $n\ge 1$.  This yields
$$\widetilde \nu_H \left(E^\leq _{\M_\mu}(H) \setminus \bigcup_{n\geq 1} 
  \underline E^\leq_{\M_\mu}(H-1/n) \right) =1,$$
and finally that  $D_{\M_\mu}(H) =D(H)$.}

 \subsection{Proof of  Part (1) of Theorem \ref{cantor}}
 \label{sec5proof}
 
 If $\alpha \in (H_{\min},H_{\max})$, we can choose for $\nu_\alpha$  the Gibbs measure $\mu_\alpha$ of item (4) of Proposition~\ref{fm}  and it is not too difficult  to prove the desired property by using natural coverings because the Hausdorff dimension of $\mu_\alpha $ is positive.  
 
\medskip

We give a construction of a measure $\nu_\alpha$ that works for   $\alpha\in\{H_{\min},H_{\max}\}$, based on a concatenation method. It is also possible to adapt this method to get   another   choice for $\nu_\alpha$ when  $\alpha \in (H_{\min},H_{\max})$ (as explained at the end of the proof). 
 
\medskip

Due to Lemma~\ref{lem2} we can fix a positive sequence $(\ep_j)_{j\ge 1}$ converging to 0, such that, with probability 1, for $j$ large enough,  for all $W\in\Sigma_{\lfloor \eta j \rfloor }$, 
\begin{equation}
\label{eq13sss}
\#   \mathcal{S}_j(\eta,W)  \le 2^{\eta j \ep_j}.
\end{equation}
 Without loss of generality we can assume that $(\ep_j)_{j\ge 1}$ is non-increasing, $1/j\le \ep_j\le d$ for all $j\ge 1$, and  $\ep_{j+1}/\ep_j$ converges to 1 as $j\to +\infty$.

\mk

We treat the case $H_{\min}$, the case $H_{\max}$ is identical.
\subsubsection{Construction of the measure $\nu_{H_{\min}} $ and an associated Cantor set $\mathcal{C}_{H_{\min}}$}

 Let $(q_k)_{k\ge 1}$ be an increasing  sequence of real numbers, and let $\alpha_k:=\tau'_\mu(q_k)$. We assume that :
 \begin{itemize}
 \sk\item
  If ${D_\mu}(H_{\min})=0$, we choose $q_k$ such that ${D_\mu}(\alpha_{k})=\sqrt{\ep_k}$, for every $ k\ge 1$.
  
Hence $(q_k)$  satisfies $\lim_{k\to+\infty} \alpha_k =H_{\min}$ and    $\lim_{k\to\infty}{D_\mu}(\alpha_k)=0 $.

\sk
\item If ${D_\mu}(H_{\min})>0$, we choose $q_k$ such that   $\lim_{k\to+\infty} \alpha_k =H_{\min}$ and one also has  $\lim_{k\to+\infty}{D_\mu}(\alpha_k)={D_\mu}(H_{\min})$.
\end{itemize}
In all cases, by construction we have 
\begin{equation}
\label{condeps}
|{D_\mu}(\alpha_{k+1})-{D_\mu}(\alpha_k)|= \theta_k \sqrt{\ep_{k}}\le \theta_k{D_\mu}(\alpha_{k}),
\end{equation}
with $\lim_{k\to +\infty} \theta_k=0$.

\mk

We start by selecting some intervals words at which $\mu$ and $\mu_{\alpha_k}$ have the desired scaling properties.

\mk
Recall that by item (4) of Proposition \ref{fm}, the measure Gibbs $\mu_{\alpha_k}$ satisfies 
$$\mu_{\alpha_k}\big (\widetilde E_\mu(\alpha_k)\big) = \mu_{\alpha_k}\big(\widetilde {E}{\mu_{\alpha_k}} \big({D_\mu}(\alpha_k) \big) \big) =1.$$

Hence, for all $k\ge 1$,  the sets
\begin{eqnarray*}
\mathcal{A}_J ^k& = &   \big \{ W\in\Sigma_J:  \, \forall \, W'\in \mathcal N(W),  \ W'\in \mathcal{E}_\mu \big (J,\alpha_k\pm\ep_k \big ) \big \}\\
\text{and}\quad {\mathcal{B}_J  ^k}& = &  \big \{ W\in\Sigma_J: \, \forall \, W'\in \mathcal N(W), \, W'\in \mathcal{E}_{\mu_{\alpha_k}} \big (J, {D_\mu}(\alpha_k)\pm\ep_k \big ) \big \} 
\end{eqnarray*}
satisfy $\lim_{J\to +\infty} \mu_{\alpha_k}(\mathcal{A}^k_J) = \lim_{J\to +\infty} \mu_{\alpha_k}(\mathcal{B}^k_J) =1$.
Up to extraction of a subsequence, one deduces that there exists an integer $J_k\in\N_+$ and a collection $\mathcal{W}_{k}$ of words of generation $J_k$ such that the cubes $I_W$, $W \in \mathcal{W}_{k}$, are pairwise disjoint, $\sum_{W\in   \mathcal{W}_{k}}\mu_{\alpha_k}(I_W)\ge e^{-\ep_k}$, and 
\begin{equation}\label{ak}
 \forall\, W \in  \mathcal{W}_{k},\, \forall \, W'\in \mathcal N(W),\,   \ W'\in   \mathcal{E}_\mu(J,\alpha_k\pm\ep_k)\cap  \mathcal{E}_{\mu_{\alpha_k}}(J, {D_\mu}(\alpha_k)\pm\ep_k).
 \end{equation}
 
  Now, let $(N_k)_{k\ge 1}$ be an increasing sequence of integers such that for all $k\ge 1$, 
\begin{eqnarray}
\label{eq12sss}
\sum_{p=1}^{k-1} N_p J_p \max\big(1,\alpha_p+2\ep_p,{D_\mu}(\alpha_p)+2\ep_p\big)&\le  & \ep_k N_k J_k, \\
  \nonumber \frac{J_{k+1}}{N_{k} J_k}\max\big(1,\alpha_{k+1}+2\ep_{k+1},{D_\mu}(\alpha_{k+1})+2\ep_{k+1}\big) & \le  &\ep_{k+1}\alpha_k.
\end{eqnarray}

We also define $\widetilde J_k=\sum_{p=1}^{k} N_p J_p$, which satisfies
$$N_kJ_k\leq \widetilde J_k  \leq N_kJ_k(1+\ep_k).$$

Then we define  recursively a Cantor-like set $\mathcal{C}_{H_{\min}}$ and simultaneously a  Borel probability measure $\nu_{H_{\min}}$ on $[0,1]^d$ supported on $\mathcal{C}_{H_{\min}}$. To do so, we use a  construction by concatenation: the measure $\nu_{H_{\min}}$ behaves like $\mu_{\alpha_k}$ between the generations  $\widetilde J_{k-1} +1$ and $\widetilde J_k$. More precisely:

\sk
- We set $I_\emptyset = \zu^d$ and  $\nu_{H_{\min}} (\zu^d)=1$, 

\sk
-  For every $k\geq 1$,  we write  $ \widetilde W_k \in  \mathcal{W}_{ k}^{N_k}$ as 
  $\widetilde W_k=W_{k,1}\cdots W_{k,N_k}$ where $W_{k,i}\in  \mathcal{W}_k\subset \Sigma_{J_k}$,

\sk
- The Cantor set is
$$
\mathcal{C}_{H_{\min}}= \ \bigcap_{k\ge 1} \ \bigcup_{(\widetilde W_1,\ldots,\widetilde W_k)\in \mathcal{W}_{ 1}^{N_1}\times\cdots\times  \mathcal{W}_{ k}^{N_k}} I_{\widetilde W_1\cdots \widetilde W_k}
$$

\sk
-  The measure $\nu_{H_{\min}}$ is defined recursively as follows: 
 for every $k\geq 1$, for every   $(\widetilde W_1,\ldots,\widetilde W_k)\in \mathcal{W}_{ 1}^{N_1}\times\cdots\times  \mathcal{W}_{ k}^{N_k}$,  
we set  for every $i\in \{1,..., N_k\}$
$$
\nu_{H_{\min}}\Big (I_{\widetilde W_1\cdots  \widetilde {W}_{k-1} W_{k,1}\cdots W_{k,i-1}W_{k,i}} \Big )=\nu_{H_{\min}} \Big (I_{\widetilde W_1\cdots \widetilde  {W}_{k-1} W_{k,1}\cdots W_{k,i-1}}\Big) \frac{\mu_{\alpha_k}(I_{W_{k,i}})}{\sum_{W'_{k}\in \mathcal{W}_{ k}}\mu_{\alpha_k}(I_{W'_k})}.
$$
It is clear that this measure $\nu_{H_{\min}}$, defined only on the cubes appearing in the Cantor's construction, uniquely extends to a Borel probability measure on the cube $[0,1]^d$.

\subsubsection{Properties of the measure $\nu_{H_{\min}} $}
We first prove that the Cantor set lies on the elements $x\in \zu^d$ satisfying simultaneously $ \underline{\dim}(\mu,x) = H_{\min}$ and   $ \underline{\dim}(\nu_{H_{\min}},x) = D_\mu(H_{\min})$.

\begin{lemma}
One has $\mathcal{C}_{H_{\min}}\subset \widetilde E_\mu \big(H_{\min}\big) \cap \widetilde E_{\nu_{H_{\min}}} \big( {D_\mu}({H_{\min}})\big)$.
\end{lemma}
\begin{proof}
If $k\ge 2$ and $\widetilde J_k< J \le \widetilde  J_{k+1}$,   we   set $k_J=k$. 

Fix $x\in \mathcal{C}_{H_{\min}}$. 
Using the quasi-Bernoulli property \eqref{quasib} of $\mu$ and the inequalities \eqref{ak}, one gets that 
  for every $k\ge 2$ and $\widetilde J_k< J \le \widetilde  J_{k+1}$, and for every cube $W\in \mathcal{N}_J(x)$,
\begin{eqnarray*}
&&
2^{-\big (N_1J_1(\alpha_1+\ep_1) +.... +N_{k}J_{k}(\alpha_{k}+\ep_{k}) +  (J-\widetilde J_{k})( \alpha_{k+1}+\ep_{k+1} ) \big)}\\
&&\le\mu(I_W)\le 2^{-\big (N_1J_1(\alpha_1-\ep_1) +.... +N_{k}J_{k}(\alpha_{k}-\ep_{k}) +  (J-\widetilde J_{k})( \alpha_{k+1}-\ep_{k+1} ) \big)}.
\end{eqnarray*}
Using the equations \eqref{eq12sss}, we deduce that
\begin{eqnarray*}
 2^{-\big ((\alpha_k+2\ep_k )N_{k}J_{k}+  (J-\widetilde J_{k})( \alpha_{k+1}+\ep_{k+1} ) \big)}  \le\mu(I_W)\le 2^{-\big(  (\alpha_k-2\ep_k ) N_{k}J_{k} +  (J-\widetilde J_{k})( \alpha_{k+1}-\ep_{k+1} ) \big)}, 
\end{eqnarray*}
which yields 
\begin{eqnarray*}
 2^{-\big ((\alpha_k+2\ep_k )(1+\ep_k )\widetilde J_k+  (J-\widetilde J_{k})( \alpha_{k+1}+\ep_{k+1} ) \big)}   \le\mu(I_W)\le 2^{-\big(  (\alpha_k-2\ep_k )(1+\ep_k )\widetilde J_k +  (J-\widetilde J_{k})( \alpha_{k+1}-\ep_{k+1} ) \big)}.
\end{eqnarray*}
Since $\alpha_k \to H_{\min} $ when $k\to+\infty$, we deduce that $\lim_{J\to +\infty} \frac{\log_2 \mu(I_W)}{-J} = H_{\min}$, where $W\in \mathcal{N}_J(x)$. This proves that $x  \in  \widetilde E_\mu(H_{\min}) $.

Similarly,   the same arguments show that   for every $k\ge 2$ and $\widetilde J_k< J \le \widetilde  J_{k+1}$, and for every cube $W\in \mathcal{N}_J(x)$,
\begin{eqnarray*}
&&2^{-\big (({D_\mu}(\alpha_k)+2\ep_k )(1+\ep_k )\widetilde J_k+  (J-\widetilde J_{k})({D_\mu}( \alpha_{k+1})+\ep_{k+1} ) \big)} \\
&& \le\nu_{H_{\min}}(I_W)\le 2^{-\big(  ({D_\mu}(\alpha_k) - 2\ep_k )(1+\ep_k )\widetilde J_k +  (J-\widetilde J_{k})( {D_\mu}(\alpha_{k+1} )-\ep_{k+1} ) \big)}.
\end{eqnarray*}
The equation  \eqref{condeps} then yields
\begin{equation}\label{nuH}
2^{-J {D_\mu}(\alpha_{k}) (1+\widetilde\theta_k )}\le \nu_{H_{\min}} \big(I_W \big)\le 2^{- J  {D_\mu}(\alpha_{k})(1-\widetilde\theta_k )},
\end{equation}
for some decreasing sequence $\widetilde\theta_k$, tending to $ 0$ when $k\to +\infty$. This yields that  $x\in  \widetilde E_{\nu_{H_{\min}}} \big( {D_\mu}({H_{\min}})\big)$, since  ${D_\mu}(\alpha_{k}) \to {D_\mu}( {H_{\min}})$ when $k\to +\infty$.
\end{proof}
Observe also that  \eqref{nuH} implies that for each $j$ large enough, we have  
\begin{equation}\label{u}
\#\{ W \in \Sigma_j : I_{W}\cap \mathcal{C}_{H_{\min}} \neq\emptyset\}\le 2^{j {D_\mu}(\alpha_{k_j})(1+\widetilde\theta_{ k_j})}.
\end{equation}

\subsubsection{Proof that well-approximated points have $\nu_{H_{\min}} $-measure 0}

Fix an approximation rate $\delta>1$.    To get the result, one focuses first on the value of
$$\nu_{H_{\min}} \Big (\bigcup_{  w\in   \mathcal{S}_j(\eta)} B(x_w,(2\cdot 2^{-\lfloor \eta j \rfloor})^\delta)\Big) =  \, \nu_{H_{\min}} \Big (\bigcup_{ W \in\Sigma_{{\lfloor \eta j \rfloor} }}\bigcup_{w\in   \mathcal{S}_j(\eta,W)} B(x_w,(2\cdot 2^{-\lfloor \eta j \rfloor})^\delta)\Big)  .$$

For each $j$ large enough, consider $W\in \Sigma_{\lfloor j \eta\rfloor}$ such that $I_{W }\cap\mathcal{C}_{H_{\min}}\neq\emptyset$. One looks {for points   $x\in I_{W}\cap \mathcal{C}_{H_{\min}}$   such that  $x\in B(x_w, (2\cdot 2^{-\lfloor \eta'j \rfloor})^\delta)$ for some  surviving word  $w\in  \mathcal{S}_j(\eta,W)$.} Hence, one sees that 

\begin{eqnarray*}
&&\nu_{H_{\min}} \Big (\bigcup_{W \in\Sigma_{\lfloor \eta j  \rfloor }}\bigcup_{w\in   \mathcal{S}_j(\eta,W )} B(x_w,(2\cdot 2^{-\lfloor \eta j \rfloor})^\delta)\Big )\\
&\le & \sum_{W \in \Sigma_{\lf \eta j \rf} :  \, I_{W }\cap \mathcal{C}_{H_{\min}}\neq\emptyset} \ \sum_{\substack{ w\in   \mathcal{S}_j(\eta,W), \\ \, I_w \cap \mathcal{C}_{H_{\min}} \neq\emptyset}}\nu_{H_{\min}}\left ( B(x_w,(2\cdot 2^{-\lfloor \eta j \rfloor})^\delta) \right).
\end{eqnarray*}
Recall that by  \eqref{eq13sss}, the number of such possible surviving vertices $w$ (in the second sum above) is bounded from above by $2^{\eta j \ep_j}$.   Applying  \eqref{u} to $ B(x_w,(2\cdot 2^{-\lfloor \eta j \rfloor})^\delta)$ for the generation $J =\lfloor \eta j \delta\rfloor$, we get  
{\begin{eqnarray*}
&&\nu_{H_{\min}} \Big (\bigcup_{W \in\Sigma_{\lfloor \eta j  \rfloor }}\bigcup_{w\in   \mathcal{S}_j(\eta,W)} B(x_w,(2\cdot 2^{-\lfloor \eta j \rfloor})^\delta)\Big )\\
&\le & \Big ( \# \{W\in \Sigma_{\lfloor \eta j\rfloor}: I_{W}\cap\,  \mathcal{C}_{H_{\min}} \neq\emptyset\} \Big)  \cdot  2^{\eta j  \ep_j}  \cdot  2^{-\lfloor \eta j \delta\rfloor {D_\mu}(\alpha_{k_{\lfloor \eta j \delta\rfloor}})(1-\widetilde\theta_{k_{\lfloor \eta j \delta\rfloor}})}\\
&\le&  2^{\eta j \ep_j+ \lfloor \eta j \rfloor {D_\mu}(\alpha_{k_{\lfloor \eta j \rfloor}}) (1+\widetilde\theta_{k_{\lfloor \eta j \rfloor}}) -\lfloor \eta j \delta\rfloor({D_\mu}(\alpha_{k_{\lfloor \eta j \delta\rfloor}})(1-\widetilde\theta_{k_{\lfloor \eta j \delta\rfloor}})}.
\end{eqnarray*}}
It follows now from the properties imposed to the sequences $(\varepsilon_k)_{k\ge 1}$ and $(\alpha_k)_{k\ge 1}$ that 
$$
{ \xi_j:= \nu_{H_{\min}} \Big (\bigcup_{  w\in   \mathcal{S}_j(\eta)} B(x_w,(2\cdot 2^{-\lfloor \eta j \rfloor})^\delta)\Big) } \le C'2^{\eta j (1-\delta) {D_\mu}(\alpha_{k_{\lfloor \eta j \rfloor}}) (1+o(1))},
$$
where $C  $ is another constant coming from the fact that we dropped some integer parts. 

\begin{itemize}
\sk\item
When  ${D_\mu}(H_{\min})>0$, it is direct that {the series  $\displaystyle \sum_j \xi_j$ converges}. 
\sk\item
When  ${D_\mu}(H_{\min})=0$, for large values of $j$ one has by construction $j>k_{\lfloor \eta j  \rfloor}$, so $j^{-1}\le\ep_j\le \ep_{k_{\lfloor \eta j \rfloor}}$, and ${D_\mu}(\alpha_{k_{\lfloor \eta j \rfloor}})= \sqrt{ \ep_{k_{\lfloor \eta j \rfloor}}}$. Thus, for $j$ large enough we get 
$$
2^{\eta j (1-\delta) {D_\mu}(\alpha _{k_{\lfloor \eta j \rfloor}}) (1+o(1))}\le 2^{-\sqrt{j}\eta(1-\delta)(1+o(1))}, 
$$
hence the series {$\displaystyle \sum_j \xi_j$} still converges. 
\end{itemize}

Finally, {the Borel-Cantelli lemma   proves Part (1) of Theorem \ref{cantor}.} 
 
 \sk
 
  Observe that everything works similarly if we replace $H_{\min}$ with $H_{\max}$ and change the sequence $\alpha_k$ accordingly. When $\alpha \in (H_{\min}, H_{\max})$, we can even take the sequence $(\alpha_k)_{k\geq 1}$ to be constant (if not, this process gives other measures sitting on $\widetilde E_\mu(\alpha)$).

\subsection{Proof of  Part (2) of Theorem \ref{cantor}}

Recall Proposition~\ref{discretization}. Applying the  Borel-Cantelli lemma, it is enough to prove that for all integers $N\ge 1$ and $p>2(H_{\max}-H_{\min})^{-1}$,  
\begin{equation}\label{unifcont2}
\mathbb{E} \Big (\sup_{\alpha \in \{H_{\min},H_{\max}\}\cup  I_p} \  \sum_{J \ge 1}\sum_{W\in \Sigma_J }\nu_\alpha (I_W)\mathbf{1}_{\mathscr{C} (N,J,W)}\Big )<+\infty.
\end{equation}
where $ I_p=  [H_{\min}+1/p,H_{\max}-1/p]$. 

At first, notice that for any  Borel probability measure $\nu$ on $[0,1]$ we have 
\begin{eqnarray*}
\mathbb{E}\Big (\sum_{J\ge 1}\sum_{W\in \Sigma_J}\nu(I_W)\mathbf{1}_{\mathscr{C} (N,J,W) }\Big ) & =  & \sum_{ J \ge 1}\sum_{W\in \Sigma_J}\nu(I_W)\mathbb{P}\big (\mathscr{C}(N,J,W)\big )\\
& \le & \sum_{J \ge 1}2^{-J \ep_N} \sum_{W\in \Sigma_J}\nu(I_W )<\infty.
\end{eqnarray*}
Applying this to $\nu_{H_{\min}}$ and $\nu_{H_{\max}}$ constructed above, it remains us to prove \eqref{unifcont2} only with the interval  $I_p$.

Recall that when $\alpha\in  I_p\subset (H_{\min},H_{\max})$, one can take $\nu_\alpha=\mu_\alpha$ {(where $\mu_\alpha$ is the Gibbs measure ofProposition \ref{fm})} .

Let us write the interval  $  I_p$ as $   I_p =\tau_\mu'([q'_p,q_p])$, for some real numbers {$q_p> q'_p$}. Recall that the Gibbs capacity  $\mu$ is associated with a H\"older potential $\phi$ which belongs to the $C^\beta$ H\"older class, for some $\beta>0$. Standard arguments based on the bounded distorsion property give that for $\kappa=2\|\phi\|_\infty/\log(2)$ and $C_{q,q'} =e^{\frac{(|q|+|q'|) C}{(1-2^{-\beta})}}$, for all $q,q'\in {[q'_p,q_p]}$ and $W\in\Sigma^*$, setting $\alpha _q=\tau_\mu'(q)$, we have 
$$
\mu_{\alpha_{q'}}(I_W)\le C_{q,q'}2^{{\kappa |q-q'| \cdot|W|}} \mu_{\alpha_{q}}(I_W).
$$
 
The interval $  I_p$ being compact, one can extract a finite collection of intervals $[\alpha_{\widetilde q_{k-1}},\alpha_{\widetilde q_{k}}]$, $1\le k\le K$, such that {$q_p=\widetilde q_{0}> \ldots >\widetilde q_K=q'_p$} and $|\widetilde q_k-\widetilde q_{k-1}|\le \epsilon_N/ (2\kappa)$. Setting $C_k=\sup_{q'\in  {[\widetilde q_{k},\widetilde q_{k-1}]}  }C_{q',\widetilde q_k}$,  one rewrites the above properties as follows: for all $W\in\Sigma^*$, for all $1\le k\le K$,
$$
\sup_{q'\in {[\widetilde q_{k},\widetilde q_{k-1}]}  } \mu_{\alpha_{q'}}(I_W)\le   C_k 2^{ |W| \ep_N/2} \mu_{\alpha_{q_k}}(I_W).
$$
From these considerations, we get for $1\le k\le K$, 
\begin{eqnarray*}
&&\mathbb{E} \Big (\sup_{\alpha\in [\alpha_{\widetilde q_{k-1}},\alpha_{\widetilde q_{k}}] } \sum_{J\ge 1}\sum_{W\in \Sigma_J}\nu_\alpha (I_W)\mathbf{1}_{\mathscr{C}(N,J,W)}\Big )\\
&\le &\mathbb{E} \Big ( \sum_{J\ge 1}\sum_{W\in \Sigma_J }\sup_{\alpha \in [\alpha_{\widetilde q_{k-1}},\alpha _{\widetilde q_{k}}]}  \nu_\alpha (I_W)\mathbf{1}_{\mathscr{C}(N,J,W)}\Big )\\
&\le &C_k \sum_{J\ge 1} \sum_{W\in \Sigma_J} 2^{J\ep_N/2}\nu_{\alpha_{q_k}}(I_W)\mathbb{P} \big(\mathscr{C}(N,J,W)\big)\\
&\le &C_k  \sum_{J \ge 1} 2^{-J\ep_N/2}<+\infty.
\end{eqnarray*}
It follows that 
$$
\mathbb{E} \Big (\sup_{\alpha \in I_p} \Big (\sum_{J\ge 1}\sum_{W\in \Sigma_J}\nu_\alpha(I_W)\mathbf{1}_{\mathscr{C}(N,J,W)}\Big )  \Big ) \le \sum_{k=1}^K C_k \sum_{J\ge 1} 2^{-J\ep_N/2}<+\infty,
$$
i.e. \eqref{unifcont2} holds. 

\section{Free energy and large deviations for $\M_\mu$}\label{free}

{Recall the definitions \eqref{defqeta} and \eqref{defqetap} for $q_{\weta}$ and $q_{\eta_\ell}$, and also formula \eqref {formtaut2} for $
\widetilde\tau(q)$ that we reproduce for convenience:
$$\widetilde\tau(q)=
\begin{cases}
\ \tau_\mu(q)+{\Ht_\ell}(\widetilde\eta)q&\text{if } q\le {q_{\widetilde \eta}},\\
\ \tau_\mu(q)+d(1-\eta)&\text{if }{q_{\widetilde \eta}}<q<{q_{\eta_\ell}},  \\
\ {{H_\ell}} (0) q&\text{if ${q_{\eta_\ell}}<\infty$ and $q\ge {q_{\eta_\ell}}$}.\end{cases}$$}

{In this section we prove that, with probability 1:
\begin{itemize}
\sk\item   (Section \ref{sec91}) the Legendre transform of $\tilde\tau$ equals $D_{\Mm}$,
\sk\item     (Section \ref{sec92})  the lower $L^q$-spectrum $\tau_{\Mm}$ is bounded below by  $\tilde\tau$ ,
\sk\item (Section \ref{sec93}) the lower large deviation spectrum satisfies 
\begin{equation}\label{LD}
\underline{f}_{\M_\mu} (H) \ge  \tau_{\Mm}^*(H).
\end{equation}
 \end{itemize}}

\noindent
{\noindent Using that $D_{\Mm}(H)  \leq \tau_{\M_\mu}^*(H)$ holds true for every $H$,   the first result yields $\taut  \geq \tau_{\M_\mu}$.   \\
The second result ensures that there is in fact equality:  $\taut  = \tau_{\M_\mu}$.  \\
Since one always has $\underline{f}_{\M_\mu} (H) \le \overline{f}_{\M_\mu} (H)  \leq \tau_{\M_\mu}^*(H)$, we conclude that $D_{\Mm}(H) = \tau_{\M_\mu}^*(H) = \underline{f}_{\M_\mu} (H) = \overline{f}_{\M_\mu} (H) =\taut(H)$. In particular, $\M_\mu$ obeys the multifractal formalism.\\
Finally, by Varadhan's lemma (or in our situation very simple estimates),   the  free energy $\tau_{\M_\mu}(q)$ exists as a limit, not only as a liminf, for all $q\in\R$.}

{This completes Part (2) of Theorem \ref{thm-0} and   Theorem \ref{thm-2}.}

\subsection{Equality between ${\taut}^*$ and the singularity spectrum of $\M_\mu$}
\label{sec91}

  First, we prove that   $D_{\M_\mu}(H)$ (given by part (2) of Theorem \ref{thm-0}) 
 is indeed  the Legendre transform of $\widetilde \tau$.  
  \begin{lemma}
  \label{lemtaut}
{With probability one, one has $\widetilde \tau^*=D_{M_\mu}$.}
\end{lemma}

\begin{proof} 
 
Let us start with a few observations. By definition of ${H_\ell}(\weta)$ (see Figure \ref{figbetal}), $q_{\weta}$ is the slope of the tangent line to the graph of  $\tau_\mu^*$  at ${H_\ell}(\weta)$, and this tangent line passes through the point $(0, d(1-\eta))$. Hence   $   \tau_\mu^*({{H_\ell}}(\widetilde\eta))   - d(1-\eta)  = q_\weta {{H_\ell}}(\widetilde\eta) $. Recalling that  $ \tau_\mu^*({{H_\ell}}(\widetilde\eta))  =   d\frac{1-\eta}{1-\weta}$, one deduces  that 
$$\weta  \tau_\mu^*( {{H_\ell}}(\widetilde\eta)) = q_{\widetilde \eta}\tau_\mu'({q_{\widetilde \eta}})  .$$

But by the definition of the Legendre transform, one  has   $ \tau_\mu^*( {{H_\ell}}(\widetilde\eta)) = \tau_\mu^*(\tau_\mu'({q_{\widetilde \eta}})) = q_{\widetilde \eta}\tau_\mu'({q_{\widetilde \eta}}) - \tau_\mu({q_{\widetilde \eta}}) $. We deduce that  
$$ \tau_\mu(q_\weta) =   \tau_\mu^*( {{H_\ell}}(\widetilde\eta)) (\weta -1) = -d(1-\eta),$$
 i.e.  $ {\taut}({q_{\widetilde \eta}})=0$. 
Notice that ${\taut} $ is continuous on $\R$, and that  there is a first order phase transition at ${q_{\widetilde \eta}}$, since $ {H_\ell}(\weta) +{\Ht_\ell}(\weta) = {\taut}'({q_{\widetilde \eta}}^-) > {\taut}'({q_{\widetilde \eta}}^+)  =  {H_\ell}(\weta)   $. 

\sk

{$\bullet$ When $H\ge  {H_\ell}(\weta) $:  Since  $\widetilde\tau$ and $\tau_\mu$ differ by a  linear term of slope ${\Ht_\ell}(\widetilde\eta)$ over $(-\infty, {q_{\widetilde \eta}}]$, their Legendre transform are translated versions of each other by  ${\Ht_\ell}(\widetilde\eta)$ over the interval $[{\taut}'(q_\weta), +\infty) = [{H_\ell}(\weta) +{\Ht_\ell}(\weta) , +\infty)$. Hence, for      $ H\ge   {H_\ell}(\weta) +{\Ht_\ell}(\weta) $, one has $\widetilde\tau^*(H)=\tau_\mu^*(H- {\Ht_\ell}(\widetilde\eta)) = D_{\M_\mu}(H)$.    
}

\sk
 $\bullet$ When $H\in [{H_\ell}(\weta),{H_\ell}(\weta) +{\Ht_\ell}(\weta)]$: The discontinuity of $({\taut})'$ at $q_{\weta}$ implies that  for  $H$ in the interval $[{\taut}'({q_{\widetilde \eta}}^+) , {\taut}'({q_{\widetilde \eta}}^-)] =  [ {H_\ell}(\weta)  , {H_\ell}(\weta) +{\Ht_\ell}(\weta) ]$, one has
$$\widetilde \tau^*(H ) = \inf_{q\in \R} (qH  - {\taut}(q) ) =  {q_{\widetilde \eta}} H  - {\taut}({q_{\widetilde \eta}})  = {q_{\widetilde \eta}} H=D_{\M_\mu}(H).
$$

\sk
 {$\bullet$ When $\eta_\ell=0$ and $H\le  {H_\ell}(\widetilde\eta) $:
In this case we have $q_{\eta_\ell}<+\infty$. Since $\widetilde\tau$ and $\tau_\mu$ differ by the constant $d(1-\eta)$ over $[q_{\weta},q_{\eta_\ell}]$,  for $H\in \big[{\taut}'(q_{\eta_\ell}), {\taut}'(q_\weta)\big] = \big [{H_\ell}(0), {H_\ell}(\widetilde\eta) \big]$, one has $\widetilde\tau^*(H)={\tau^*_\mu}(H)-d(1-\eta)$. Then, when $q\ge q_{\eta_\ell}$,   $\widetilde\tau $ is  linear   with slope $\widetilde\tau'(q_{\eta_\ell}^-)$, so $\widetilde\tau^*(H)=-\infty$ for all $H<{H_\ell} (0)$. In all cases, ${\taut}^*(H) = D_{\M_\mu}(H)$.}

\sk
{$\bullet$ When
$\eta_\ell>0$ and $H\le  {H_\ell}(\widetilde\eta) $: Here $q_{\eta_\ell}=+ \infty$ and ${H_\ell}(\eta_\ell)=H_{\min}$. The same argument as above yields  $\widetilde\tau^*(H)={\tau^*_\mu}(H)-d(1-\eta)$ for all $H\le {H_\ell}(\widetilde\eta)$.  }
\end{proof}

We  know now  that $D_{\M_\mu}(H) = {\taut}^*(H)$, for all $H\in\R$.
Since the multifractal formalism states that  $D_{\M_\mu}(H)  \leq  {\tau}_{\M_\mu}^  *(H)$ for all $H\in\R$, one deduces that ${\taut}^* \leq  {\tau}_{\M_\mu}^  *$. By inverse Legendre transform, one gets
$$\mbox{ for all $q\in\R$ },  \ \ \ {\taut}(q) \geq  {\tau}_{\M_\mu}  (q).$$

The next section establishes that  {${\taut} \leq  {\tau}_{\M_\mu}$}, so that equality indeed holds almost surely.

\subsection{Lower bound for  {$ {\tau}_{\M_\mu}$} }
\label{sec92}

\

\subsubsection{ When  ${q_{\widetilde \eta}} < q < {q_{\eta_\ell}}$:}

The sub-multiplicativity property of $\mu$ gives for $j\ge 1$  
\begin{align*}
\sum_{ W \in \Sigma_J} \M_\mu (I_W)^q &=\sum_{W\in \Sigma_J}   \left(\max_{W' \in \mathcal N_J(W)} \max_{ w\in \mathcal{S}_j(\eta,W')  }   \mu(I_w)\right)^q\\
&\le 3^d \sum_{W\in \Sigma_J }\max_{ w\in \mathcal{S}_j(\eta,W)  } \mu(I_w)^q\\
&\le 3^d C^q \sum_{W\in \Sigma_J }\mu(I_W)^q \sum_{\substack{w\in \Sigma^*,  \, p_{Ww } =1}} \mu(I_{w})^q\\
&= 3^d C^q \sum_{W\in \Sigma_J}\mu(I_W)^q  \, \sum_{k\ge 0} \,  \sum_{ w \in \Sigma_k } \mu(I_{w})^q p_{Ww}.
\end{align*}

The random variables $p_{Ww}$ being independent, with law $B(2^{-d(J+k)(1-\eta)})$, this yields 
$$
\mathbb E \left (\sum_{ W\in \Sigma_J} \M_\mu (I_W )^q \right)\le 3^d C^q  \left (\sum_{W\in \Sigma_J}\mu(I_W)^q \right) \sum_{k\ge 0} 2^{-(J+k)d(1-\eta)}  \left (\sum_{w\in \Sigma_k}\mu(I_{w})^q\right ).
$$
Observe that a direct consequence of \eqref{quasib} is that  for some positive constant $C_q>0$, 
\begin{equation}
\label{upperbound}
 \sup_{k\ge 1}2^{k\tau_\mu (q)}   \sum_{w\in \Sigma_k}\mu(I_{w})^q\le C_q.
 \end{equation}
Consequently, 
$$
\mathbb E \left (\sum_{ W\in \Sigma_J} \M_\mu (I_W )^q \right) \le 3^d C_q C^q  \left ( 2^{-J d(1-\eta)}  \sum_{W\in \Sigma_J}\mu(I_W)^q \right) \sum_{k\ge 0} 2^{- k( \tau_\mu(q) +d(1-\eta)) }  .
$$
Since $q>{q_{\widetilde \eta}}$, we have $\tau_\mu(q)+d(1-\eta)>0$. Hence for some constant  $C'_q$ depending on $q$ only, 
$$
\mathbb E \left (\sum_{ W\in \Sigma_J} \M_\mu (I_W )^q \right) \le 3^d C'_q  \left ( 2^{-J d(1-\eta)}  \sum_{W\in \Sigma_J}\mu(I_W)^q \right)    .
$$
Finally, for every $\ep>0$, applying \eqref{upperbound} we get 
$$
\mathbb E \left (\sum_{J\ge 1} 2^{J(\tau_\mu(q)+d(1-\eta)-\ep)} \sum_{W\in \Sigma_J} \M_\mu (I_W)^q \right )\le 3^d C'_qC_q \sum_{J \ge 1} 2^{-J \ep},
$$
which is finite. We conclude that with probability 1, we have 
$$ {\liminf_{J\to +\infty} \frac{1}{J }\log_2 \sum_{W\in \Sigma_J} \M_\mu (I_W)^q}\le -\tau_\mu(q)-d(1-\eta),$$
i.e.  $ \tau_{\M_\mu}(q)\ge \tau_\mu(q)+d(1-\eta) = \tilde\tau(q)$. 

This holds for each ${q_{\widetilde \eta}} < q < {q_{\eta_\ell}}$ almost surely, and by concavity (hence continuity) of $\tau_{\M_\mu}$ and $\widetilde\tau$, this holds almost surely for all ${q_{\widetilde \eta}} \leq q \leq {q_{\eta_\ell}}$.

\subsubsection{ When $q\in(0,{q_{\widetilde \eta}})$:}

$\bullet$ Suppose for a while  that both $\eta_\ell$ and $\eta_r$ are positive.  

\sk 
  Fix $0<\ep<{{H_\ell}}(\etal)$, and two integers $N_\ell,N_r > 2/\ep$. Due to Proposition~\ref{p2} and the continuity of the mappings ${{H_\ell}}$ and ${H_r}$, there exists $j_0\ge 1$, and two sets of parameters $\eta_\ell= \eta_{\ell,1}<\ldots<\eta_{\ell,N_\ell}=\eta $ and $\eta_r=\eta_{r,1}<\ldots<\eta_{r,N_r} $ such that for $j \ge j_0$, if $w\in\mathcal{S}_{j}(\eta)$, then  $w\in \Tt(j,  \eta_{i,k}, \ep)$  for some  $i\in\{\ell,r\}$ and $1\le k \le N_{i}$, i.e.
 \begin{equation}\label{discret}
 (j-\lfloor \eta_{i,k} j \rfloor)({H_i}(\eta_{i,k})-\ep) \le -\log_2 \mu(I_{\sigma^{\lfloor \eta_{i,k} j\rfloor}w})\le (j-\lfloor \eta_{i,k} j  \rfloor)({H_i}(\eta_{i,k})+\ep).
\end{equation}

For $J$ large enough, given $W\in\Sigma_J$,  {assume that $\M_\mu(I_W)  $ is realized at $w$, i.e.  $\M_\mu(I_W) =\mu(I_{w})$,  for some word of length  $|w| = j \geq J$, $I_{w}\subset I_{W'}$ and $W'\in \mathcal{N} _{J}(W) $. In this case,} there exists $\eta_{i,k}$ such that \eqref{discret} holds. We have to distinguish the two following possibilities for the parameters $\{\eta_{i,k}\}_{i\in \{\ell,r\}, k\in \{1,..., \max(N_\ell,N_r)\}}$: 
\begin{itemize}
\sk\item[-]
if  $\lfloor \eta_{i,k}j \rfloor \le J$,  then  $I_{W} \subset \bigcup_{ u \in \mathcal {N}_{\lfloor \eta_{i,k}j \rfloor} ( w_{|\lfloor \eta_{i,k}j \rfloor})} I_{u} $. 

\sk\item[-]
  if $\lfloor \eta_{i,k} j  \rfloor > J$,  then  
  $$\M_\mu(I_W) \leq C\mu(I_w)  \mu(I_{ \sigma^{j-\lfloor\eta_{i,k} j  \rf} w}  )\leq C\mu(I_W) 2^{(j-\lfloor j \eta_{i,k}\rfloor))({H_i}(\eta_{i,k})-\ep)},$$ where \eqref{discret}  has been used.
\end{itemize}

In the second case, some information is lost between the generations $J$ and $\lf\eta'j\rf$. We deduce from these observations and the quasi-Bernoulli property of $\mu$  that
\begin{align*}
\sum_{W \in \Sigma_J} \M_\mu (I_W)^q &\le \sum_{i\in\{\ell,r\}}  \sum_{\substack{k=1 \\ \lf \eta_{i,k}j\rf\leq J} }^{N_i} \sum_{w\in \Tt(j,  \eta_{i,k}, \ep)}   \mu(I_w)^q+  \sum_{i\in\{\ell,r\}} \sum_{\substack{k=1 \\ \lf \eta_{i,k}j\rf >J}}^{N_i} \sum_{w\in \Tt(j,  \eta_{i,k},  \ep)}   \mu(I_w)^q \\
&\le \sum_{i\in\{\ell,r\}} \ \sum_{k=1}^{N_i}  3^d\sum_{J\le j \le J/\eta_{i,k}} \ \sum_{ u \in\Sigma_{\lfloor j \eta_{i,k}\rfloor}}C^{q}\mu(I_{u} )^q 2^{-q(j-\lfloor j \eta_{i,k}\rfloor)({H_i}(\eta_{i,k})-
\ep)} \\
&+\sum_{i\in\{\ell,r\}} \ \sum_{k=1}^{N_i} 3^d \sum_{W\in\Sigma_J} \ \sum_{j >J/\eta_{i,k}}C^q \mu(I_W)^q  2^{-q(j -\lfloor j \eta_{i,k}\rfloor)({H_i}(\eta_{i,k})-\ep)}.
\end{align*}
Recalling \eqref{upperbound} and the fact that $\widetilde H_i (\eta') = H_i(\eta') (\eta'^{-1}-1)$ for every $\eta'$, the first term in the last sum is bounded from above by
\begin{align*}
&  3^dC^{q} C_q\sum_{i\in\{\ell,r\}}\ \sum_{k=1}^{N_i} \ \sum_{J \le j \le J/\eta_{i,k}} 2^{-\lfloor j \eta_{i,k}\rfloor\tau_\mu(q)} 2^{-q(j -\lfloor j \eta_{i,k}\rfloor)({H_i}(\eta_{i,k})-\ep)} \\
&\le C'_q \sum_{i\in\{\ell,r\}} \ \sum_{k=1}^{N_i} 2^{qJ \ep/\eta_{i,k}}\sum_{J \le j \le J/\eta'_{i,k}} 2^{- j \eta_{i,k}(\tau_\mu(q)+q \Ht_i(\eta_{i,k}))}
\end{align*}
and the second by
 \begin{align*}
& 3^dC^{q}C_q\sum_{i\in\{\ell,r\}} \ \sum_{k=1}^{N_i}  2^{-J\tau_\mu(q)}\sum_{j> J/\eta_{i,k}}   2^{-q(j -\lfloor j \eta_{i,k}\rfloor)({H_i}(\eta_{i,k})-\ep)}\\
&\leq C'_q\sum_{i\in\{\ell,r\}} \ \sum_{k=1}^{N_i}  2^{-J (\tau_\mu(q)+q\Ht_i(\eta_{i,k}))} 2^{q J({\eta_{i,k}}^{-1}-1)\ep}
\end{align*}
for some other constant $C'_q$. 
Since $q\ge 0$, $\tau_\mu(q)+q\Ht_i(\eta_{i,k})$ is bounded from below for all $p \in\{1,...,\max( N_\ell,N_r)\}$  by $\tau_\mu(q)+q{\Ht_\ell}(\widetilde \eta)$, which is negative. Consequently, 
$$\sum_{J\le j \le J/\eta_{i,k}} 2^{- j \eta_{i,k}(\tau_\mu(q)+q
Ht_i(\eta_{i,k}))}=O(2^{-J (\tau_\mu(q)+q{\Ht_\ell}(\widetilde\eta))}).$$
In addition, one always has  ${\eta_{i,k}} \geq  \eta_i$, hence 
$$2^{- J (\tau_\mu(q)+q\Ht_i(\eta_{i,k}))} 2^{qJ ({\eta_{i,k}}^{-1}-1)\ep}\le 2^{-J (\tau_\mu(q)+q\Ht(\widetilde\eta))} 2^{qJ({\eta_i}^{-1}-1)\ep}.$$
Putting everything together we get for some $C''_q>0$ 
 \begin{align*}
& 
\sum_{W \in \Sigma_J} \M_\mu (I_W)^q\\
&\leq  C''_q \Big (N_\ell (2^{qJ\ep/\eta_\ell} + 2^{q J ({\eta_\ell}^{-1}-1)\ep})+ N_r (2^{q J \ep/\eta_r} + 2^{qJ ({\eta_r}^{-1}-1)\ep})\Big )2^{-J (\tau_\mu(q)+q{\Ht_\ell}(\widetilde\eta))}.
\end{align*}
This yields $\tau_{\M_\mu}(q)\ge \tau_\mu(q)+q{\Ht_\ell}(\widetilde\eta)+O(\ep)$, and letting $\ep$ tend to 0 gives the desired conclusion. 

\mk 

$\bullet$ Now we deal with the   case where at least one of parameters $\etal$ and $\etar$ equals zero.  

\sk

According to the value of $\eta_i$, we construct a subset  $ \Sigma^{(i)}_J$  of words  of length $J$ having specific properties:

\mk

{\bf First case:  $\eta_i>0$:}  set $\Sigma^{(i)}_J=\emptyset$. 

\mk
 
{{\bf Second case:  $\eta_i=0$:}   in this case, ${D_\mu}({H_i}(\eta_i)) = d(1-\eta)$. Heuristically, $ \Sigma^{(i)}_J$ contains those words  $W$  such that $\Mm(I_W ) = \mu(I_w)$ for some surviving vertex   $w\in \mathcal{S}_j(\eta)$ having an ``extreme'' behavior, i.e.  $\mu(I_w)\sim 2^{-j  H_i(\eta_i)}$. We proceed as follows:}

 {At first, let $K\ge H_{\max}$ be as in Proposition~\ref{fm}(6). For $\eta'_i\in [\eta_i,\eta]$ close to $\eta_i$, we denote by $\widehat{\eta_i}$ the unique real number in $[\eta_i,\eta]$ such that $H_i(\widehat\eta_i)= H_i(\eta_i')+K\eta'_i$ if $i=\ell$ and $H_i(\widehat\eta')= H_i(\eta_i')-K\eta'_i$ if $i=r$. Notice that $\widehat\eta_i>\eta_i'$. }

 {Now, fix $\varepsilon=q H_i(\eta_i)/4$, and choose $\eta'_i$ small enough so that $(1-\eta'_i) ({H_i}(\widehat\eta_i)-K\widehat\eta_i)>H_i(\eta_i)/2$ if $i=\ell$ and ${H_i}(\eta'_i)-2K\widehat\eta_i>H_i(\eta_i)/2$ if $i=r$, and 
$$
\begin{cases}
H_i(\eta'_i)+2K\widehat \eta_i<H_s \text{ and } D_\mu(H_i(\eta'_i)+2K\widehat \eta_i)\le D_\mu({H_i}(\eta_i)+\varepsilon/2& \text{if }i=\ell,\\  
H_i( \eta'_i)-2K\widehat \eta_i>H_s\text{ and } D_\mu(H_i(\eta'_i)-2K\widehat \eta_i)\le D_\mu({H_i}(\eta_i))+\varepsilon/2& \text{if }i=r
\end{cases}.
$$
 By item (5) of Proposition \ref{fm}, there exists  an integer $J_i$ such that  for $j \ge J_i$,
\begin{equation}\label{eq24sss}
\begin{split} 
&\# \mathcal{E}_\mu(j , [0,H_\ell(\eta'_i)+2K\widehat\eta_i])  \le 2^{j  (D_\mu(H_i(\eta'_i)+2K\eta'_i)+\varepsilon/2)}\le 2^{j  (D_\mu(H_i(\eta_i))+\varepsilon)} \text{ if }i=\ell,\\  
&\# \mathcal{E}_\mu(j , [H_i(\eta'_i)-2K\widehat\eta_i,+\infty))  \le 2^{j  (D_\mu(H_r(\eta'_i)+2K\eta'_i)+\varepsilon/2)}\le 2^{j  (D_\mu(H_i(\eta_i))+\varepsilon)} \text{ if }i=r
\end{split}.
\end{equation}}

\sk

 {We can also choose $J_i$ such that $\ep^2_{{j}}\le K\eta'_i/2\le  K\widehat \eta_i/2$ for $j\ge J_i$, where$( \ep^2_{{j}})_{j\ge 1}$ is the sequence introduced in Proposition~\ref{p2}. }

 {For  $J\ge J_i$,  we take  $\Sigma^{(i)}_J$  as the set of those words $W\in \Sigma_J$ such that   $\M_\mu(I_W) =\mu(I_w)$, where   $w \in  \mathcal{S}_j(\eta,W) \cap \mathcal{T}_{\mu,i }(j, \eta', \ep^2_j,W)   $     for some $\eta'$ satisfying
 $$
 \begin{cases}
 {H_\ell}(\eta')+\ep^2_{{j}}\le H_\ell(\eta'_\ell)+K\eta'_\ell&\text{if }i=\ell,\\
 {H_r}(\eta')-\ep^2_{{j}}\ge H_r(\eta'_r)-K\eta'_r&\text{if }i=r
 \end{cases}.
$$
 In particular, we have $\eta'\le \widehat\eta_i$. 
 The words    $W\in \Sigma^{(i)}_{J}$ are the ones that may cause problems when compared to the case where $\eta_\ell,\eta_r>0$. The other words $W$ are such that $\M_\mu(I_W)$ is reached at some $w$ associated with $\eta' $ satisfying ${H_\ell}(\eta')\in  [{{H_\ell}}(\eta'_\ell)+K\eta'_\ell/2,{H_r}(\eta'_r)-K\eta'_r/2]$, i.e. $\eta'$ stays bounded away from 0.} 
 
 \sk
 
  {When $J\ge J_i$ and  $W\in \Sigma^{(i)}_{J}$, for the associated word $w \in  \mathcal{S}_j(\eta,W) \cap \mathcal{T}_{\mu,i }(j, \eta', \ep^2_j,W)   $  (according to the previous notations), we have, using   \eqref{quasib} and the definition of $K$:
\begin{eqnarray*}
 C^{-1} 2^{-\lfloor j\eta'\rfloor K}\mu(I_{\sigma^{\lfloor \eta' {j}\rfloor}w}) \le \M_\mu(I_W) =   \mu(I_w) \le C  \mu(I_{\sigma^{\lfloor \eta' {j}\rfloor}w}),
   \end{eqnarray*}
   which yields, due to the property of $(W,w)$ and the fact that $\eta '\le\widehat \eta_i$:
\begin{eqnarray*}
 C^{-1} 2^{- j\widehat\eta_iK} 2^{-{j}(H_i(\eta')+\varepsilon_j^2) }\le \M_\mu(I_W) =   \mu(I_w)\le 
   C   2^{-({j}-\lfloor \widehat\eta_i j\rfloor)({H_i}(\eta')-\ep^{2}_j)}.
   \end{eqnarray*}  
 This yields, for $J$ large enough,
 $$
 \begin{cases}
2^{- j(H_i(\eta'_i)+2K\widehat \eta_i)}\le \M_\mu(I_W) =   \mu(I_w)\le 
 2^{-({j}(1-\eta'_i) ({H_i}(\widehat\eta_i)-K\widehat\eta_i)}&\text{ if }i=\ell,\\
   \M_\mu(I_W) =   \mu(I_w)\le 
2^{-{j}({H_i}(\eta'_i)-2K\widehat\eta_i)}&\text{if }i=r
\end{cases}
.$$
}

 {Hence, each such word $W$ is associated with one surviving word $w\in  \mathcal{E}_\mu(j , [0, H_i(\eta'_i) +2K\widehat\eta_i])   $, for some $ j\geq J$ if $i=\ell$, and one surviving word $w\in  \mathcal{E}_\mu(j , [H_i(\eta'_i) -2K\widehat\eta_i,+\infty))   $, for some $ j\geq J$ if $i=r$.}   

 {Then, if $i=\ell$, writing $j=J+k$ one gets :
\begin{align*}
\sum_{ W \in \Sigma^{(i)}_J} \M_\mu (I_W)^q & \leq \sum_{k=0}^{+\infty}  \ \sum_{w\in  \mathcal{E}_\mu(J+k,  [0, H_i(\eta'_i) +2K\widehat\eta_i])  } p_w   C^q   2^{-(J+k)q(1-\eta'_i) ({H_i}(\widehat\eta_i)-K\widehat\eta_i)} .
\end{align*}
Taking expectations and recalling \eqref{eq24sss}, one gets 
\begin{eqnarray*}
 \mathbb E\Big (\sum_{W\in \Sigma^{(i)}_{J}}\M_\mu (I_W)^q\Big ) \!
&\le& \! 3^d C^q\sum_{k\ge 0} 2^{-(J+k)d(1-\eta)} 2^{(J+k) (d(1-\eta)+ \ep)} 2^{-q(J+k) (1-\eta'_i) ({H_i}(\widehat\eta_i)-K\widehat\eta_i)}\\
&=&3^d C^q\sum_{k\ge 0} 2^{(J+k)(\ep - q (1-\eta'_i) ({H_i}(\widehat\eta_i)-K\widehat\eta_i))}.
\end{eqnarray*}
}

{Now by our choice for $\varepsilon$ and $\eta'_i$, we have  $\ep - q (1-\eta'_i) ({H_i}(\widehat\eta_i)-K\widehat\eta_i)\le -\ep$. We deduce that 
\begin{eqnarray*}
 \mathbb E\Big (\sum_{W\in \Sigma^{(i)}_{J}}\M_\mu (I_W)^q\Big )  \leq C_{q,\ep} 2^{-J\ep} 
\end{eqnarray*}
for some constant $C_{q,\ep}>0$. Finally,  applying the Borel-Cantelli lemma, we deduce that with probability 1, for $J$ large enough we have 
\begin{equation}
\label{taufinal}
\sum_{W\in \Sigma^{(i)}_{J}}  \M_\mu  (I_W)^q\le 1.
\end{equation}
Observe that \eqref{taufinal} holds true even if $\eta_i>0$ (in which case $\Sigma^{(i)}_{J}$ is empty). If $i=r$, similar computations yield
$$
\mathbb E\Big (\sum_{W\in \Sigma^{(i)}_{J}}\M_\mu (I_W)^q\Big )\le 3^d C^q\sum_{k\ge 0} 2^{(J+k)(\ep - q({H_i}(\eta'_i)-2K\widehat\eta_i))},
$$
with  a similar conclusion. }

\mk

Finally,  the same estimates as when both $\eta_\ell$ and $\eta_r$ are strictly positive yield  
$$\liminf_{J\to + \infty} \frac{-1}{J}\log_2 \sum_{W\in \Sigma_J\setminus ( \Sigma^{(\ell)}_J\cup \Sigma^{(r)}_J)} \M_\mu (I_W)^q\ge \tau_\mu(q)+q{\Ht_\ell}(\widetilde\eta) = \tilde\tau(q).$$
Since $\tilde\tau(q)<0$ and    \eqref{taufinal} holds for $J$ large enough, we conclude that  $\tau_{\M_\mu}(q)\ge \tilde\tau(q)$.

\subsubsection{  When $q<0$:} 

Applying Proposition~\ref{p'2}   with $\eta'=\widetilde\eta$, there exists a positive sequence $(\ep^3_j)_{j\ge 1}$ converging to 0 such that with probability 1, for $j$ large enough, for all $W\in \Sigma_{\lfloor \weta j\rfloor}$, there exists $w\in \mathcal{S}_j(\eta,W)$ such that  the $\weta$-tail of $w$ satisfies 
$$
2^{-(j-\lfloor \weta j\rfloor)({{H_\ell}}(\weta)+\ep^3_j)}\le \mu(I_{\sigma^{\lfloor\weta  j\rfloor} w}).$$
The quasi-Bernoulli property implies that   $\M_\mu (I_{Ww})\ge C^{-1} \mu(I_{W}) 2^{-j(1-\widetilde \eta)({{H_\ell}}(\widetilde \eta)+\ep_j)}$, which for $q<0$ yields 
\begin{eqnarray*}
\sum_{W\in \Sigma_{\lfloor \weta j\rfloor}} \!\!  \M_\mu  (I_{W})^q & \le  &  C^q 2^{-jq(1-\widetilde \eta)({{H_\ell}}(\widetilde \eta)+\ep_j)}  \!\! \sum_{ W \in \Sigma_{\lfloor \weta j\rfloor}}\mu(I_{W})^q \\
& \le& C^{q+1}  2^{-\lf \weta j\rf q  ({{\Ht_\ell}}(\widetilde \eta)+\ep_j/\weta )} \!\!  \sum_{ W \in \Sigma_{\lfloor \weta j\rfloor}}\mu(I_{W})^q   .
\end{eqnarray*}
One concludes that  $ \tau_{\M_\mu }(q)\ge  \tau_\mu(q)+{\Ht_\ell} (\widetilde\eta)q =\tilde\tau(q)$.

\subsubsection{ When ${q_{\eta_\ell}}<+\infty$ and $q>{q_{\eta_\ell}}$.}

 Recall that this implies $\eta_\ell =0$. 
We have already shown that $\tau_{\M_\mu} (q)\ge\tau_\mu(q)+d(1-\eta)$ when $q\in [{q_{\widetilde \eta}},{q_{\eta_\ell}}]$. 

The tangent to the graph of  $ q \mapsto \tau_{\M_\mu}(q)$ at $({q_{\eta_\ell}},\tau_{\M_\mu}({q_{\eta_\ell}}))$ is the affine line passing through $(0,0)$, whose slope  is $ \tau_{\M_\mu}'(q_{\eta_\ell}) =  {H_\ell}(0)$. Consequently, the concavity  of $\tau_{\M_\mu}$  implies  that    $\tau_{\M_\mu}(q)\le q {{H_\ell}}(0)$ for all $q\ge {q_{\eta_\ell}}$. On the other hand, if $q\ge {q_{\eta_\ell}}$, for all integers $J\ge 1$ we have 
$$
\sum_{W\in\Sigma_J}\M_\mu(I_W)^q\le\Big( \sum_{W\in\Sigma_J} \M_\mu(I_W)^{{q_{\eta_\ell}}}\Big )^{q/{q_{\eta_\ell}}},
$$
from which it follows that $ \tau_{\M_\mu}(q) \ge \frac{q}{{q_{\eta_\ell}}} \tau_{\M_\mu}({q_{\eta_\ell}})=q{{H_\ell}}(0)$.

\subsection{Lower bound for the lower large deviations spectrum $\underline{f}_{\Mm}(H)$}
\label{sec93}

Let us check that \eqref{LD} holds. It is enough to deal with a dense countable subset of the support  $ [{{H_\ell}}(\eta_\ell),{H_{\max}}+{\Ht_\ell}(\widetilde\eta)]$ of $  {\tau_{\Mm} }$.

\mk$\bullet$ 
Suppose first that $H \in [{{H_\ell}}(\etal),{{H_\ell}}(\widetilde\eta)]$. Recall the definition \eqref{defetaH} of $\eta_H$: $H={H_\ell}(\eta_H)$.

By item (4) of Proposition \ref{fm},   for every $\ep>0$, there exists $\beta(\ep)>0$   such that  when  $j$ becomes large, 
$$ \#\mathcal{E}_\mu( \lfloor \eta_H  j  \rfloor, [0, {{H_\ell}}(\eta_H)+\ep]) \geq  2^{\lfloor \eta_H  j \rfloor ({D_\mu}({{H_\ell}}(\eta_H))- \beta(\ep))}.$$

One also knows that $\beta(\ep)$ can be taken so that $\beta(\ep) \to 0$ when $\ep\to 0$.

In addition, by Proposition \ref{p'2}, there is a positive sequence $(\ep^3_j)_{j\ge 1}$ converging to~0 such that, with probability 1, for $j$ large enough, each cube $I_W$  (with $W\in   \mathcal{E}_\mu( \lfloor \eta_H  j  \rfloor, [0, {{H_\ell}}(\eta_H)+\ep])$) contains a smaller cube $I_{w}$, with $w\in\mathcal{S}_j(\eta,W)  \cap \Tl ( j,\eta_H,  \ep_j^3)$.

By the quasi-Bernoulli property  of $\mu$, 
$$\M_\mu(I_{w}) \geq \mu(I_{w} ) \geq C^{-1} 2^{-\lfloor \eta_H  j \rfloor ({{H_\ell}}(\eta_H)+\ep)} 2^{-(j-\lfloor \eta_H  j \rfloor) ({{H_\ell}}(\eta')+\ep^3_j)} =   2^{-  j  ({{H_\ell}}(\eta_H)+2\ep)} $$
when $j$ becomes large. Thus, 
\begin{eqnarray*}
\liminf_{j\to + \infty}\frac{1}{j}\log_2\#  \mathcal{E}_{\M_\mu} (j,[Õ0, {H_\ell}(\eta_H)+2\ep])  \ge \eta_H \Big({D_\mu}({{H_\ell}}(\eta_H) ) -\beta(\ep)\Big), 
\end{eqnarray*}
Since by construction $H ={H_\ell} (\eta_H)$ and $ \eta_H {D_\mu}({{H_\ell}}(\eta_H))={D_\mu}(H) - d(1-\eta)$, letting $\ep$ go to zero gives
$$
\lim_{\ep\to 0^+} \liminf_{j\to +\infty}\frac{1}{j}\log_2\#  \mathcal{E}_{\M_\mu} (j,[0, H +2\ep]) \geq   {D_\mu}(H) - d(1-\eta) = \tau_{\Mm} ^*(H).
$$
We conclude that 
$\underline{f}_{\M_\mu} (H) \geq   {\taut }^*(H),$
for otherwise there would exist $H'<H$ such that 
$$
\limsup_{j\to+ \infty}\frac{1}{j}\log_2\#\mathcal{E}_{\M_\mu} (j,[Õ0, H']) \ge \tau_{\Mm} ^*(H)> \tau_{\Mm} ^*(H'),
$$
which contradicts the fact that for  all $H'\le H_s+{\Ht_\ell}(\weta)=\tau_{\M_\mu}'(0)$, Proposition~\ref{chernov} yields
$$
 \limsup_{j\to+ \infty}\frac{1}{j}\log_2\#  \mathcal{E}_{\M_\mu} (j,[Õ0, H']) \le \tau_{\M_\mu}^*(H') .
$$

\mk$\bullet$  For $H\in [{{H_\ell}}(\widetilde\eta), {{H_\ell}}(\widetilde\eta)+{\Ht_\ell}(\etal)]$ we use the same idea: There exist a positive sequence $(\beta_j)_{j\ge 1}$ converging to 0 such that, for  $j$ is large enough, at generation $\lfloor j\widetilde \eta\rfloor$ there are at least $2^{\lfloor j\widetilde \eta\rfloor ({D_\mu}({{H_\ell}}(\weta) )-\ep_j))}$ elements $W$  in $ \mathcal{E}_\mu(\lf j\weta\rf, [0, {H_\ell}(\weta)+\beta_j])$.

In addition, by Proposition~\ref{p'2}, with probability 1, for $j$ large enough, each of these $I_W$ contains a smaller cube $I_{w}$, with $w\in\mathcal{S}_j(\eta,W) \cap  \Tl (j,\weta,  \ep_j^3)$. 

Then, let $w'$ be the word of generation $j'=\lfloor j {{H_\ell}}(\widetilde \eta)/H\rfloor$ such that $I_{w}\subset I_{w'}\subset I_W$. We have $-\log _2 \M_\mu (I_{w'}) \ge -\log \M_\mu (I_W) \sim  j {{H_\ell}} (\eta')\sim j' H$.  It follows that  for any $\ep>0$,  
$$
\liminf_{j'\to+\infty}\frac{1}{j'}\log_2\# \mathcal{E}_{\M_\mu}(j', [0,H+\ep])  \ge \frac{\widetilde \eta {D_\mu}({{H_\ell}}(\widetilde \eta)}{{{H_\ell}}(\widetilde\eta)} H=\widetilde\tau^*(H) =  \tau_{\Mm} ^*(H).
$$
We conclude as in the previous case.

\mk$\bullet$ For $H\in [ {{H_\ell}}(\widetilde\eta)+{\Ht_\ell}(\widetilde\eta),{H_{\max}}+{\Ht_\ell}(\widetilde\eta)]$, we can use Section~\ref{sharplow1} which directly yields the conclusion.


\section{Dimension of the level sets $E_{\M_\mu}(H)$ and $\overline E_{\M_\mu}(H)$ for $H\ge H_\ell(\weta)+\widetilde H_\ell(\weta)$}\label{limlimd}

Due to Remark~\ref{remmin} we only have to prove the sharp upper bound for $\dim \overline E_{\M_\mu}(H)$.

 Proposition~\ref{majexp} yields ${\overline \dim_\locloc}(\mu,x)\ge  {\overline \dim_\locloc}(\M_\mu,x)  -{\Ht_\ell}(\weta)$ for all $x\in [0,1]^d$. Hence for all $H\ge 0$ we have  $\overline E_{\M_\mu}(H) \subset \overline {E}^{\ge }_{\mu}(H-{\Ht_\ell}(\weta))$. By item (3) of Proposition~\ref{fm}, we deduce that $\dim\, \overline E_{\M_\mu}(H) \le {D_\mu}(H-{\Ht_\ell}(\weta))$ for $H\ge H_s+{\Ht_\ell}(\weta)$. 

\mk

{It remains us to treat the case  $H\in [H_\ell(\weta)+\widetilde H_\ell(\weta),H_s+\widetilde H_\ell(\weta))$. We already know by item (2) of Proposition~\ref{fm} that $\dim \Big( \overline E_{\M_\mu}(H)\cap \underline E^{\le}_\mu(H-\widetilde H_\ell(\weta)) \Big)\le D_\mu(H-\widetilde H_\ell(\weta))$. }

In order to complete the proof, it is enough to prove that $\dim \Big( \overline E_{\M_\mu}(H)\cap \underline E^{\ge}_\mu(H-\widetilde H_\ell(\weta)) \Big)\le D_\mu(H-\widetilde H_\ell(\weta))$. For this, consider
 $x\in \overline E_{\M_\mu}(H)\cap  \underline E^{\ge}_\mu(H-\widetilde H_\ell(\weta))$. Based on the discussion achieved in the proof of Proposition~\ref{propmajs}, we know that 
$x\in \widetilde F_\mu(H)$, where 
$$
\widetilde  F_{\mu} (H)=\bigcup_{i\in\{\ell,r\}}\bigcap_{\ep\in(0,1)}  \ \bigcap_{k\ge 1} \ \bigcup_{\substack { (\alpha,\eta',\delta) \in \mathcal{P}_i(H): \\
 \delta \in [1,1/\eta'],\ 
 \frac{\alpha+{\Ht_i}(\eta')}{\delta}\le H+\ep } }\limsup_{j\to+ \infty} F_{\mu,\ell}(j,\alpha,\eta',\delta -\ep,k),
$$
and $\mathcal{P}_i(H)$ is a countable set of parameters $(\alpha,\eta',\delta) $ dense in $[H-\widetilde H_\ell(\weta),H_{\max}]\times  (\eta_i,\eta] \times [1,+\infty)$. Also, if $\alpha\ge H-\widetilde H_\ell(\weta)$ and $\frac{\alpha+{\Ht_i}(\eta')}{\delta}\le H+\ep $, then  
$$
\frac{D_\mu (\alpha)}{\delta}\le (H+\ep) \frac{D_\mu (\alpha)}{\alpha+{\Ht_i}(\eta')}\le (H+\ep)\frac{D_\mu (\alpha)}{\alpha+{\Ht_\ell}(\weta)}\le (H+\ep)\frac{D_\mu (H-{\Ht_\ell}(\weta))}{H},
$$
since $\alpha\mapsto \frac{D_\mu (\alpha)}{\alpha+{\Ht_\ell}(\weta)}$ is decreasing over $[H_\ell(\weta),H_{\max}]$ and $H\ge H_\ell(\weta)+\widetilde H_\ell(\weta)$. Then, using the same estimates as in the proof of Proposition~\ref{ubip}, we get $\dim \widetilde  F_{\mu} (H)\le D_\mu (H-{\Ht_\ell}(\weta))$, hence $\dim  \Big( \overline E_{\M_\mu}(H)\cap  \underline E^{\ge}_\mu(H-\widetilde H_\ell(\weta)) \Big)\le D_\mu (H-{\Ht_\ell}(\weta))$. Hence the conclusion.
 
\begin{remark}
We could  conclude for all $H\in  [H_{\min}+\widetilde H_\ell(\weta),H_{\max}+\widetilde H_\ell(\weta)]$ if  we were able to prove the property  $ {\overline \dim_\locloc}(\M_\mu,x)  \ge {\underline \dim_\locloc}(\mu,x)+{\Ht_\ell}(\weta)$ for all $x\in [0,1]^d$ to hold. 
\end{remark}

\section{Case of a homogeneous Gibbs measure}
\label{sechomo}

We rapidly explain the case of a homogeneous capacity that we denote $\lambda$. We assume without loss of generality that for some $\beta >0$, for every finite word $w\in \Sigma^*$,  $\lambda(I_w) \sim 2^{-\beta |w|}$.

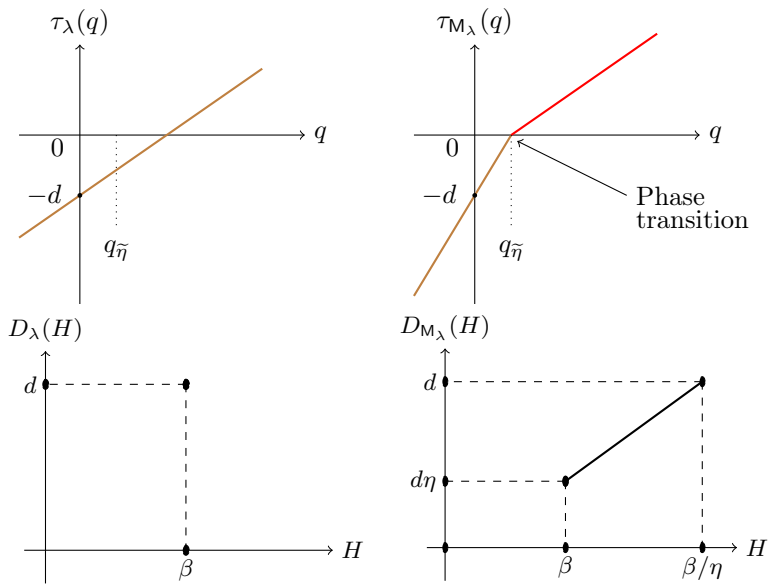
\begin{figure}\begin{tikzpicture}[xscale=0.8,yscale=0.8]
{\small
\draw [->] (0,-2.8) -- (0,1.5) [radius=0.006] node [above] {$\tau_\lambda(q)$};
\draw [->] (-1.,0) -- (3.7,0) node [right] {$q$};
 \draw [thick, domain=-1 :3, color=brown]  plot ({\x},  {0.7*\x-1}); 
 \draw [dotted] (0.6,0)-- (0.6,-1.5) [fill] (0.6,-1.6) node [below] {${q_{\widetilde \eta}}$};  
\draw [fill] (-0.1,-0.20)   node [left] {$0$}; 
\draw [fill] (-0,-1) circle [radius=0.03]  node [left] {$-d$ \ }; }
\end{tikzpicture}    \ \hspace{5mm} \ 
\begin{tikzpicture}[xscale=0.8 ,yscale=0.8 ]{\small
\draw [->] (0,-2.8) -- (0,1.5) [radius=0.006] node [above]  {$ \tau_{\M_\lambda}(q)$};
\draw [->] (-1.,0) -- (3.7,0) node [right] {$q$};
  \draw [thick, domain=-1 :0.6, color=brown]  plot ({\x},  {1.666*\x-1}); 
  \draw [thick, domain=0.6:3, color=red]  plot ({\x},  {0.7*\x-0.42}); 
 \draw [<-] (0.7,-0.1) -- (2.5,-1) node [right] {Phase};
 \draw[fill] (2.5,-1.4)   node   [right] {transition};
 \draw [dotted] (0.6,0)-- (0.6,-1.5) [fill] (0.6,-1.6) node [below] {${q_{\widetilde \eta}}$};  
\draw [fill] (-0.1,-0.20)   node [left] {$0$}; 
\draw [fill] (-0,-1) circle [radius=0.03]  node [left] {$-d$ \ }; }
\end{tikzpicture} 

    \begin{tikzpicture}[xscale=1.4,yscale=2.2]{\footnotesize
\draw [->] (0,-0.2) -- (0,1.2) [radius=0.006] node [above] {$D_\lambda(H)$ };
\draw [->] (-0.2,0) -- (2.7,0) node [right] {$H$};
  \draw[fill] (1.32,1) circle [radius=0.03] [dashed] (1.32,1) -- (0,1) [fill] circle [radius=0.03]    node [left]  {  ${d} $};  
   \draw[fill] (1.32,1) circle [radius=0.03] [dashed] (1.32,1) -- (1.32,-0.0) [fill] circle [radius=0.03]  [dashed] (1.32,0.0) -- (1.32,-0.0)  node [below]  {  ${\beta} $};  }
\end{tikzpicture}   
\          \begin{tikzpicture}[xscale=1.2,yscale=2.2]{\footnotesize
   \draw [->] (0,-0.2) -- (0,1.2) [radius=0.006] node [above] {$D_{\M_\lambda}(H)$ };
\draw [->] (-0.2,0) -- (3.2,0) node [right] {$H$};
  \draw[fill] (2.82,1) circle [radius=0.03] [dashed] (2.82,1) -- (0,1) [fill] circle [radius=0.03]    node [left]  {  ${d} $};  
   \draw[fill] (0,0.4) circle [radius=0.03]  node [left] {$d\eta$} [dashed] (0,0.4) -- (1.32,0.4)[fill] (1.32,0.4) circle [radius=0.03] [dashed] (1.32,0.4) -- (1.32,-0.0) [fill] circle [radius=0.03]  [dashed] (1.32,0.0) -- (1.32,-0.0)  node [below]  {  ${\beta} $};  
  \draw[fill] (1.32,0.4) circle [radius=0.03];
  \draw [fill,thick] (1.32,0.4) -- (2.82,1) ;
  \draw[fill] circle [radius=0.03]  [dashed] (2.82,1) -- (2.82,-0.0)  [fill] (2.82,0) circle [radius=0.03] node [below]  {${\beta}/\eta $};  }   
\end{tikzpicture} 
 \caption{{\bf Left:} Free energy function  (top) and associated multifractal spectrum (bottom) of a homogeneous capacity $\lambda$. {\bf Right:} The almost sure  free energy (top) and the multifractal spectrum (bottom)  of   $ {\M_\lambda}$.}

\end{figure}

 In this situation, $H_{\min}= H_s = H_{\max} = \beta$, so $\etal=\etar=\weta= \eta$, and ${H_\ell}(\etal) = H_r(\eta_r) = \beta$. One also has ${\Ht_\ell}(\weta) = \beta (1/\eta-1)$.
Moreover, 
$q_{\eta_\ell}=+\infty$. 

The free energy function $\tau_\lambda(q) =  \beta q -d$  is linear, and  $q_\weta  $ is the solution to $\tau_\lambda(q) = -d(1-\eta) $, i.e. $q_{\weta } = d \eta/\beta$. 

The proof  follows exactly the same lines as in the previous sections, except that most of the arguments are trivial. Indeed, all the survivors at a given generation $j$  satisfy $\lambda(I_w) \sim 2^{-j\beta}$ (there is no   dependence of the value $\lambda (I_w)$ on the location of $w$). The sets $\mathcal{R}_\lambda$, $\mathcal{T}_{\lambda}$ are similarly defined, but are also trivial.

The obtained energy function is
$$
 \tau_{\M_\lambda}(q)=
\begin{cases}
\ \tau_\lambda(q)+ \beta (1/\eta-1) q  \, = \,  q \beta/\eta   - d&\text{if } q\le d {\eta/\beta},\\
\ \tau_\lambda(q)+d(1-\eta)  \ \ \ \ \  \, =  \, q \beta -d(1-\eta) &\text{if } q > d {\eta/\beta},
\end{cases}
$$
and the associated multifractal spectrum is
$$
D_{\M_\lambda}(H)= 
\begin{cases}\sk
\  \displaystyle  \frac{d\eta}{\beta} H  &\text{if }   H\in [ \beta, \beta/\eta],\\\sk
 \ -\infty&\text{otherwise}.
\end{cases}
$$

\mk

Actually, this case  was already studied by Jaffard in the context of ``lacunary wavelet series'' and  multifractal analysis of functions~\cite{JLac}. More precisely,   Jaffard computes the singularity spectrum of   wavelet series whose wavelet coefficients are defined as follows: Fix  $(\psi_{j,k})_{j,k\in\mathbb Z}$, a wavelet basis of $L^2(\mathbb R)$ associated with a smooth mother wavelet and normalized so that all its elements have the same $L^\infty$ norm. Fix $\beta>0$, and for each $j\ge 1$  select uniformly and independently $ 2^{\lfloor \eta  j \rfloor}$ intervals  among the $2^{j}$ dyadic subintervals of $\zu$ of generation $j$. Then assign the coefficient $2^{-\beta j}$ to $\psi_{j,k}$ if $[k2^{-j}, (k+1)2^{-j}]$ has been selected; otherwise assign the coefficient 0. Though different, this sparse collection of coefficients is close to that  obtained  by sampling  the homogeneous capacity $\lambda$ as above  in the special situation where  $\lambda  (I_w)=2^{-\beta |w|}$ for all $w\in\Sigma^*$. It turns out that the multifractal  analysis of the resulting sparse wavelet series  is essentially reducible to that of $\M_\lambda$, which in this case follows from quite a direct application of homogeneous ubiquity theory \cite{Dodson,JLac}. Of course, Jaffard obtained the same multifractal spectrum, although he did not compute  the free energy function.


\end{document}